\documentclass[11pt]{article}
 
\usepackage{fullpage}
\usepackage{palatino} 
\usepackage{tabularx}
\usepackage{amsmath}

\newtheorem{theorem}{Theorem}[section]
\newtheorem{lemma}[theorem]{Lemma}

\newtheorem{corollary}[theorem]{Corollary}

\newtheorem{fact}[theorem]{Fact}

\newtheorem{protocol}[theorem]{Protocol}
\newtheorem{pointgame}[theorem]{Point Game}

\newtheorem{definition}[theorem]{Definition}

\newenvironment{proof}[1][Proof]{\begin{trivlist}
\item[\hskip \labelsep {\bfseries #1}]}{\end{trivlist}}

\newenvironment{remark}[1][Remark]{\begin{trivlist}
\item[\hskip \labelsep {\bfseries #1}]}{\end{trivlist}}

\usepackage{graphicx}
\usepackage{amssymb}
\usepackage{epstopdf}
\title{Quantum and classical coin-flipping protocols \\ based on bit-commitment and their point games}

\author{
Ashwin Nayak\thanks{
Department of Combinatorics and Optimization,
and Institute for Quantum Computing, University of Waterloo.
Address: 200 University Ave.\ W.,
Waterloo, ON, N2L 3G1, Canada. Email: {
\tt
ashwin.nayak@uwaterloo.ca}.}
\and
Jamie Sikora\footnote{{Some of the results in this paper were announced earlier in the second author's PhD thesis~\cite{SikoraPHD12}.}} \thanks{
Centre for Quantum Technologies, National University of Singapore, {and MajuLab, CNRS-UNS-NUS-NTU International Joint Research Unit, UMI 3654, Singapore}. Address: Block S15, 3 Science Drive 2, Singapore 117543. Email: {
\tt
cqtjwjs@nus.edu.sg}.}
\and
Levent Tun\c cel\thanks{
Department of Combinatorics and Optimization, University of Waterloo. Address: 200 University Ave.\ W., Waterloo, ON, N2L 3G1, Canada.
Email: {
\tt
ltuncel@uwaterloo.ca}.}
}

\date{April 17, 2015}

\newcommand{\SDP}{\mathrm{SDP}}

\newcommand{\T}{\top}
\newcommand{\A}{\mathrm{A}}
\newcommand{\B}{\mathrm{B}}
\newcommand{\rF}{\mathrm{F}}
\newcommand{\half}{\frac{1}{2}}
\newcommand{\quarter}{\frac{1}{4}}
\newcommand{\sqrtt}[1]{\sqrt{#1}\sqrt{#1}^{\T}}

\newcommand{\calA}{\mathcal{A}}

\newcommand{\calP}{\mathcal{P}}

\newcommand{\prob}{\textup{Prob}}

\newcommand{\tr}{\mathrm{Tr}}
\newcommand{\ket}[1]{| #1 \rangle}
\newcommand{\bra}[1]{\langle #1 |}

\newcommand{\ketbra}[2]{\ket{#1} \bra{#2}}
\newcommand{\kb}[1]{\ketbra{#1}{#1}}
\newcommand{\braket}[2]{\langle #1 | #2 \rangle}
\newcommand{\dsum}{\displaystyle\sum}
\newcommand{\inner}[2]{\langle #1, #2 \rangle}

\newcommand{\set}[1]{\left\{ #1 \right\}}

\newcommand{\Lor}{\mathrm{SOC}}
\newcommand{\RL}{\mathrm{RSOC}}
\newcommand{\eps}{\varepsilon}

\newcommand{\Diag}{\mathrm{Diag}}

\newcommand{\Null}{\mathrm{Null}}
\newcommand{\supp}{\mathrm{supp}}

\newcommand{\Prob}{\mathrm{Prob}} 
\newcommand{\abort}{\mathrm{abort}}

\newcommand{\diag}{\mathrm{diag}}

\newcommand{\nulll}{\mathrm{Null}}

\newcommand{\C}{\mathbb{C}}
\newcommand{\R}{\mathbb{R}}

\newcommand{\Pos}{\mathbb{S}_+}
\newcommand{\pos}{\mathbb{S}_+}

\newcommand{\Herm}{\mathbb{S}}

\newcommand{\forAll}{\textrm{for all }}
\newcommand{\bit}{\in \set{0,1}}

\newcommand{\qed}{$\quad \square$}

\newcommand{\norm}[1]{\left\| #1 \right\|}

\newcommand{\id}{\mathrm{I}}

\newcommand{\aand}{\quad \text{ and } \quad}

\newcommand{\trho}{\tilde{\rho}}

\newcommand{\BCCF}{\mathrm{BCCF}}

\newcommand{\zo}{\{ 0, 1 \}}
\newcommand{\eig}{\mathrm{eig}} 

\newcommand{\rA}{\mathrm{A}}
\newcommand{\rB}{\mathrm{B}}

\newcommand{\pg}[2]{\left[ #1, \phantom{\frac{.}{.}} \!\! #2 \right]}
\newcommand{\pgf}[4]{\left[ #1, \phantom{\frac{.}{.}} \!\! #2, #3, #4 \right]}

\usepackage{hyperref}

\begin{document}
\maketitle

\begin{abstract}
We focus on a {family} of quantum coin-flipping protocols based on quantum bit-commitment. 
We discuss how the semidefinite programming formulations of cheating strategies can be reduced to optimizing a linear combination of fidelity functions over a polytope. These turn out to be much simpler semidefinite programs which can be modelled using \emph{second-order cone programming {problems}}. 
We then use these simplifications to construct their point games as developed by Kitaev by {exploiting} the structure of optimal dual solutions.  
 
We also study a family of classical coin-flipping protocols based on classical bit-commitment. Cheating strategies for these classical protocols can be formulated as linear {programs} which are closely related to the semidefinite programs for the quantum version. In fact, we can construct point games for the classical protocols as well using the analysis for the quantum case. 

We discuss the philosophical connections between the classical and quantum protocols and their point games as viewed from optimization theory. In particular, we observe an analogy between a spectrum of physical theories (from classical to quantum) and a spectrum of convex optimization problems (from linear programming to semidefinite programming, through second-order cone programming).  In this analogy, classical systems correspond to linear programming {problems} and the level of quantum features in the system is correlated to the level of sophistication of the semidefinite programming models on the optimization side.

Concerning security analysis, we use the classical point games to prove that every classical protocol of this type allows exactly one of the parties to entirely determine the coin-flip. Using the intricate relationship between the semidefinite programming {based} quantum protocol analysis and the linear programming {based} classical protocol analysis, we show that only ``classical'' protocols can saturate Kitaev's lower bound for strong coin-flipping.  
Moreover, if the product of Alice and Bob's optimal cheating probabilities is $1/2$, then exactly one party can perfectly control the outcome of the protocol. This rules out quantum protocols of this type from {attaining} the optimal level of security. 
\end{abstract}

\newpage
\tableofcontents

\newpage
\section{Introduction}

{
Security levels of quantum coin-flipping protocols 
as well as
classical coin-flipping protocols can be modelled and analyzed via utilization of
convex optimization theory.  In particular,  
the cheating strategies
determining the security level of such quantum protocols
can be modelled by semidefinite programming problems.  In this paper, we deeply explore this connection by examining an algebraic construct known as \emph{point games}. These point games are constructed from feasible dual solutions and, in this  sense, are dual to the notion of protocols. We fully flesh out the details of these connections for a specific class of protocols and discuss how these connections extend to the classical version as well. We then discuss the philosophical ideas behind these connections and show some theoretical implications.}  

{Being able to interpret dual solutions to optimization problems has been very fruitful, even in the very special case of linear programming problems. Such interpretations typically lead to a deeper understanding of the behaviour of optimal solutions and better formulations of optimization problems modelling related phenomena.} 
 
\subsection{Quantum coin-flipping}

Coin-flipping is a classic cryptographic task introduced by
Blum~\cite{Blu81}. In this task, two remotely situated parties,
Alice and Bob, would like to agree on a uniformly random bit by
communicating with each other. The complication is that neither party
trusts the other. If Alice were to toss a coin and send the
outcome to Bob, Bob would have no means to verify whether this was a
uniformly random outcome. In particular, if Alice wishes to cheat, she
could send the outcome of her choice without any possibility of being caught
cheating. We are interested in a communication protocol that is \emph{designed
to protect\/} an honest party from being cheated.

More precisely, a ``strong coin-flipping protocol'' with
bias~$\epsilon$ is a two-party communication protocol in the style
of Yao~\cite{Yao79,Yao93}. In the protocol, the two players, Alice and Bob,
start with no inputs and compute a value~$c_\rA, c_\rB \in
\set{0,1}$, respectively, or declare that the other player is
cheating. If both players are honest, i.e., follow the protocol, then
they agree on the outcome of the protocol ($c_\rA = c_\rB$), and
the coin toss is fair ($\Pr(c_\rA = c_\rB = b) = 1/2$, for any~$b
\in \set{0,1}$). Moreover, if one of the players deviates arbitrarily
from the protocol in his or her local computation, i.e., is ``dishonest''
(and the other party is honest), then the probability of either outcome~$0$
or~$1$ is at most~$1/2 + \epsilon$. Other variants of coin-flipping have
also been studied in the literature. However, in the rest of the
article, by ``coin-flipping'' (without any modifiers) we mean
\emph{strong\/} coin flipping.

A straightforward game-theoretic argument proves that if the two
parties in a coin-flipping protocol
communicate classically and are computationally unbounded,
at least one party can cheat perfectly (with bias~$1/2$).
In other words, there is at least one party, say Bob, and at least one
outcome~$b \in \set{0,1}$ such that Bob can ensure outcome~$b$ with
probability~$1$ by choosing his messages in the protocol appropriately.
Consequently, classical coin-flipping protocols with bias~$\epsilon <
1/2$ are only possible under complexity-theoretic assumptions, and
when Alice and Bob have limited computational resources.

{The use of} quantum communication offers the possibility of ``unconditionally
secure'' cryptography, wherein the security of a protocol
rests solely on the validity of
quantum mechanics as a faithful description of nature. The first few
proposals for quantum information processing, namely the Wiesner quantum
money scheme~\cite{Wiesner83} and the Bennett-Brassard quantum key expansion
protocol~\cite{BB84} were motivated by precisely this idea. These
schemes were eventually shown to be unconditionally secure in
principle~\cite{M01,LC99,PS00,MVW12}. In light of these
results, several researchers have studied the possibility of
\emph{quantum\/} coin-flipping protocols, as a step towards studying
more general secure multi-party computations.

Lo and Chau~\cite{LC97} and Mayers~\cite{May97} were the first to
consider quantum protocols for coin-flipping without any computational
assumptions. They proved that no protocol with a finite number of rounds
could achieve~$0$ bias. Nonetheless, Aharonov, Ta-Shma, Vazirani, and
Yao~\cite{ATVY00} designed a simple, three-round quantum protocol that
achieved bias~$\approx 0.4143 < 1/2$. This is impossible classically, even
with an unbounded number of rounds. Ambainis~\cite{Amb01} designed
a protocol with bias~$1/4$ \emph{{\`a} la\/} Aharonov \emph{et al.\/},
and proved that it is optimal within a class (see also
Refs.~\cite{SR01,KN04} for a simpler version of the protocol and a
complete proof of security).  Shortly thereafter, Kitaev~\cite{Kit03}
proved that any strong coin-flipping protocol with a finite number of
rounds of communication has bias at least~$(\sqrt{2}-1)/2 \approx 0.207$ (see
Ref.~\cite{GW07} for an alternative proof). Kitaev's seminal work
uses semidefinite optimization in a central way.  This argument extends to
protocols with an unbounded number of rounds. This remained the state
of the art for several years, with inconclusive evidence in either
direction as to whether~$1/4 = 0.25$ or~$(\sqrt{2}-1)/2$
is optimal. In 2009, Chailloux and Kerenidis~\cite{CK09} settled this
question through an elegant protocol scheme that has
bias at most~$(\sqrt{2}-1)/2 + \delta$ for any~$\delta > 0$ of our
choice (building on~\cite{Moc07}, see below). We refer to this as
the CK protocol.

The CK protocol uses breakthrough work by Mochon~\cite{Moc07}, which
itself builds upon the ``point game'' framework proposed by Kitaev.
Mochon shows there are \emph{weak\/} coin-flipping protocols with
arbitrarily small bias. This work has since been simplified by experts on the topic; see e.g.~\cite{AharonovCGKM14}.) 
A weak coin-flipping protocol
is a variant of
coin-flipping in which each party favours a distinct outcome, say Alice
favours~$0$ and Bob favours~$1$. The requirement when they are honest is
the same as before. We say it has bias~$\epsilon$ if the following
condition holds. When Alice is dishonest and Bob honest, we only
require that Bob's outcome is~$0$ (Alice's favoured outcome) with
probability at most~$1/2+\epsilon$. A similar condition to protect
Alice holds, when she is honest and Bob is dishonest.
The weaker requirement of security against a dishonest player allows us
to circumvent the Kitaev lower bound. While Mochon's work pins down
the optimal bias for weak coin-flipping, it does this in a non-constructive
fashion: we only know of the \emph{existence\/} of protocols with
arbitrarily small bias, not of its \emph{explicit description\/}.
Moreover, the number of rounds tends to infinity as the bias decreases to $0$.
As a consequence, the CK protocol for strong coin-flipping is also
existential, and the number of rounds tends to infinity as the bias
decreases to $(\sqrt{2}-1)/2$.  It is perhaps
very surprising that no progress on finding better explicit protocols has
been made in over a decade.
 
\subsection{Our results}

To state our results, we introduce the following four quantities:
\begin{center}
\begin{tabularx}{\textwidth}{rX}
  $P_{\B,c}^*$: & The maximum probability with which a dishonest Bob can force an honest Alice to output $c \in \{0, 1\}$ by digressing from protocol. 
  \\
  $P_{\A,c}^*$: & The maximum probability with which a dishonest Alice can force an honest Bob to output $c \in \{0, 1\}$ by digressing from protocol. 
\end{tabularx} 
\end{center} 
We define a {family} of quantum coin-flipping protocols based on bit-commitment which we call $\BCCF$-protocols. These protocols are parameterized by four probability distributions $\alpha_0, \alpha_1$ defined on a finite set $A$ and $\beta_0, \beta_1$ defined on a finite set $B$. We formulate the cheating strategies for Alice and Bob forcing an outcome of $0$ or $1$ as semidefinite programs in the style of Kitaev~\cite{Kit03}. It can then be shown that the optimal cheating probabilities of a cheating Alice and a cheating Bob can be written as the maximization of a linear combination of fidelity functions over respective polytopes $\calP_{\A}$ and $\calP_{\B}$ (this was also proved in our previous work~\cite{NST14} using direct arguments). For example, 
\[ P_{\A,0}^* = \max \left\{ \half \sum_{a \bit} \sum_{y \in B} \beta_{a,y} \;  \rF(s^{(a,y)}, \alpha_{a}) \; : \;
(s_{1}, \ldots, s_{n}, s) \in \calP_{\A} \right\}, \]
where $s^{(a,y)}$ is the projection of $s$ onto the fixed indices $a$ and $y$,  and  
\[ P_{\B,1}^* = \max \left\{ \half \sum_{a \bit} \, \rF \left( (\alpha_a \otimes \id_B)^{\T} p_n, \, \beta_{\bar{a}} \right) : (p_{1}, \ldots, p_{n}) \in \calP_{\B} \right\}. \] 
(See Theorems~\ref{thm:reducedBob} and \ref{thm:reducedAlice} for formal statements of Bob's and Alice's cheating probabilities, respectively.)
We discuss how these optimization problems can be written as semidefinite programs (which are much simpler than the original formulations) and, furthermore, it was noted in~\cite{NST14} that one can use \emph{second-order cone programming} to model such optimization problems. We remark why this is interesting below. 

Using the above semidefinite programs, we develop the point games~\cite{Moc07, AharonovCGKM14} corresponding to a $\BCCF$-protocol, which we call $\BCCF$-point games. We then prove connections between the cheating probabilities in $\BCCF$-protocols and the \emph{final point} $\pg{\zeta_{\B,1}}{\zeta_{\A,0}}$ of a $\BCCF$-point game. More precisely, we prove that the final point $\pg{\zeta_{\B,1}}{\zeta_{\A,0}}$ of any $\BCCF$-point game satisfies $P_{\A,0}^* \leq \zeta_{\A,0}$ and $P_{\B,1}^* \leq \zeta_{\B,1}$ in the corresponding $\BCCF$-protocol and there exist point games with final point $\pg{P_{\B,1}^*}{P_{\A,0}^*}$. To bound all four cheating probabilities in a $\BCCF$-protocol, we consider the point games in \emph{pairs}, one of which bounds $P_{\A,0}^*$ and $P_{\B,1}^*$ and the other bounds $P_{\A,1}^*$ and $P_{\B,0}^*$.  More precisely, we have the following theorem.
 
\begin{theorem}[(Informal) See Theorem~\ref{SCFTHEOREM} for a formal statement] \label{QPGequiv}
Suppose $\pg{\zeta_{\B,1}}{\zeta_{\A,0}}$ is the final point of a $\BCCF$-point game and $\pg{\zeta_{\B,0}}{\zeta_{\A,1}}$ is the final point of its pair. Then
\[ P_{\B,0}^* \leq \zeta_{\B,0}, \quad P_{\B,1}^* \leq \zeta_{\B,1}, \quad P_{\A,0}^* \leq \zeta_{\A,0}, \; \textup{ and } \; P_{\A,1}^* \leq \zeta_{\A,1}. \] 
Moreover, there exists a pair of $\BCCF$-point games with final points 
$\pg{P_{\B,1}^*}{P_{\A,0}^*}$ and $\pg{P_{\B,0}^*}{P_{\A,1}^*}$.  
\end{theorem}

This is a restatement of weak duality/strong duality of semidefinite programming in the context of protocols and point games. We discuss these connections in  Section~\ref{BCCFPG}. 

Our analysis of the quantum protocols shows similarities to a related family of classical coin-flipping protocols based on bit-commitment, which we call classical $\BCCF$-protocols. 
We can write the maximum cheating probabilities of these classical protocols in a very similar way, but using linear programming instead of semidefinite programming or second-order cone programming. For example, we can write 
\[ P_{\A,0}^* = \max \left\{ 
\half \sum_{a \in A'_0} \sum_{y \in B} \sum_{x \in \supp(\alpha_a)} \beta_{a,y} s_{a,x,y}
: (s_{1}, \ldots, s_{n}, s) \in \calP_{\A} \right\} \] 
and 
\[ P_{\B,1}^* = \max \left\{ 
\half \sum_{a \in A'_0} \sum_{y \in \supp(\beta_{\bar{a}})} \sum_{x \in A} \alpha_{a,x} \, p_{n,x,y}
: (p_{1}, \ldots, p_{n}) \in \calP_{\B} \right\}. \]  
Using the similarities to the quantum case, we develop their point games as well which we call classical $\BCCF$-point games. Considering them in pairs, we have the classical version of Theorem~\ref{QPGequiv}, below. 

\begin{theorem}[(Informal) See Theorem~\ref{SCFTHEOREMc} for a formal statement] \label{CPGequiv}
Suppose $\pg{\zeta_{\B,1}}{\zeta_{\A,0}}$ is the final point of a classical $\BCCF$-point game and $\pg{\zeta_{\B,0}}{\zeta_{\A,1}}$ is the final point of its pair. Then
\[ P_{\B,0}^* \leq \zeta_{\B,0}, \quad P_{\B,1}^* \leq \zeta_{\B,1}, \quad P_{\A,0}^* \leq \zeta_{\A,0}, \; \textup{ and } \; P_{\A,1}^* \leq \zeta_{\A,1}, \] 
where $P_{\B,0}^*, P_{\B,1}^*, P_{\A,0}^*, P_{\A,1}^*$ are the maximum cheating probabilities  for the classical $\BCCF$-protocol. Moreover, there exists a pair of classical $\BCCF$-point games with final points 
$\pg{P_{\B,1}^*}{P_{\A,0}^*}$ and $\pg{P_{\B,0}^*}{P_{\A,1}^*}$.  
\end{theorem}

The relationships between the quantum and classical versions of the $\BCCF$-protocols and $\BCCF$-point games are illustrated in Figure~\ref{crystal-intro} below.  

\begin{figure}[h]
  \centering
   \includegraphics[width=6.25in] {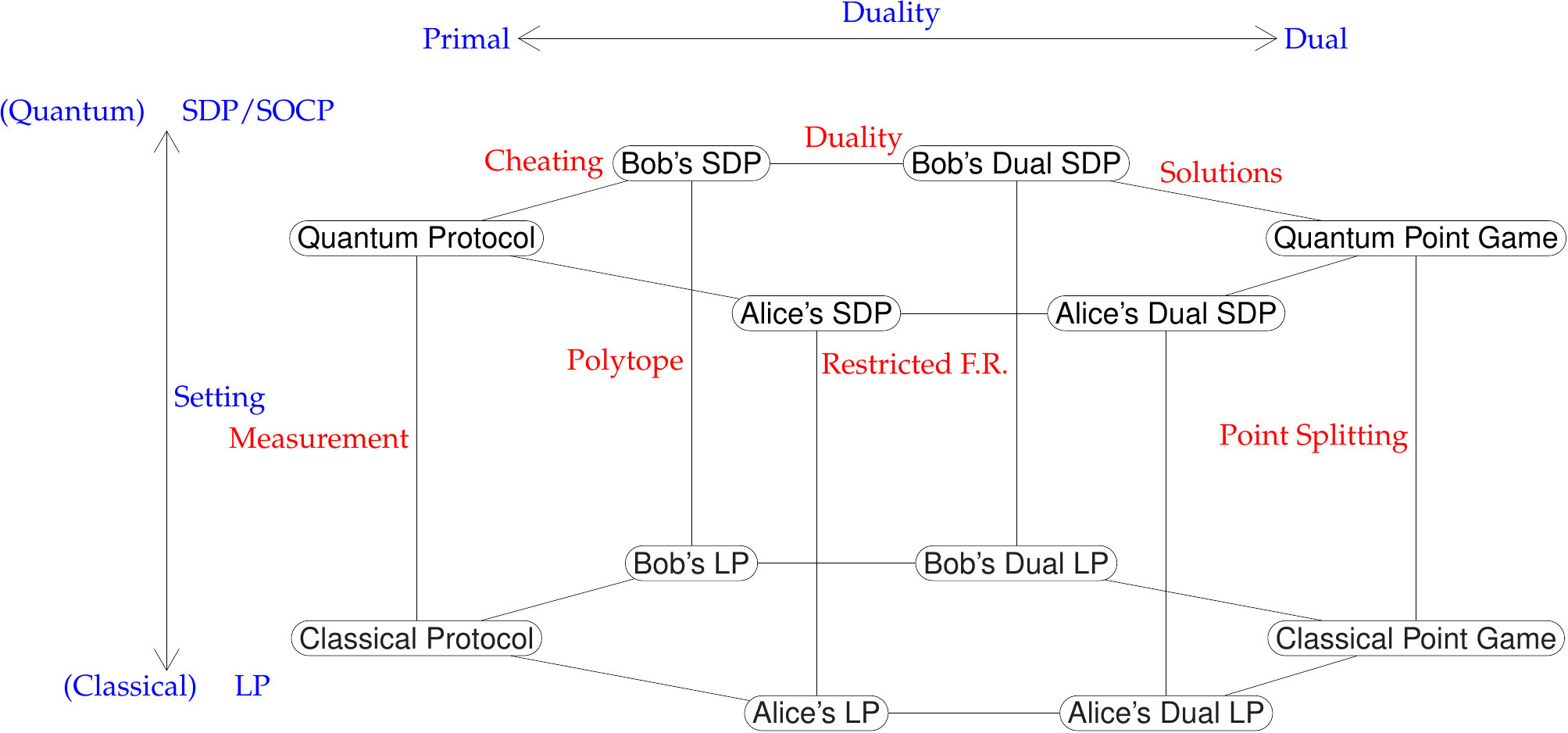} 
  \caption{Crystal structure of $\BCCF$-protocols. F.R. denotes ``feasible region'', SDP abbreviates ``semidefinite programming", SOCP abbreviates ``second-order cone programming", and LP abbreviates ``linear programming''.}
\label{crystal-intro}
\end{figure}

This figure gives a nice philosophical view of how the generalization of quantum mechanics from classical mechanics is analogous to the generalization of semidefinite programming from linear programming. As mentioned previously, it was shown in~\cite{NST14} that the optimal cheating strategies in the quantum version can be formulated using \emph{second-order cone programming} which is a special case of semidefinite programming but still a generalization of linear programming (see Subsection~\ref{ssect:opt}). This suggests that $\BCCF$-protocols are very simple compared to general quantum protocols, which is indeed the case. However, they are still provably more general than classical protocols. To put another way, just as the simple structure that makes our family of quantum protocols fit nicely between the set of classical and quantum protocols, the class of optimization problems that can be modelled as second-order cone programs fits nicely between those that can be modelled as linear programs and those that can be modelled as semidefinite programs. We discuss further this analogy and how to view a spectrum of optimization problems between linear programming and semidefinite programming (and beyond) in Section~\ref{class}. 

{Independent of our work and observations above, a similar phenomenon was exposed by Fiorini, Massar, Pokutta, Tiwary, and de Wolf~\cite{FMPTW2012} in research involving extended linear programming vs. extended semidefinite programming  formulations in combinatorial optimization.} 

Moreover, we can use these relationships to prove theoretical results. In particular, by examining the classical $\BCCF$-point games, we can prove that at least one party can cheat with probability $1$. A closer look reveals that there is no classical $\BCCF$-protocol where both parties can cheat with probability $1$ (which extends to the quantum case as well). This is summarized in the following theorem.
 
\begin{theorem}[(Informal) See Theorem~\ref{Thm:security} for a formal statement] 
Alice and Bob cannot {both} cheat perfectly in a quantum $\BCCF$-protocol. Exactly one of Alice or Bob can cheat perfectly in a classical $\BCCF$-protocol.
\end{theorem}
 
We then address the problem of finding the smallest  bias for quantum $\BCCF$-protocols. We do this by examining what happens when both of Kitaev's lower bounds 
$P_{\A,0}^* P_{\B,0}^* \geq 1/2$ and $P_{\A,1}^* P_{\B,1}^* \geq 1/2$ are  saturated. 
 
\begin{theorem}[(Informal) See Theorem~\ref{Kitproof} for a formal statement] 
If a quantum $\BCCF$-protocol saturates both of Kitaev's lower bounds, then the cheating probabilities are the same as in the corresponding classical protocol. 
\end{theorem}
We can combine the above two results to use classical protocols to lower bound the quantum bias. 
  
\begin{corollary} \label{CorSec}
In {every} quantum $\BCCF$-protocol, we have $\max \{ P_{\A,0}^*, P_{\A,1}^*, P_{\B,0}^*, P_{\B,1}^* \} > 1/ \sqrt{2}$.
\end{corollary}

\subsection{Organization of the paper} 

We start with establishing notation {and terminology on linear algebra, optimization problems of interest and some technical lemmas} in Section~\ref{BG}. Background on coin-flipping and Kitaev's protocol and point game formalisms can be found in Appendix~\ref{CFSDP}. In Section~\ref{family}, we introduce the family of {quantum} protocols we consider in this paper and {formulate} their cheating strategies using semidefinite programming. The corresponding point games are developed and analyzed in Section~\ref{BCCFPG}. A family of related classical protocols and {their} point games are examined in Section~\ref{class} and used to lower bound the quantum bias in Section~\ref{LB}. We end with conclusions in Section~\ref{conc}. 
 
\section{Background} 
\label{BG}

In this section, we establish the notation and the necessary background for this paper.

\subsection{Linear algebra}

For a finite set $A$, we denote by $\R^A$, $\R_+^A$, $\Prob^A$, and $\C^A$ the set of real vectors, nonnegative real vectors, probability vectors, and complex vectors, respectively, each indexed by $A$. We use $\R^n$, $\R_+^n$, $\Prob^n$, and $\C^n$ for the special case when $A = \set{1, \ldots, n}$. We denote by $\mathbb{S}^A$ and $\pos^A$ the set of Hermitian matrices and positive semidefinite Hermitian matrices, respectively, each over
the reals with columns and rows indexed by $A$.

It is convenient to define $\sqrt{x}$ to be the element-wise square root of a nonnegative vector $x$. The element-wise square root of a probability vector yields a unit vector (in the Euclidean norm). This operation, in some sense, is a conversion of a probability vector to a quantum state. For a vector $p \in \R^A$, we denote by $\Diag(p) \in \mathbb{S}^A$ the diagonal matrix with $p$ on the diagonal. For a matrix $X \in \mathbb{S}^A$, we denote by $\diag(X) \in \R^A$ the vector on the diagonal of $X$. For a vector $x \in \C^A$, we denote by $\supp(x)$ the set of indices of $A$ where $x$ is nonzero. We denote by $x^{-1}$ the vector of inverses, i.e., each entry in the support of $x$ is inverted, and $0$ entries are mapped to $0$.

For vectors $x$ and $y$, the notation $x \geq y$ denotes that $x-y$ has nonnegative entries, $x > y$ denotes that $x-y$ has positive entries, and for {Hermitian} matrices $X$ and $Y$, the notation $X \succeq Y$ denotes that  $X - Y$ is positive semidefinite, and $X \succ Y$ denotes $X - Y$ is positive definite when the underlying spaces are clear from context.

The Schur {complement} of the block matrix $X := \left[ \begin{array}{cc} A & B \\ C & D \end{array} \right]$ is $S := A - BD^{-1}C$. Note that when $D \succ 0$ {and $C = B^*$ with $A$ Hermitian}, then $X \succeq 0$ if and only if $S \succeq 0$.

{The Kronecker product of two  matrices $X$ and $Y$, denoted $X \otimes Y$, is defined such that the $i,j$'th block is equal to $X_{i,j} \cdot Y$. Note that $X \otimes Y \in \pos^{A \times B}$ when $X \in \pos^A$ and $Y \in \pos^B$ and $\tr(X \otimes Y) = \tr(X) \cdot \tr(Y)$ when $X$ and $Y$ are square.}

The \emph{Schatten $1$-norm}, or \emph{trace norm}, of a matrix $X$ is defined as
\[ \norm{X}_1 := \tr(\sqrt{X^* X}), \] 
where $X^*$ is the adjoint of $X$ and $\sqrt{X}$ denotes the {Hermitian} square root of a {Hermitian} positive semidefinite matrix $X$, i.e., the {Hermitian} positive semidefinite matrix $Y$ such that $Y^2 = X$. Note that the $1$-norm of a matrix is the sum of its singular values. The $1$-norm of a vector $p \in \C^A$ is defined as
\[ \norm{x}_1 := \sum_{x \in A} |p_x|. \]

For a matrix $X$, we denote by $\nulll(X)$ the nullspace of $X$. We denote by $\inner{X}{Y}$ the standard inner product of matrices acting on the same space given by $\tr(X^* Y)$.

We use the notation $\bar{a}$ to denote the complement of a bit $a$ with respect to $0$ and $1$ and $a \oplus b$ to denote the XOR of the bits $a$ and $b$. We use $\mathbb{Z}_2^n$ to denote the set of $n$-bit binary strings.

A convex set $C$ is a \emph{convex cone} if $\lambda x \in C$ when $\lambda \geq 0$ and $x \in C$. The dual of the convex cone $C$, denoted $C^*$, is the set $\{ y : \inner{x}{y} \geq 0, \, \forall x \in C \}$.

A function $f: \mathbb{S}^n \to \mathbb{S}^m$ is said to be \emph{operator monotone} if
\[ f(X) \succeq f(Y) \quad \textup{ when } \quad X \succeq Y. \]
The set of operator monotone functions is a convex cone.

A \emph{polyhedron} is the solution set of a system of finitely many linear inequalities (or equalities). A \emph{polytope} is a bounded polyhedron.

The \emph{(quantum) partial trace over ${A_1}$}, denoted $\tr_{{A_1}}$, is defined as the unique linear transformation 
which satisfies 
\[ \tr_{A_1}(\rho_1 \otimes \rho_2) = \tr(\rho_1) \cdot \rho_2 \] for all $\rho_1 \in \mathbb{S}^{A_1}$ and $\rho_2 \in \mathbb{S}^{A_2}$. More explicitly, given any matrix $X \in {\mathbb{S}}^{A_1 \times A_2}$, we have 
\[ \tr_{A_1}(X) := \sum_{x_1 \in A_1} \left( e_{x_1}^* \otimes \id_{A_2} \right) X \left( e_{x_1} \otimes \id_{A_2} \right), \]
where $\set{ e_{x_1}: x_1 \in A_1}$ is the standard basis for $\C^{A_1}$. In
fact, the definition is independent of the choice of basis, so long as
it is orthonormal. 
{The adjoint of the partial trace is the transformation   
$\tr_A^*(X) = X \otimes \id_A$.} 

We also define the \emph{classical partial trace over $A_1$}, denoted $\tr_{{A_1}}: \C^{A_1 \times A_2} \to \C^{A_2}$, {as the linear transformation} 
\[ \tr_{A_1}(p) = (e_{A_1}^{\T} \otimes \id) \, p, \] 
where $e_{A_1}$ is the vector of all ones indexed by $x_1 \in A_1$. 
If $p$ is a probability vector over $A_1 \times A_2$, then $\tr_{A_1}(p)$ is the marginal probability vector of $p$ over $A_2$. 
{The adjoint of the classical partial trace is the transformation   
$\tr_A^*(p) = p \otimes e_A$.} 

We define the \emph{fidelity} of two nonnegative vectors $p,q \in \R_+^A$ as
\[ \rF(p,q) := \left( \sum_{x \in A} \sqrt{p_x} \sqrt{q_x} \right)^2 \]
and the fidelity of two positive semidefinite matrices $\rho_1$ and $\rho_2$ as
\[ \rF(\rho_1, \rho_2) := \norm{\sqrt{\rho_1}\sqrt{\rho_2}}_1^2. \]
Notice, $\rF(\rho_1, \rho_2) \geq 0$ with equality if and only if $\inner{\rho_1}{\rho_2} = 0$ and, if $\rho_1$ and $\rho_2$ are quantum states, $\rF(\rho_1, \rho_2) \leq 1$ with equality if and only if $\rho_1 = \rho_2$. An analogous statement can be made for the fidelity over probability vectors.

Another distance measure is the \emph{trace distance}. We define the trace distance between two probability vectors $p$ and $q$, denoted $\Delta(p, q)$, as
\[ \Delta(p, q) := \half \norm{p - q}_1. \] 
{This is also commonly known as the total variation distance.} 
We similarly define the trace distance between two quantum states $\rho_1$ and $\rho_2$ as
\[ \Delta(\rho_1, \rho_2) := \frac{1}{2} \norm{\rho_1 - \rho_2}_1. \]
Notice $\Delta(\rho_1, \rho_2) \geq 0$ with equality if and only if $\rho_1 = \rho_2$ and $\Delta(\rho_1, \rho_2) \leq 1$ with equality if and only if $\inner{\rho_1}{\rho_2} = 0$. The analogous statement can be made for the trace distance between probability vectors. 

We use the notation $\eig(X)$ to denote the set of (distinct) eigenvalues of a matrix $X$ {and $\Pi^{[\lambda]}_X$ to denote the projection onto the eigenspace of $X$ corresponding to the eigenvalue $\lambda \in \eig(X)$}. 

\subsection{Optimization classes}\label{ssect:opt} 

\subsubsection{Semidefinite programming} 
 
A natural {class} of optimization {problems} when studying quantum information is {semidefinite programming}. A {semidefinite program}, abbreviated
as $\SDP$, {is an optimization
problem with finitely many Hermitian matrix variables, a linear objective function of these variables,
and finitely many constraints enforcing positive semidefiniteness of some linear functions of these variables.
Every SDP can be put into the following standard form using some elementary reformulation tricks:} 
\[ \begin{array}{rrrcllllllllllllll}
\textrm{(P)} & \sup                         & \inner{C}{X} \\
                     & \textrm{subject to} & \calA(X) & = & b, \\
                     &                                   & X & \in & \Pos^n,  
\end{array} \] 
where $\calA: \Herm^n \to \R^m$ is linear, $C \in \Herm^n$, and $b \in \R^m$.
The SDPs that arise in quantum computation involve optimization
over complex matrices.
However, they may be transformed to the above standard form in a
straightforward manner, by observing that Hermitian matrices form
a real subspace of the vector space of~$n \times n$ complex matrices.
{We remark here that the data defining the optimization problems in this paper {are} always real and thus we can restrict ourselves to real matrix variables without loss of generality.}

{We can write the \emph{dual} of (P) as} 
\[ \begin{array}{rrrcllllllllllllll}
\textrm{(D)} & \inf                         & \inner{b}{y} \\
                     & \textrm{subject to} & \calA^*(y) - S & = & C, \\
                     &                                   & S & \in & \Pos^n, 
\end{array} \]
where $\calA^*$ is the adjoint of $\calA$. We refer to (P) as the primal problem and to (D) as its dual. It is straightforward to verify that the dual of (D) is (P).

We say~$X$ is \emph{feasible\/} for (P) if it satisfies the
constraints~$\calA(X) = b$ and~$X \in \Pos^n$, and~$(y,S)$ is feasible
for~(D) if~$\calA^*(y) - S = C$, and $S \in \Pos^n$.
The usefulness of defining the dual in the above manner is apparent in the following lemmas.

\begin{lemma}[Weak duality]
For every $X$ feasible for $\textup{(P)}$ and $(y,S)$ feasible for $\textup{(D)}$ we have
\[ \inner{C}{X} \leq \inner{b}{y}. \]
\end{lemma}

Using weak duality, we can prove bounds on the optimal objective value of $\textup{(P)}$ and $\textup{(D)}$, i.e., the objective function value of any primal feasible solution yields a lower bound on $\textup{(D)}$ and the objective function value of any dual feasible solution yields an upper bound on $\textup{(P)}$.

Under mild conditions, we have that the optimal objective values of $\textup{(P)}$ and $\textup{(D)}$ coincide.

\begin{lemma}[Strong duality]
If the objective function of $\textup{(P)}$
is bounded from above
on the set of feasible solutions of $\textup{(P)}$ and there exists a strictly feasible solution, i.e., there exists $\bar{X} \succ 0$ such that $\calA(\bar{X}) = b$, then $\textup{(D)}$ has an optimal solution and the optimal objective values of $\textup{(P)}$ and $\textup{(D)}$ coincide.
\end{lemma}
A strictly feasible solution as in the above lemma is also called a
\emph{Slater point}. Semidefinite programming has a powerful and rich duality theory and the interested reader is referred to \cite{SDP}, \cite{TW08}, and the references therein. 

\subsubsection{Second-order cone  programming} 

The \emph{second-order cone} (or \emph{Lorentz cone}) in~$\R^n$, $n \ge
2$, is defined as
\[ \Lor^n := \set{(x,t) \in {\R^{n-1} \oplus \R} : t \geq \norm{x}_2}. \]
A \emph{second-order cone program}, denoted SOCP, is an optimization problem of the form
\[ \begin{array}{rrrcllllllllllllll}
\textrm{(P)} & \sup                         & \inner{c}{x} \\
                     & \textrm{subject to} & Ax & = & b, \\
                     &                                   & x & \in & \Lor^{n_1} \oplus \cdots \oplus \Lor^{n_k}, 
\end{array} \]
where $A$ is an $m \times (\sum_{i=1}^k n_k)$ matrix, $b \in \R^m$, {$c \in \R^{\sum_{i=1}^k n_k}$,} and $k$ is finite. We say that a feasible solution $\bar{x}$ is strictly feasible if $\bar{x}$ is in the interior of $\Lor^{n_1} \oplus \cdots \oplus \Lor^{n_k}$.

An SOCP also has a dual which can be written as
\[ \begin{array}{rrrcllllllllllllll}
\textrm{(D)} & \inf                         & \inner{b}{y} \\
                     & \textrm{subject to} & A^\top y - s & = & c, \\
                     &                                   & s & \in & \Lor^{n_1} \oplus \cdots \oplus \Lor^{n_k}. \\
\end{array} \]
Note that weak duality and strong duality also hold for SOCPs for the {above} definition of a strictly feasible
solution.

A related cone, called the \emph{rotated second-order cone}, is defined as
\[ \RL^n := \set{(a,b,x) \in {\R \oplus \R \oplus \R^{n-2}} : a, b \geq 0, \, 2 a b \geq \norm{x}_2^2}. \]
Optimizing over the rotated second-order cone is also called second-order cone programming because $(x,t) \in \Lor^{n}$ if and only if $(t/2,t,x) \in \RL^{n+1}$ and $(a,b,x) \in \RL^n$ if and only if $(x,a,b,a+b) \in \Lor^{n+1}$
and~$a,b \ge 0$. In fact, both second-order cone constraints can be cast as positive semidefinite constraints:
\[ t \geq \norm{x}_2 \iff \left[ \begin{array}{cc} t & x^\top \\ x & t \, \id \end{array} \right] \succeq 0 \quad \textup{ and } \quad a,b \geq 0, \, 2ab \geq \norm{x}_2^2 \iff \left[ \begin{array}{cc} 2a & x^\top \\ x & b \, \id \end{array} \right] \succeq 0. \]

{Despite second-order cone programming being a special case of semidefinite programming, there} are some notable differences. One is that the algorithms for solving second-order cone programs can be more efficient and robust than those for solving semidefinite programs. We refer the interested reader to~\cite{Stu99, Stu02, Mit03, AG03} and the references therein.

\subsubsection{Linear programming} 

A linear program, denoted LP, {is an optimization problem of the form} 
\[ \begin{array}{rrrcllllllllllllll}
\textrm{(P)} & \max                         & \inner{c}{x} \\
                     & \textrm{subject to} & Ax & = & b, \\
                     &                                   & x & \in & \R_+^n, 
\end{array} \]    
where $A$ is an $m \times n$ matrix, $c \in \R^n$ and $b \in \R^m$.

Linear programming is a special case of both second-order cone programming and semidefinite programming. This can be seen by casting a nonnegativity constraint  $t \geq 0$ as the SOC constraint $(0,t) \in \Lor^2$. Associated with every linear program is its dual which is defined as 
\[ \begin{array}{rrrcllllllllllllll}
\textrm{(D)} & \min                         & \inner{b}{y} \\
                     & \textrm{subject to} & A^{\T} y - s & = & c, \\
                     &                                   & s & \in & \R_+^n. 
\end{array} \]  
Note that in this special case, we do not require strict feasibility to guarantee strong duality. {If a linear program is feasible and its objective function is bounded over its feasible region, then} it and its dual attain an optimal solution and the optimal values always coincide. 

\subsection{Technical lemmas} \label{techlemmas}

In this subsection, we present a few lemmas which are helpful in the analysis in this paper.

\begin{lemma}[\cite{NST14}] \label{FidelityLemma}
For {every} $p, q \in \R_+^{A}$, we have
\[ \rF(p,q) = \max \{ \inner{X}{\sqrtt{p}} : \diag(X) = q, \, X \in \pos^{A} \}. \]
\end{lemma} 
Note that 
\[
\rF(p,q) 
\! = \! 
\inf_{y \in \R^A} \; \{ \inner{y}{q} : \Diag(y) \succeq \sqrtt{p} \}
\! = \!
\inf_{y > 0} \; \{ \inner{y}{q} : \inner{y^{-1}}{p} \leq 1 \}
\! = \! 
\inf_{y > 0} \; \{ \inner{y}{q} \inner{y^{-1}}{p} \}
\]
by using the observation
\[ \Diag(y) \succeq \sqrtt{p} \iff \id_A \succeq \Diag(y)^{-1/2} \sqrtt{p}   \Diag(y)^{-1/2} \iff  1 \geq \sum_{x \in A} \frac{p_x}{y_x}. \]
We use this characterization of the inequality $\Diag(y) \succeq \sqrtt{p}$ several times throughout this paper.

Notice that $\rF(p,q) = \inf_{y > 0} \; \{ \inner{y}{q} \inner{y^{-1}}{p} \}$ is the classical version of Alberti's Theorem~\cite{A83}, which states that $\rF (\rho, \sigma) = \inf_{X \succ 0} {\inner{X}{\rho} \inner{X^{-1}}{\sigma}}$ 
for quantum states $\rho$ and $\sigma$.

We can apply the same trick above to the inequality $\Diag(y) \otimes \id_A \succeq \kb{\psi}$, when $y > 0$ to get the equivalent condition $1 \geq \bra{\psi} \Diag(y)^{-1} \otimes \id_A \ket{\psi}$, 
which works for any $\ket{\psi} \in \C^{A \times A}$. In particular, we have the following lemma.

\begin{lemma} \label{FidLemma2}
For {every} $p \in \R_+^A$ and $\ket{\psi} := \sum_{x \in A} \sqrt{p_x} \, \ket{xx}$, we have 
\[ \{ y > 0 : \Diag(y) \succeq \sqrtt{p} \} = \{ y > 0 : \Diag(y) \otimes \id_A \succeq \kb{\psi} \}. \]
\end{lemma}
 
We also make use of the lemma below.

\begin{lemma}[\cite{NST14}] \label{mn}
For {every} $\beta_0, \beta_1 \in \prob^B$, we have 
\[ \sum_{y \in B} \max_{a \bit} \set{\beta_{a,y}} = 1 + \Delta(\beta_0, \beta_1). \] 
\end{lemma}

\section{A family of quantum coin-flipping protocols}
\label{family}

In this section we introduce the coin-flipping protocols examined in this paper. Intuitively, Alice ``commits'' to a bit $a$ (in superposition) by creating a state $\ket{\psi_a}$ and revealing its subsystems one at a time. Bob does the same, he ``commits'' to a bit $b$ by creating a state $\ket{\phi_b}$ and revealing its subsystems one at a time. Afterwards, they reveal their bits to each other and the outcome of the protocol is $a \oplus b$, if they both pass cheat detection. 

We now formally define the class of protocols considered in this paper.
 
\begin{protocol}[$\mathbf{\BCCF}$-protocol~\cite{NST14}]
A \emph{coin-flipping protocol based on bit-commitment\/}, denoted here as a $\BCCF$-protocol, is specified
by four finite sets
\[ 
A_0 := \{ 0, 1 \}, 
\quad 
A := A_1 \times A_2 \times \cdots \times A_{n},  
\quad 
B_0 := \{ 0, 1\}, 
\quad 
B := B_1 \times B_2 \times \cdots \times B_{n},  
\] 
two probability distributions $\alpha_0, \alpha_1$ over $A$, 
and two probability distributions $\beta_0, \beta_1$ over $B$. From these parameters, we define the  quantum states: 
\begin{eqnarray*} 
\ket{\psi} := \frac{1}{\sqrt 2} \sum_{a \bit} \ket{aa} \ket{\psi_a} \in \C^{A_0 \times A'_0 \times A \times A'} 
\quad & \text{where} & \quad 
\ket{\psi_a} := \sum_{x \in A} \sqrt{\alpha_{a,x}} \ket{xx} \in \C^{A \times A'}, \\ 
\ket{\phi} := \frac{1}{\sqrt 2} \sum_{b \bit} \ket{bb} \ket{\phi_b} \in \C^{B_0 \times B'_0 \times B \times B'} 
\quad & \text{where} & \quad 
\ket{\phi_b} := \sum_{y \in B} \sqrt{\beta_{b,y}} \ket{yy} \in \C^{B \times B'}, 
\end{eqnarray*} 
and $A'_0 := A_0$, $A' := A$, $B'_0 := B_0$, and $B' := B$ are copies.  

The preparation, communication, and cheat detection of the protocol proceed as follows: 
\begin{itemize}
\item Alice prepares the state 
$\ket{\psi}$ and Bob prepares the state $\ket{\phi}$. 
\item For $i$ from $1$ to $n$: Alice sends $\C^{A_i}$ to Bob who replies with $\C^{B_i}$.
\item Alice fully reveals her bit by sending $\C^{A'_0}$. She also sends $\C^{A'}$ which Bob uses later to check if she was honest. Bob then reveals his bit by sending $\C^{B'_0}$. He also sends $\C^{B'}$ which Alice uses later to check if he was honest.
\item Alice performs the measurement 
$(\Pi_{\rA,0}, \Pi_{\rA,1}, \Pi_{\rA, \abort})$ on the space $\Pos^{A_0 \times B'_0 \times B \times B'}$, where 
\[ \Pi_{\rA,0} := \sum_{b \bit} \kb{b} \otimes \kb{b}  \otimes \kb{\phi_b}, \quad
\Pi_{\rA,1} := \sum_{b \bit} \kb{{\bar b}} \otimes \kb{b} \otimes \kb{\phi_b}, \]
and $\Pi_{\rA, \abort} := \id - \Pi_{\rA,0} - \Pi_{\rA,1}$. 
\item Bob performs the measurement  
$(\Pi_{\rB,0}, \Pi_{\rB,1}, \Pi_{\rB, \abort})$ on the space 
$\Pos^{B_0 \times A'_0 \times A \times A'}$, where
\[ \Pi_{\rB,0} := \sum_{a \bit} \kb{a} \otimes \kb{a}  \otimes \kb{\psi_a}, \quad
\Pi_{\rB,1} := \sum_{a \bit} \kb{{\bar a}} \otimes \kb{a} \otimes \kb{\psi_a}, \]
and $\Pi_{\rB, \abort} := \id - \Pi_{\rB,0} - \Pi_{\rB,1}$. (These last two steps can be interchanged.)
\end{itemize}
\end{protocol} 

A six-round $\BCCF$-protocol is depicted in Figure~\ref{BCCFprotocol}. Note that the measurements check two things. First, it checks whether the outcome, $a \oplus b$, is $0$ or $1$. The first two terms determine this, i.e., whether $a = b$ or if $a \neq b$. Second, it checks whether the other party was honest. For example, if Alice's measurement projects onto a space where $b=0$ and Bob's messages are not equal to $\ket{\phi_0}$, then Alice has detected that Bob has cheated and aborts.
 
\begin{figure}[ht]
  \centering
   \includegraphics[width=6.25in]{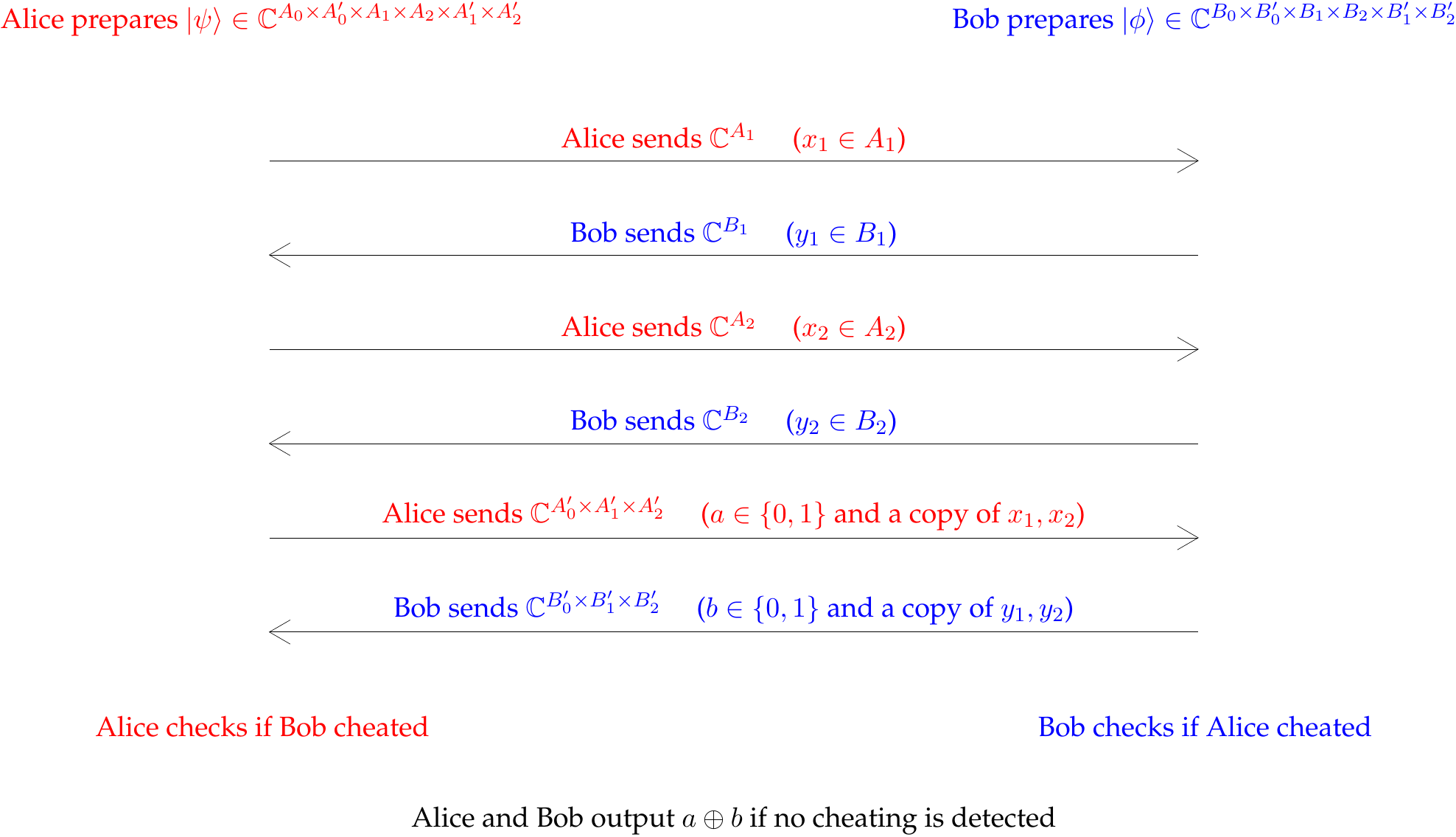} 
  \caption{A six-round $\BCCF$-protocol. Alice's actions are in red and Bob's actions are in blue.}
\label{BCCFprotocol}
\end{figure}

As is shown in Figure~\ref{BCCFprotocol}, we shall reserve the notation for indices: $a \in A_0$, $b \in B_0$, $x \in A$, $x_i \in A_i$, $y \in B$, and $y_i \in B_i$. We sometimes omit the sets when it is clear from context. 
 
\subsection{Formulating optimal quantum cheating strategies as semidefinite programs}
\label{ssect:form}

We can formulate strategies for cheating Bob and cheating Alice as semidefinite programs in the same manner as {Kitaev}, as discussed in Appendix~\ref{CFSDP}. The extent to which Bob can cheat is captured by the following lemma.

\begin{lemma}[\cite{NST14}]
Bob's optimal cheating probability for forcing honest Alice to accept the outcome $c \bit$ is given by the optimal objective value of the following semidefinite program:
\[ \begin{array}{rrrcllllllllllllll}
& P_{\B,c}^* \; = \; \sup                         & \langle \, \rho_F , \Pi_{\A,c} \, \rangle \\
                     & \textup{subject to} & \tr_{B_1}(\rho_1) & = & \tr_{A_1} \kb{\psi}, \\
                     &                                   & \tr_{B_j} (\rho_j) & = & \tr_{A_j} (\rho_{j-1}), & \forall j \in \{ 2, \ldots, n \}, \\
                     &                                   & \tr_{B' \times B'_0}(\rho_F) & = & \tr_{A' \times A'_0}(\rho_n), \\
                     &                                   & \rho_j & \in & \mathbb{S}_+^{A_0 \times A'_0 \times B_1 \times \cdots \times B_j \times A_{j+1} \times \cdots \times A_n \times A'}, & \forall j \in \{ 1, \ldots, n \}, \\
                     &                                   & \rho_F & \in & \mathbb{S}_+^{A_0 \times B'_0 \times B \times B'}. 
\end{array} \]
\end{lemma}

The actions of a cheating Bob and the variables in the SDP are depicted in Figure~\ref{BCCFprotocol6B}. 
  
\begin{figure}[ht]
  \centering
   \includegraphics[width=6.25in]{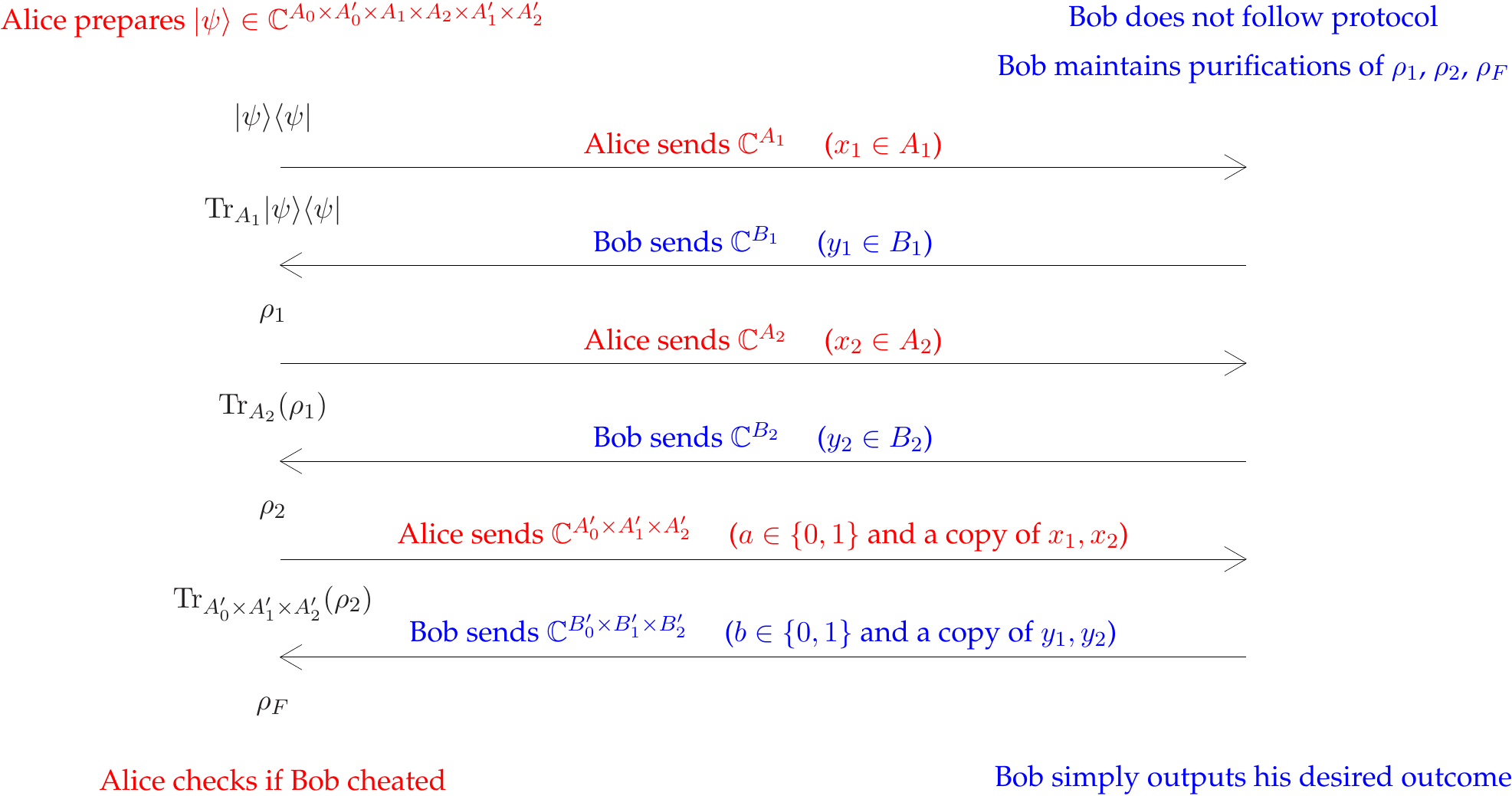} 
  \caption{Bob cheating in a six-round $\BCCF$-protocol.}
\label{BCCFprotocol6B}
\end{figure}
 
We now present a theorem showing that the cheating SDPs can have a certain, restricted form while retaining the same optimal objective value. {At high level, we cut down the algebraic representation of the feasible region. 
Surprisingly, we are able to reformulate the feasible region by a polytope defined below.}
   
\begin{definition} \label{BobsPolytope}
We define \emph{Bob's cheating polytope}, denoted $\calP_{\B}$, as the set of vectors 
$(p_{1}, p_{2}, \ldots, p_{n})$ satisfying
\[ \begin{array}{rrrcllllllllllllll}
                     &  & \tr_{B_1}(p_1) & = & e_{A_{1}}, \\
                     &  & \tr_{B_2} (p_2) & = & p_{1} \otimes e_{A_{2}}, \\
                     & & & \vdots \\
                     &  & \tr_{B_n} (p_n) & = & p_{n-1} \otimes e_{A_{n}}, \\
                     & & p_j & \in & \R_+^{A_{1} \times B_{1} \times \cdots \times A_{j} \times B_{j}}, \; \forAll j \in \{ 1, \ldots, n \},
\end{array} \]
where $e_{A_j}$ denotes the vector of all ones in the corresponding space $\C^{A_j}$.
\end{definition}

We now use Bob's cheating polytope to capture his optimal cheating probabilities. 
 
\begin{theorem}[Bob's reduced problems~\cite{NST14}]
\label{thm:reducedBob}
For the $\BCCF$-protocol defined by the parameters $\alpha_0, \alpha_1 \in \prob^A$ and $\beta_0, \beta_1 \in \prob^B$, we have
\[ P_{\B,0}^* = \max \left\{ \half \sum_{a \bit} \, \rF \left( (\alpha_a \otimes \id_B)^{\T} p_n, \, \beta_a \right) : (p_{1}, \ldots, p_{n}) \in \calP_{\B} \right\} \]
and
\[ P_{\B,1}^* = \max \left\{ \half \sum_{a \bit} \, \rF \left( (\alpha_a \otimes \id_B)^{\T} p_n, \, \beta_{\bar{a}} \right) : (p_{1}, \ldots, p_{n}) \in \calP_{\B} \right\}. \]
\end{theorem}
      
We refer to these as \emph{Bob's reduced problems}. Note that we sometimes refer to them as \emph{Bob's reduced} SDPs,  implying we have replaced the fidelity with its SDP characterization from Lemma~\ref{FidelityLemma}. 

The above theorem can also be proved using the fact that the set $\set{ \lambda \, xx^*: \lambda > 0}$ is an extreme ray of the cone of positive semidefinite matrices. That is, {$\lambda \, x x^* \in \pos^A$ for every $\lambda > 0$ and,} if $X_1, X_2 \in \pos^A$ satisfy $X_1 + X_2 = \lambda \, xx^*$ for some $\lambda > 0$, then $X_1 = \lambda_1 \, xx^*$ and $X_2 = \lambda_2 \, xx^*$ for some $\lambda_1, \lambda_2 \geq 0$ satisfying $\lambda_1 + \lambda_2 = \lambda$. This proof relies on a reduction of the primal problem alone and can be found in~\cite{NST14}. In Appendix~\ref{app:reduced_app}, we give an alternative proof via duality theory since some of the structure of optimal dual solutions are required for the construction of the point games in Section~\ref{BCCFPG}. In that appendix we also give context to the variables in the cheating polytope by deriving the corresponding cheating strategy.

In a similar fashion, we formulate cheating strategies for Alice in the lemma below.

\begin{lemma}[\cite{NST14}]
Alice's optimal cheating probability for forcing honest Bob to accept the outcome $c \bit$ is given by the optimal objective value of the following semidefinite program:
\[ \begin{array}{rrrcllllllllllllll}
& P_{\A,c}^* \; = \; \sup
& \inner{\sigma_F}{\Pi_{\B,c} \otimes \id_{B'_0 \times B'}} \\
                     & \textup{subject to} & \tr_{A_1}(\sigma_1) & = & \kb{\phi}, \\
                     &                               & \tr_{A_j}(\sigma_j) & = & \tr_{B_{j-1}}(\sigma_{j-1}), \!\! & \!\! \forall j \in \{ 2, \ldots, n \}, \\
                     &                               & \tr_{A' \times A'_0}(\sigma_F) & = & \tr_{B_n}(\sigma_n), \\
                     &                               & \sigma_{j} & \in & \mathbb{S}_{+}^{B_{0} \times B'_{0} \times A_1 \times \cdots \times A_j \times B_j \times \cdots \times B_n \times B'}, \!\! &  \!\!                                  \forall j \in \{ 1, \ldots, n \}, \\
                     & & \sigma_{F} & \in & \mathbb{S}_{+}^{B_{0} \times B'_{0} \times A'_{0} \times A \times A' \times B'}.
\end{array} \]
\end{lemma}
 
Similar to cheating Bob, we can reduce the feasible region to a polytope, defined below. 
   
\begin{definition} \label{AlicePolytope}
We define \emph{Alice's cheating polytope}, denoted $\calP_{\A}$, as the set of vectors 
$(s_{1}, s_{2}, \ldots, s_{n}, s)$ satisfying
\[ \begin{array}{rrrcllllllllllllll}
                     &                               & \tr_{A_1}(s_1) & = & 1, \\
                     &                               & \tr_{A_2}(s_2) & = & s_1 \otimes e_{B_{1}}, \\
                     & & & \vdots \\
                     &                               & \tr_{A_n}(s_n) & = & s_{n-1} \otimes e_{B_{n-1}}, \\
                     &                               & \tr_{A'_{0}}(s) & = & s_n \otimes e_{B_{n}}, \\
                     &                               & s_{1} & \in & \R_{+}^{A_{1}}, \\
                     &                               & s_{j} & \in & \R_{+}^{A_{1} \times B_{1} \times \cdots \times B_{j-1} \times A_{j}}, \; \forAll j \in \{ 2, \ldots, n \}, \\
                     & & s & \in & \R_{+}^{A \times B \times A'_{0}},
\end{array} \]
where $e_{B_j}$ is the vector of all ones in the corresponding space $\C^{B_j}$.
\end{definition}

We can use this polytope to capture Alice's optimal cheating probabilities. 
 
\begin{theorem}[Alice's reduced problems~\cite{NST14}]
\label{thm:reducedAlice}
For the $\BCCF$-protocol defined by the parameters $\alpha_0, \alpha_1 \in \prob^A$ and $\beta_0, \beta_1 \in \prob^B$, we have
\[ P_{\A,0}^* = \max \left\{ \half \sum_{a \bit} \sum_{y \in B} \beta_{a,y} \;  \rF(s^{(a,y)}, \alpha_{a}) \; : \;
(s_{1}, \ldots, s_{n}, s) \in \calP_{\A} \right\} \]
and
\[ P_{\A,1}^* = \max \left\{ \half \sum_{a \bit} \sum_{y \in B} \beta_{\bar{a},y} \;  \rF(s^{(a,y)}, \alpha_{a}) \; : \;
(s_{1}, \ldots, s_{n}, s) \in \calP_{\A} \right\}, \] 
where $s^{(a,y)}$ is the projection of $s$ onto the fixed indices $a$ and $y$. 
\end{theorem}
 
We refer to these as \emph{Alice's reduced problems} or \emph{Alice's reduced} SDPs when using the SDP characterization of the fidelity function from Lemma~\ref{FidelityLemma}. {Context of the variables in Alice's  cheating polytope and a proof of the above theorem are in Appendix~\ref{app:reduced_app}. 
 
\begin{remark}
We can see from Theorems~\ref{thm:reducedBob} and~\ref{thm:reducedAlice} that switching $\beta_0$ and $\beta_1$ switches the values of $P_{\B, 0}^*$ and $P_{\B, 1}^*$ and it also switches the values of $P_{\A, 0}^*$ and $P_{\A, 1}^*$. We make use of this symmetry several times in this paper.
\end{remark} 

{As an example, and for future reference, we write the dual of Bob's reduced cheating SDP for forcing outcome $1$ and the dual for Alice's reduced cheating SDP for forcing outcome $0$, respectively, below}  
\[ \begin{array}{rrcclrrrcll}
\textrm{} & \inf                         & \tr_{A_1}(w_1) & & 
& \quad  
& \inf                         & z_1 \\

                     & \textup{s.t.} & w_1 \otimes e_{B_1} & \geq & \tr_{A_2}(w_2), 
& \quad 
& \textup{s.t.} & z_1 \cdot e_{A_1} & \geq & \tr_{B_1}(z_2), \\

                     & & w_2 \otimes e_{B_2} & \geq & \tr_{A_3}(w_3), 
& \quad 
&                               & z_2 \otimes e_{A_2} & \geq & \tr_{B_2}(z_3), \\

                     & & & \vdots 
& \quad 
& & & & \vdots \\

                     &                               & w_{n} \otimes e_{B_{n}} & \geq & \half \sum_{a \in \zo} \alpha_a \otimes v_a,
& \quad 
&                               & z_n \otimes e_{A_n} & \geq & \tr_{B_n}(z_{n+1}), \\

                     & & \Diag(v_a) & \succeq & \sqrtt{\beta_{\bar{a}}}, \;\; \forall a, 
& \quad 
& & \Diag(z_{n+1}^{(y)}) & \succeq & \half \beta_{a,y} \sqrtt{\alpha_a}, \;\; \forall a, y.  
\end{array} \] 

{The structure of the reduced problems was an observation after numerically solving some cheating SDP examples.} We note that there are some similarities between the reduced problems above and the optimal solutions of the cheating SDPs for the weak coin-flipping protocols in~\cite{Moc05}. The protocols Mochon considers in~\cite{Moc05} also give rise to ``reduced problems'' being the maximization of fidelity functions over a polytope. However, the analysis is much cleaner in Mochon's work since the objective function only involves a single fidelity function as opposed to the linear combination of fidelity functions that arise {for} $\BCCF$-protocols. This difference is due to the fact that weak coin-flipping protocols often allow a stronger cheat detection step than those for strong coin-flipping. Having a single fidelity function allowed Mochon to construct an optimal solution using a  dynamic programming approach. The structure of the objective functions in the reduced problems above for $\BCCF$-protocols {has not so far revealed an obvious way} to solve it using dynamic programming, making this family of protocols harder to analyze.
   
\section{Point games for $\BCCF$-protocols} \label{BCCFPG}
 
In this section, we develop the point games corresponding to $\BCCF$-protocols. Although this section is self-contained, {an} interested reader may wish to see our Appendix~\ref{CFSDP} for a review of point games or consult the work of~\cite{Moc07, {AharonovCGKM14}}. 
 
{To summarize the idea behind point games, we take a feasible dual solution for Bob cheating towards $1$ and the same for Alice cheating towards $0$ and consider the behaviour of their eigenvalues. When pairing certain eigenvalues from Bob with those from Alice, we obtain a collection of finitely-many weighted points in the two-dimensional nonnegative orthant. The points have a time-ordering to them and the transitions from one time step to the next are called ``{moves}'' or simply ``{transitions}'' and the rules for these transitions can be described independently from the protocol description. In this section, we examine the set of allowable moves for point games {derived} from $\BCCF$-protocols in the manner described above, and use them to find a protocol independent definition.} 
 
We start by examining Kitaev's lower bound involving the quantities $P_{\B,1}^*$ and $P_{\A,0}^*$. Since we are concerned with strong coin-flipping, the choice of Bob desiring outcome $1$ and Alice desiring outcome $0$ for this part is somewhat arbitrary. However, this way we can compare them to point games for other classes of weak coin-flipping protocols (see~\cite{Moc07}). We later show that we lose no generality in choosing these two values, as we consider all four values simultaneously by viewing the point games in pairs.
 
{The dual for Bob's cheating SDP for forcing outcome $1$ is given by 
\[ \begin{array}{rrrcllllllllllllll}
\textrm{} &  P_{\B, 1}^* \; = \;  \inf                             & \inner{W_1}{\tr_{A_1} \kb{\psi}} \\
                     & \textup{subject to} & W_j \otimes \id_{B_j} & \succeq & W_{j+1} \otimes \id_{A_{j+1}}, \quad \forAll \, j \in \set{1, \ldots, n-1}, \\
                     &                               & W_n \otimes \id_{B_n} & \succeq & W_{n+1} \otimes \id_{A'} \otimes \id_{A'_0}, \\
                     &                               & W_{n+1} \otimes \id_{B'} \otimes \id_{B'_0} & \succeq & \Pi_{\A,1}, \\
                     &                               & W_j & \in & \mathbb{S}^{A_0 \times A'_0 \times B_1 \times \cdots \times B_{j-1} \times A_{j+1} \times \cdots \times A_n \times A'}, \\
                     & & & & \forAll j \in \{ 1, \ldots, n \}, \\
                     &                               & W_{n+1} & \in & \mathbb{S}^{A_0 \times B}, \\
\end{array} \]
and the dual for Alice's cheating SDP for forcing outcome $0$ is given by
\[ \begin{array}{rrrcllllllllllllll}
\textrm{} & P_{\A, 0}^* \; = \; \inf                                & \inner{Z_1}{\kb{\phi}}  \\
                     & \textup{subject to} & Z_j \otimes \id_{A_j} & \succeq & Z_{j+1} \otimes \id_{B_j}, \quad \forAll j \in \set{1, \ldots, n}, \\
                     &                               & Z_{n+1} \otimes \id_{A'} \otimes \id_{A'_0} & \succeq & \Pi_{\B,0} \otimes \id_{B'_0} \otimes \id_{B'}, \\
                     & & Z_j & \in & \mathbb{S}^{B_0 \times B'_0 \times A_1 \times \cdots \times A_{j-1} \times B_j \times \cdots \times B_n \times B'}, \\
                     & & & & \forAll j \in \{ 1, \ldots, n, n+1 \}. 
\end{array} \] 
\noindent From SDP strong duality, we know that for {every} $\delta > 0$, we can choose $(W_1, \ldots, W_{n+1})$ feasible {for} the dual of Bob's cheating SDP and $(Z_1, \ldots, Z_{n+1})$ feasible {for} the dual of Alice's cheating SDP such that 
\[ 
\left( P_{\B,1}^* + \delta \right) \left( P_{\A,0}^* + \delta \right)
> 
\inner{W_1 \otimes Z_1}{\tr_{A_1} (\kb{\psi} \otimes \kb{\phi})}. \] 
For brevity, we define $\ket{\xi_j}$ and $\ket{\xi'_j}$ equal to $\ket{\psi}\ket{\phi}$ (with the spaces permuted accordingly) to be the states of the protocol {before Alice's $j$'th message} and {before Bob's $j$'th message}, respectively, when they follow the protocol honestly. From the dual constraints, we have 
\[ 
\inner{W_j \otimes Z_j}{\tr_{A_j} \kb{\xi_j}} 
\geq
\inner{W_j \otimes Z_{j+1}}{\tr_{B_j} \kb{\xi'_j}}  
\geq 
\inner{W_{j+1} \otimes Z_{j+1}}{\tr_{A_{j+1}} \kb{\xi_{j+1}}} 
\] 
for $j \in \{ 1, \ldots, n-1  \}$, and for the last few messages we have 
\begin{eqnarray*}
\inner{W_{n} \otimes Z_{n}}{\tr_{A_{n}} \kb{\xi_{n}}} 
& \geq & 
\inner{W_{n} \otimes Z_{n+1}}{\tr_{B_{n}} \kb{\xi'_{n}}} \\
& \geq & 
\inner{W_{n+1} \otimes  Z_{n+1}}{\tr_{A'_0 \times A'} \kb{\xi_{n+1}}} \\
& \geq & 
\inner{W_{n+1} \otimes \Pi_{\B,0}}{\tr_{B'_0 \times B'} \kb{\xi'_{n+1}}} \\
& \geq & 
\inner{\Pi_{\A,1} \otimes \Pi_{\B,0}}{\kb{\xi_{n+2}}} 
\end{eqnarray*}
and the last quantity equals $0$ since Alice and Bob never output different outcomes when they are both honest. 
Note that these are dual variables from the original cheating SDPs, not the reduced version. The dual variables for the reduced version are scaled eigenvalues of the corresponding dual variables above. However, we do reconstruct Kitaev's proof above using the reduced SDPs in Section~\ref{LB}.

As was done in~\cite{Moc07}, we use the function $\prob : \pos^A \times \pos^B \times \pos^{A \times B} \to \R$,  defined as 
\[ \Prob(X,Y, \sigma) := \sum_{\lambda \in \eig(X)} \sum_{\mu \in \eig(Y)} \inner{\Pi_{X}^{[\lambda]} \otimes \Pi_{Y}^{[\mu]}}{\sigma} \, \pg{\lambda}{\mu}, \]
where $\pg{\lambda}{\mu} : \R^2 \to \{ 0, 1 \}$ denotes the function that takes value $1$ on input $(\lambda, \mu)$ and $0$ otherwise. Note this function has \emph{finite support} which are the \emph{points} in the point game. The quantity 
\[ \inner{\Pi_{X}^{[\lambda]} \otimes \Pi_{Y}^{[\mu]}}{\sigma} \] 
is said to be the associated \emph{probability} of the point $\pg{\lambda}{\mu}$. 

To create a point game for a $\BCCF$-protocol, we use the points that arise from feasible dual solutions in the following way:
\[ \begin{array}{rcll}
p_0 & := & \Prob(\Pi_{\A,1}, \Pi_{\B,0}, \kb{\xi_{n+2}}), \\
p'_1 & := & \Prob(W_{n+1}, \Pi_{\B,0}, \tr_{B'_0 \times B'} \kb{\xi'_{n+1}}), \\
p_1 & := & \Prob(W_{n+1}, Z_{n+1}, \tr_{A'_0 \times A'} \kb{\xi'_{n+1}}), \\
p'_{(n+2)-j} & := & \Prob(W_j, Z_{j+1}, \tr_{B_j} \kb{\xi'_{j}}), & \forAll j \in \{ 1, \ldots, n \}, \\
p_{(n+2)-j} & := & \Prob(W_j, Z_j, \tr_{A_j} \kb{\xi_{j}}), & \forAll j \in \{ 1, \ldots, n \},
\end{array} \]
noting that the $i$'th point corresponds to the $i$'th last message in the protocol. This gives rise to the point game moves (or transitions):
\[ p_0 \to p'_1 \to p_1 \to \cdots \to p'_j \to p_j \to \cdots \to p'_{n+1} \to p_{n+1}, \]
which we give context to in the next subsection. The reason we define point games in reverse time order is so that they always have the same starting state and it is shown later that the final point captures the two objective function values of the corresponding dual feasible solutions. The reverse time order ensures that we always start with the same $p_0$ and aim to get a desirable last point, instead of the other way around. 

First, we calculate $\Prob(W_j, Z_j, \tr_{A_j} \kb{\xi_{j}})$, for $j \in \set{1, \ldots, n}$.

\begin{definition} 
For a string $z \in \{ 0,1 \}^*$, we define $p(z)$ as the probability of string $z$ being revealed during an honest run of a fixed $\BCCF$-protocol. 
\end{definition}

To capture these probabilities, we use the following (unnormalized) states defined from the honest states in a $\BCCF$-protocol. 

\begin{definition}  
For $x = (x_1, \ldots, x_n) \in A$, $y = (y_1, \ldots, y_n) \in B$, and $j \in \{ 1, \ldots, n \}$, define 
\[ \ket{\psi_{x_1, \ldots, x_j}} := \frac{1}{\sqrt 2} \sum_{x_{j+1} \in A_{j+1}} \cdots \sum_{x_n \in A_n} \sum_{a \in \zo} \sqrt{\alpha_{a,x}} \, \ket{aa} \ket{x_{j+1}, \ldots, x_n}\ket{x_{j+1}, \ldots, x_n} \]
and 
\[ \ket{\phi_{y_1, \ldots, y_j}} := \frac{1}{\sqrt 2} \sum_{y_{j+1} \in B_{j+1}} \cdots \sum_{y_n \in B_n} \sum_{b \in \zo} \sqrt{\beta_{b,y}} \, \ket{bb} \ket{y_{j+1}, \ldots, y_n}\ket{y_{j+1}, \ldots, y_n}. \]
\end{definition} 

Note we have $p(x_1, \ldots, x_j) = \braket{\psi_{x_1, \ldots, x_j}}{\psi_{x_1, \ldots, x_j}}$, for all $(x_1, \ldots, x_j) \in A_1 \times \cdots \times A_j$, and $p(y_1, \ldots, y_j) = \braket{\phi_{y_1, \ldots, y_j}}{\phi_{y_1, \ldots, y_j}}$, for all $(y_1, \ldots, y_j) \in B_1 \times \cdots \times B_j$, for $j \in \{ 1, \ldots, n \}$. 

From the proof of the reduced problems {in Appendix~\ref{app:reduced_app}}, we can assume an optimal choice of $W_j$ has  eigenvalues $\frac{w_{j, x_1, y_1, \ldots, y_{j-1}, x_j}}{p(x_1, \ldots, x_j)}$, where $w_j$ is the corresponding variable in the dual of Bob's reduced cheating SDP. {Note that we do not need to worry about the case when ${p(x_1, \ldots, x_j) = 0}$ (nor the division by $0$) since this implies $w_{j, x_1, y_1, \ldots, y_{j-1}, x_j} = 0$.} The same argument holds in the following cases whenever there is an issue of dividing by $0$. The positive eigenvalues have  respective eigenspace {projections} 
\[ \Pi_{W_j}^{[x_1, y_1, \ldots, y_{j-1}, x_j]} := \kb{x_1, y_1, \ldots, y_{j-1}, x_j} \otimes \kb{\tilde{\psi}_{x_1, \ldots , x_j}}, \]
where $\ket{\tilde{\psi}_{x_1, \ldots , x_j}}$ is $\ket{{\psi}_{x_1, \ldots , x_j}}$ normalized. The other eigenvalues do not contribute to the points (this can be verified since these eigenvalues already contribute to probabilities adding to $1$). Similarly, an optimal choice of $Z_j$ has eigenvalues $\frac{z_{j, x_1, y_1, \ldots, x_{j-1}, y_{j-1}}}{p(y_1, \ldots, y_{j-1})}$, where $z_j$ is the corresponding variable in the dual of Alice's reduced cheating SDP, with respective eigenspaces
\[ \Pi_{Z_j}^{[x_1, y_1, \ldots, x_{j-1}, y_{j-1}]} := \kb{x_1, y_1, \ldots, x_{j-1}, y_{j-1}} \otimes \kb{\tilde{\phi}_{y_1, \ldots , y_{j-1}}}, \]
where $\ket{\tilde{\phi}_{y_1, \ldots , y_{j-1}}}$ is $\ket{{\phi}_{y_1, \ldots , y_{j-1}}}$ normalized. From these eigenspaces, we can compute
\begin{eqnarray*}
& & \inner{\Pi_{W_j}^{[x'_1, y'_1, \ldots, y'_{j-1}, x'_j]} \otimes \Pi_{Z_{j}}^{[x_1, y_1, \ldots, x_{j-1}, y_{j-1}]}}{\tr_{A_j} \kb{\xi_{j}}} \\
& = & \delta_{x_1, x'_1} \cdots \delta_{x_{j-1}, x'_{j-1}} \delta_{y_1, y'_1} \cdots \delta_{y_{j-1}, y'_{j-1}} \, p(x_1, y_1, \ldots, y_{j-1}, x_j).
\end{eqnarray*}
Thus, we can write the point $p_{(n+2)-j} := \Prob(W_j, Z_j, \tr_{A_j} \kb{\xi_{j}})$ as
\[ \sum_{x_1 \in A_1} \sum_{y_1 \in B_1} \cdots \sum_{y_{j-1} \in B_{j-1}} \sum_{x_j \in A_j} p(x_1, y_1, \ldots, y_{j-1}, x_j) \, \pg{\frac{w_{j, x_1, y_1, \ldots, y_{j-1}, x_j}}{p(x_1, \ldots, x_j)}}{\frac{z_{j, x_1, y_1, \ldots, x_{j-1}, y_{j-1}}}{p(y_1, \ldots, y_{j-1})}}. \]
We can similarly write $p'_{(n+2)-j} := \Prob(W_j, Z_{j+1}, \tr_{B_j} \kb{\xi'_{j}})$ as
\[ \sum_{x_1 \in A_1} \sum_{y_1 \in B_1} \cdots \sum_{y_{j} \in B_j} \sum_{x_j \in A_j} p(x_1, y_1, \ldots, x_j, y_j) \, \pg{\frac{w_{j, x_1, y_1, \ldots, y_{j-1}, x_j}}{p(x_1, \ldots, x_j)}}{\frac{z_{j+1, x_1, y_1, \ldots, x_{j}, y_{j}}}{p(y_1, \ldots, y_{j})}}. \]
 
The first three points are different from above as they correspond to the last few messages in the protocol (which are quite different from the first $2n$ messages). Nonetheless, the process is the same and we can calculate them to be
\begin{eqnarray*} 
p_{1} & = & \sum_{a \in \zo} \sum_{x \in A} \sum_{y \in B} p(x,a) \, p(y) \pg{v_{a,y}}{\frac{z_{n+1, x,y}}{p(y)}}, \\
p'_{1} & = & \sum_{b \in \zo} \sum_{y \in B} \half p(y, \bar{b}) \pg{v_{b,y}}{0} + \sum_{b,y} \half p(y, {b}) \pg{v_{b,y}}{1}, \\
p_{0} & = & \half \, \pg{1}{0} + \half \, \pg{0}{1},
\end{eqnarray*}
noting $z_{n+1, x,y} > 0$ when $p(y) > 0$.

We call any point game constructed from dual feasible solutions in this manner a \emph{$\BCCF$-point game}. In the next subsection, we describe rules for moving from one point to the next in any $\BCCF$-point game yielding a protocol independent definition.
 
\subsection{Describing $\BCCF$-point games using basic moves} \label{basicmoves}

Below are some basic point moves (or transitions) as Mochon describes them in $\cite{Moc07}$. 
{
\begin{definition}[Basic moves]
\quad
\begin{eqnarray*}
\textit{Point raising:} & & q \pg{w}{z} \to q \pg{w}{z'}, \text{ for } z \leq z', \\ 
\textit{Point merging:} & & q_1 \pg{w}{z_1} + q_2 \pg{w}{z_2} \to (q_1 + q_2) \pg{w}{\frac{q_1 z_1 + q_2 z_2}{q_1 + q_2}}, \\ 
\textit{Point splitting:} & & (q_1 + q_2) \pg{w}{\frac{q_1 + q_2}{\left( \frac{q_1}{z_1} \right) + \left( \frac{q_2}{z_2} \right)}} \to q_1 \pg{w}{z_1} + q_2 \pg{w}{z_2}, \text{ for } z_1, z_2 \neq 0. 
\end{eqnarray*} 
\end{definition}
}
An example of point splitting and point raising can be seen in Figure~\ref{QPG1} and examples of point mergings can be seen in Figures~\ref{QPG2}~and~\ref{QPG3}. Using a slight abuse of the definition of point splitting, if we perform a point split then raise the points, we still refer to this as a point split (for reasons that will be clear later). Also, we can merge or split on more than two points by repeating the process two points at a time.

These are moves in the second coordinate (keeping the first coordinate fixed) called \emph{vertical moves}, and we similarly define \emph{horizontal moves} acting on the first coordinate (keeping the second coordinate fixed).

Mochon gives a rough interpretation of these moves in~\cite{Moc07}. We can think of point raising as receiving a message, point merging as generating a message, and point splitting as checking a message via quantum measurement. These interpretations apply to the family of weak coin-flipping protocols in~\cite{Moc05}, and we show they also  apply to $\BCCF$-protocols.

Below are some special cases of these moves which are useful when describing $\BCCF$-point games. 
{\begin{eqnarray*}
\textit{Probability splitting:} & & (q_1 + q_2) \pg{z}{w} \to q_1 \pg{z}{w} + q_2 \pg{z}{w}, \\ 
\textit{Probability merging:} & & q_1 \pg{z}{w} + q_2 \pg{z}{w} \to (q_1 + q_2) \pg{z}{w}, \\ 
\textit{Aligning:} & & q_1 \pg{z_1}{w_1} + q_2 \pg{z_2}{w_2} \to q_1 \pg{\max \{ z_1, z_2 \}}{w_1} + q_2 \pg{\max \{ z_1, z_2 \}}{w_2}. 
\end{eqnarray*}} 
Probability splitting is the special case of point splitting where the resulting points have the same value and probability merging is the special case of point merging where the resulting points have the same value. Aligning is just raising two points to the maximum of the two (usually so a merge can be performed on the other coordinate).

We now show that each move in a $\BCCF$-point game can be described using basic moves. Consider the first transition:
\[ \half \pg{1}{0} + \half \pg{0}{1} \to \sum_{b \in \zo} \sum_{y \in B} \half p(y, \bar{b}) \pg{v_{b,y}}{0} + \sum_{b \in \zo} \sum_{y \in B} \half p(y, {b}) \pg{v_{b,y}}{1}, \]
which can be described in two steps. First,
\[ \half \pg{0}{1} \to \sum_{b \in \zo} \sum_{y \in B} \half p(y, {b}) \pg{v_{b,y}}{1}, \]
is just probability splitting followed by point raising (in the first coordinate). The {transition} 
\[ \half \pg{1}{0} \to \sum_{b \in \zo} \sum_{y \in B} \half p(y, \bar{b}) \pg{v_{b,y}}{0}, \]
is a point splitting. To see this, recall the dual constraint 
$\Diag(v_{a}) \succeq \sqrtt{\beta_{\bar{a}}}$, for $a \in \zo$. 
We have seen that this is equivalent to the condition 
$\sum_{y \in B} \frac{\beta_{\bar{a}, y}}{v_{a,y}} \leq 1$, 
when $v_a > 0$, which is the condition for a point split. Technically, a point split would have this inequality satisfied with equality, but we can always raise the points such that we get an inequality. As explained earlier, we just call this a point split.

We can interpret the point raise as Alice accepting Bob's last message $b$, and the point split as Alice checking Bob's state at the end of the protocol using her measurement. Note that these are the last two actions of a $\BCCF$-protocol.

We can do something similar for the second transition below 
\[ \sum_{b} \sum_{y} \frac{p(y, \bar{b})}{2} \pg{v_{b,y}}{0} + \!\! \sum_{b} \sum_{y} \frac{p(y, {b})}{2} \pg{v_{b,y}}{1}
\to \!\!
\sum_{a} \sum_{x} \sum_{y} p(x,a) \, p(y)  \pg{v_{a,y}}{\frac{z_{n+1, x,y}}{p(y)}}. \] 
To get this, for every $b \in \zo$, $y \in \supp(\beta_{b})$, we point split
\[ \pg{v_{b,y}}{1} \to \sum_{x \in A} \alpha_{b,x} \pg{v_{b,y}}{\frac{2 z_{n+1, x,y}}{\beta_{b,y}}}. \]
This is a valid point split since we have the dual constraint $\Diag \left( \frac{2 z_{n+1}^{(y)}}{\beta_{b,y}} \right) \succeq \sqrtt{\alpha_{b}}$, for all $b \in \zo, y \in \supp(\beta_b)$. 
Note that $z_{n+1}$ does not depend on $b$ so there are some consistency requirements when performing these splits.

The points at this stage can be seen in Figure~\ref{QPG1} for the special case of a four-round, i.e., $n=1$, $\BCCF$-protocol with $|A| = |B| = 2$ (noting that $p(y,b) = \half \beta_{b,y}$).
 
\begin{figure}[ht]
  \centering
   \includegraphics[width=6.5in]{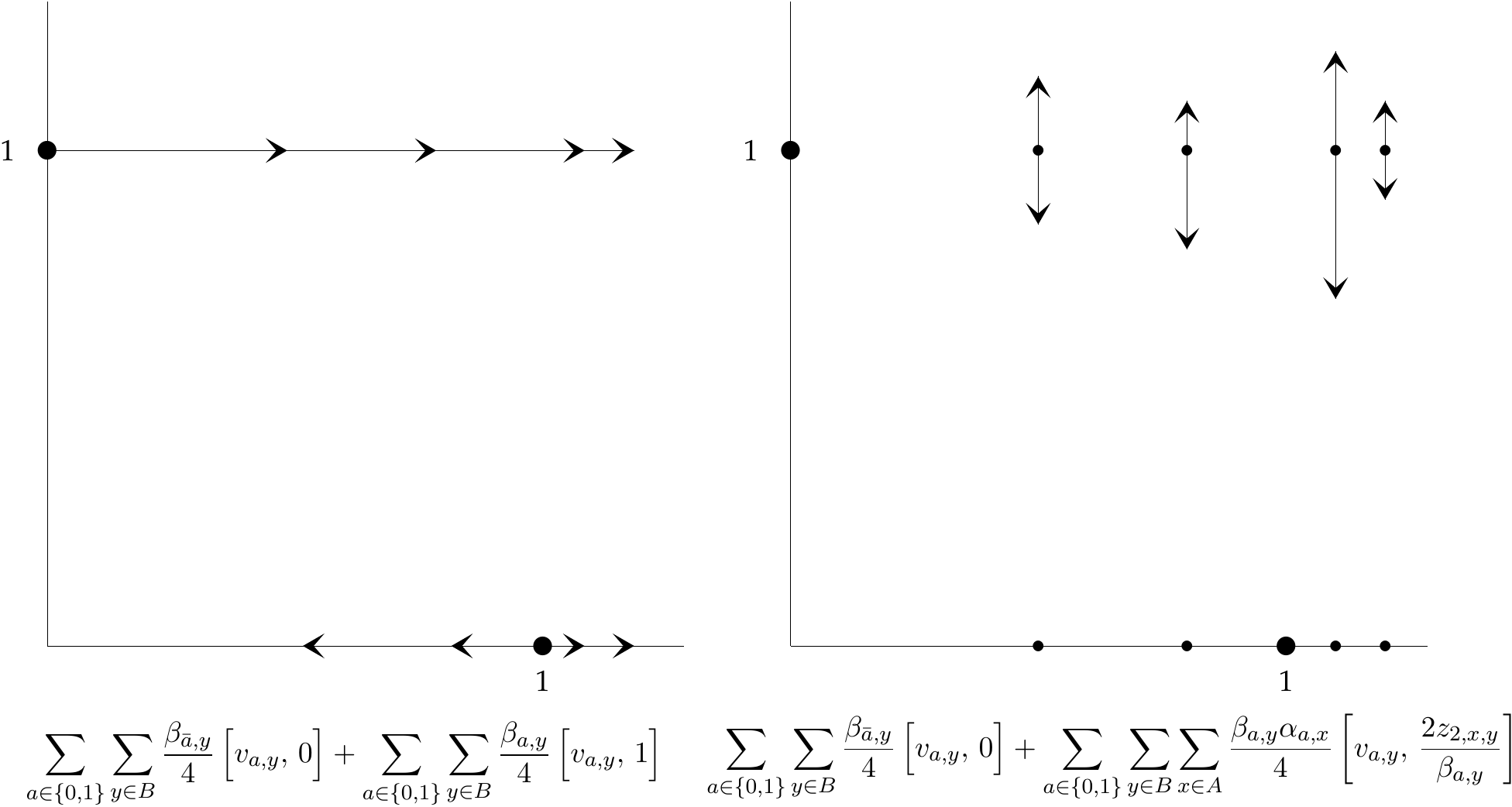} 
  \caption{After the point splits in a $\BCCF$-point game.}
\label{QPG1}
\end{figure}

For the other points, we perform the probability splitting:
\[ \sum_{b \in \zo} \sum_{y \in B} \half \, p(y, \bar{b}) \pg{v_{b,y}}{0} \to \sum_{b \in \zo} \sum_{y \in B} \sum_{x \in A} \half \,  p(y, \bar{b}) \, \alpha_{b,x} \pg{v_{b,y}}{0}, \]
yielding the current state
\[ \sum_{b \in \zo} \sum_{y \in B} \sum_{x \in A} \half \,  \alpha_{b,x} \left( p(y, \bar{b}) \pg{v_{b,y}}{0} + p(y, {b}) \pg{v_{b,y}}{\frac{2 z_{n+1, x,y}}{\beta_{b,y}}} \right). \]
Merging the part in the brackets yields
\[ \sum_{b \in \zo} \sum_{y \in B} \sum_{x \in A} \half \alpha_{b,x} p(y) \pg{v_{b,y}}{\frac{z_{n+1, x,y}}{p(y)}}
= \sum_{a \in \zo} \sum_{y \in B} \sum_{x \in A} p(x,a) \, p(y)  \pg{v_{a,y}}{\frac{z_{n+1, x,y}}{p(y)}}, \]
where the quantity on the right just relabelled $b$ as $a$. The transitions here were point splitting, point merging, and point raising (from the dual constraint on $z_{x,y}$, we can think of it as being a maximum over $a$, corresponding to a raise). These correspond to Bob checking Alice, Bob generating $b$, and Bob receiving $a$, respectively.

Fortunately, the rest of the transitions are straightforward. To explain the transition
\[ \sum_{a \in \zo} \sum_{y \in B} \sum_{x \in A} p(x,a) \, p(y)  \pg{v_{a,y}}{\frac{z_{n+1, x,y}}{p(y)}}
\to
\sum_{y \in B} \sum_{x \in A} p(x,y) \pg{\frac{w_{n, x_1, y_1, \ldots, y_{n-1}, x_n}}{p(x)}}{\frac{z_{n+1, x,y}}{p(y)}},
\] 
all we do is merge $a$, then align $y_n \in B_{n}$ in the first coordinate. To see why this is valid, we have the dual constraint $w_{n, x_1, y_1, \ldots, y_{n-1}, x_n} \geq \sum_{a \in \zo} \half \alpha_{a,x} \, v_{a,y} = \sum_{a \in \zo} p(x,a) \, v_{a,y}$. 
This corresponds to Alice generating $a$ and receiving Bob's message $y_n \in B_{n}$. This is depicted in Figure~\ref{QPG2}.
  
\begin{figure}[ht]
  \centering
   \includegraphics[width=6.5in]{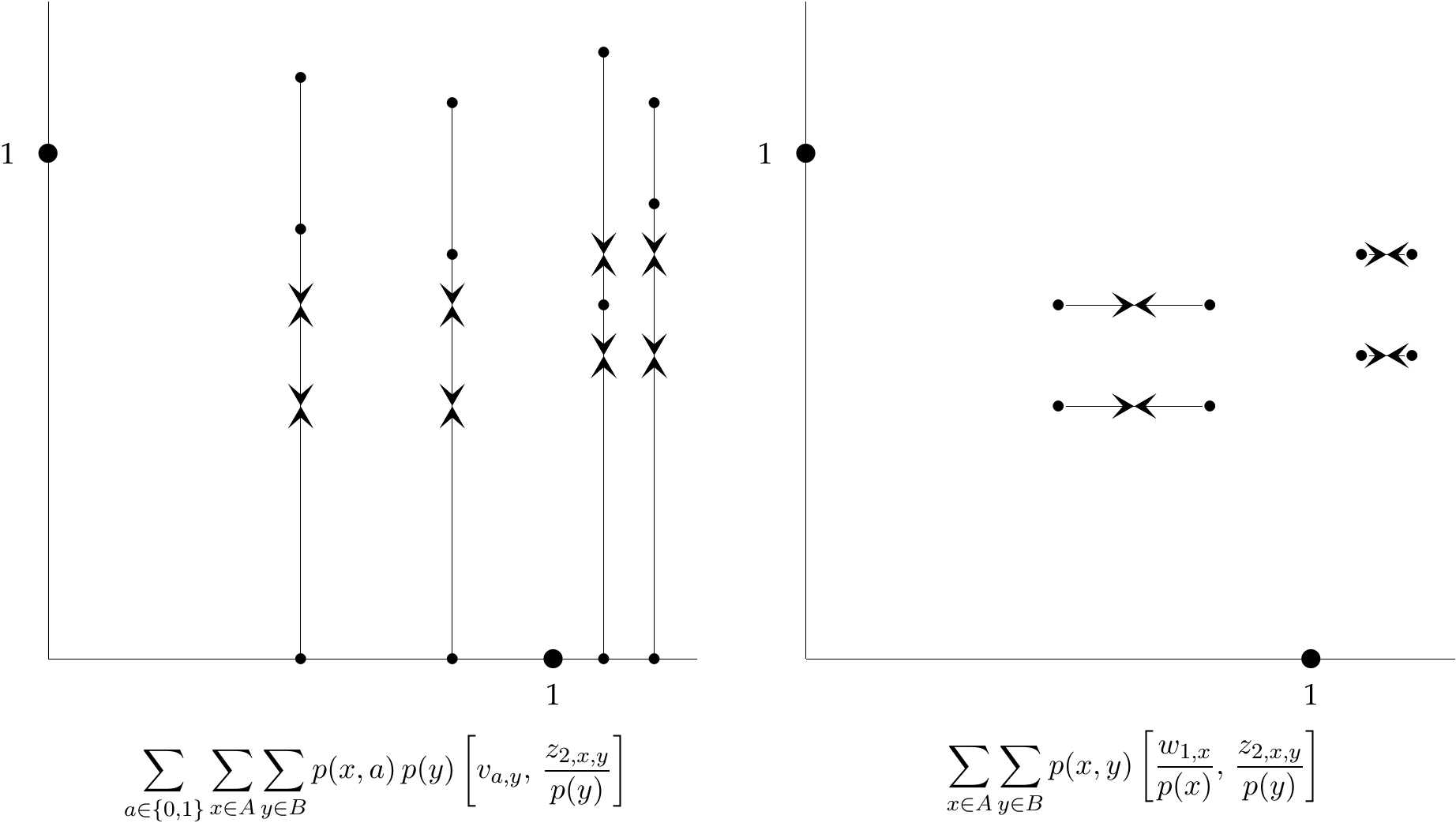} 
  \caption{After the first two merges in a $\BCCF$-point game.}
\label{QPG2}
\end{figure}

We show one more transition and the rest follow similarly. To show the transition
\begin{eqnarray*}
& & \sum_{y} \sum_{x} p(x,y) \pg{\frac{w_{n, x_1, y_1, \ldots, y_{n-1}, x_n}}{p(x)}}{\frac{z_{n+1, x,y}}{p(y)}} \\
& \to &
\sum_{y_1} \cdots \sum_{y_{n-1}} \sum_{x} p(x_1, y_1, \ldots, y_{n-1}, x_n) \pg{\frac{w_{n, x_1, y_1, \ldots, y_{n-1}, x_n}}{p(x)}}{\frac{z_{n,x_1, y_1, \ldots, x_{n-1}, y_{n-1}}}{p(y_1, \ldots, y_{n-1})}},
\end{eqnarray*}
we merge on $y_n \in B_n$ then align $x_n \in A_n$ in the second coordinate. The dual constraint corresponding to this is
$z_{n,x_1, y_1, \ldots, x_{n-1}, y_{n-1}} \geq \sum_{y_n \in B_n} z_{n+1, x,y}$. 
We can continue in this fashion until we get to the last {points} 
\[ \sum_{x_1 \in A_1} p(x_1) \pg{\frac{w_{1, x_1}}{p(x_1)}}{z_1}, \]
where $z_1$ is Alice's dual objective function value.
If we merge on $x_1$, we get Bob's dual objective function value in the first coordinate 
\[ \pg{\dsum_{x_1 \in A_1} w_{1, x_1}}{z_1}. \]  
Therefore, if $(w_1, \ldots, w_n, v_0, v_1)$ is feasible for the dual of Bob's reduced cheating SDP and if 
$(z_1, \ldots, z_n, z_{n+1})$ is feasible for the dual of Alice's reduced cheating SDP, then the final point of the point game is comprised of the two dual objective function values, as seen in Figure~\ref{QPG3}. 

\begin{figure}[ht] 
  \centering
   \includegraphics[width=6.5in]{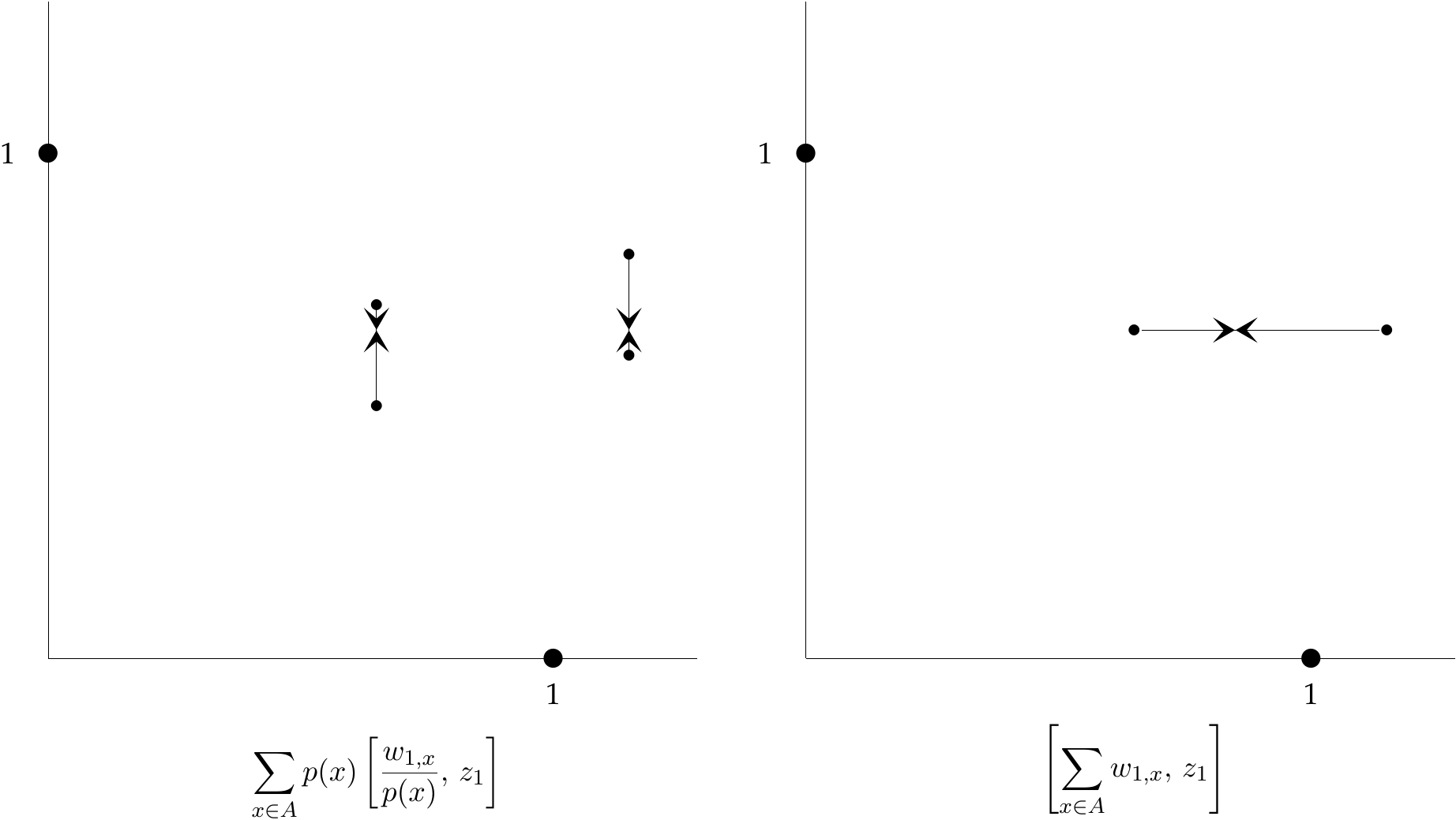} 
  \caption{The last few moves of a $\BCCF$-point game.}
\label{QPG3}
\end{figure}

We summarize this entire process as a list of basic moves in Point Game~\ref{PGBM}. 

Therefore, an optimal assignment of variables in the duals of the reduced cheating SDPs corresponds to a minimal {final point in the point game}. We now argue that these duals attain an optimal solution. Since the optimal objective values are bounded above by $1$, we can upper bound the values on all of the variables in the duals accordingly (it can be shown that $v_{a,y} \leq 2 |A|$, for all $a \in \zo, y \in B$ and the rest of the variables in the four duals are bounded above by $1$). Also, they are bounded below by $0$ from the positive semidefiniteness constraints. Since we are optimizing a continuous function over a compact set, we have that an optimal solution exists. 
  
\newpage 
\begin{pointgame}[$\BCCF$-point game with final point $\pg{\zeta_{\B}}{\zeta_{\A}}$ from basic moves] 
\label{PGBM}
\[ \begin{array}{lll}
& \dfrac{1}{2} \pg{1}{0} + \dfrac{1}{2} \pg{0}{1} & \\
\to &
\dsum_{a \in \zo} \quarter \, \pg{1}{0} + \dsum_{a \in \zo} \sum_{y \in B} \quarter \beta_{{a},y} \, \pg{0}{1}
& \textup{prob. splitting} \\
\to &
\dsum_{a \in \zo} \sum_{y \in B} \quarter \beta_{\bar{a},y} \, \pg{v_{a,y}}{0} + \dsum_{a \in \zo} \sum_{y \in B} \quarter \beta_{{a},y} \, \pg{0}{1} & \textup{point splitting} \\
\to &
\dsum_{a \in \zo} \sum_{y \in B} \quarter \beta_{\bar{a},y} \, \pg{v_{a,y}}{0} + \dsum_{a \in \zo} \sum_{y \in B} \quarter \beta_{{a},y} \, \pg{v_{a,y}}{1} & \textup{point raises} \\
\to &
\dsum_{a \in \zo} \sum_{y \in B} \quarter \beta_{\bar{a},y} \, \pg{v_{a,y}}{0} + \dsum_{a \in \zo} \sum_{y \in B} \sum_{x \in A} \quarter \beta_{{a},y} \alpha_{a,x} \, \pg{v_{a,y}}{\frac{2 z_{n+1, x,y}}{\beta_{a,y}}} & \textup{point splitting} \\
\to &
\dsum_{a \in \zo} \sum_{y \in B} \sum_{x \in A} \left( \quarter \beta_{\bar{a},y} \alpha_{a,x} \, \pg{v_{a,y}}{0} + \quarter \beta_{{a},y} \alpha_{a,x} \, \pg{v_{a,y}}{\frac{2 z_{n+1, x,y}}{\beta_{a,y}}} \right) & \textup{prob. splitting} \\
= &
\dsum_{a \in \zo} \sum_{y \in B} \sum_{x \in A} \quarter \alpha_{a,x} \left( \beta_{\bar{a}, y} \, \pg{v_{a,y}}{0} + \beta_{a,y} \, \pg{v_{a,y}}{\frac{2 z_{n+1, x,y}}{\beta_{a,y}}} \right) & \textup{} \\
\to &
\dsum_{a \in \zo} \sum_{y \in B} \dsum_{x \in A} \quarter \alpha_{a,x} \left( \dsum_{b \in \zo} \beta_{b,y} \right) \pg{v_{a,y}}{\frac{z_{n+1, x,y}}{p(y)}} & \textup{merges} \\
= &
\dsum_{a \in \zo} \dsum_{y \in B} \dsum_{x \in A} p(x,a) \, p(y)  \pg{v_{a,y}}{\frac{z_{n+1, x,y}}{p(y)}} & \textup{} \\
\to &
\dsum_{y \in B} \dsum_{x \in A} p(x,y) \pg{\frac{w_{n, x_1, y_1, \ldots, y_{n-1}, x_n}}{p(x)}}{\frac{z_{n+1, x,y}}{p(y)}} & \textup{merge $a$,} \\
& & \textup{then align $y_n$} \\
\quad \\
\to & \dsum_{y_1, \ldots, y_{n-1}} \dsum_{x \in A} p(x, y_1, \ldots, y_{n-1}) \pg{\frac{w_{n, x_1, y_1, \ldots, y_{n-1}, x_n}}{p(x)}}{\frac{z_{n, x_1, y_1, \ldots, x_{n-1}, y_{n-1}}}{p(y_1, \ldots, y_{n-1})}} & \textup{merge $y_n$,} \\
& & \textup{then align $x_n$} \\
& \quad \vdots \\
\to &
\dsum_{x_1 \in A_1} p(x_1) \pg{\frac{w_{1, x_1}}{p(x_1)}}{\zeta_{\A}} & \textup{merge $y_1$,} \\
& & \textup{then align $x_1$} \\
\to &
1 \pg{\zeta_{\B}}{\zeta_{\A}} & \textup{merge $x_1$}. \\
\end{array} \]
\end{pointgame}
 
{An example of an (optimal) $\BCCF$-point game can be found in Appendix~\ref{PGexample} corresponding to a four-round $\BCCF$-protocol with all four cheating probabilities equal to $3/4$. Note this four-round protocol is equivalent to a three-round protocol in~\cite{KN04} where we have set $\alpha_0 = \alpha_1$ to make the first message ``empty''.}
 
\subsection{Point game analysis} 
\label{PGanalysis}

From the point game description, we see that the only freedom is in how we choose the point splits since the rest of the points are determined from the merges and aligns. We {expand on} this idea {when developing} the succinct {forms} of the duals of the reduced SDPs in Subsection~\ref{succinct}. In each of the succinct forms of these duals, the only freedom is in how we choose to satisfy the positive semidefiniteness constraints. Once these variables {are} fixed, there {is} an obvious way to choose an optimal assignment of the rest of the variables. Coincidentally, the last constraints in each dual correspond to the point splits in the point game.

This brings us to the following protocol independent definition of $\BCCF$-point games.
 
\begin{definition}[$\BCCF$-point game (protocol independent)]
A $\BCCF$-point game defined on the parameters $\alpha_0, \alpha_1 \in \prob^A$ and $\beta_0, \beta_1 \in \prob^B$, with final point $\pg{\zeta_{\B}}{\zeta_{\A}}$, is any point game of the form
\[ p_0 := \half \pg{1}{0} + \half \pg{0}{1} \to p_1 \to p_2 \to \cdots \to p_m := \pg{\zeta_{\B}}{\zeta_{\A}}, \]
where the transitions are exactly the basic moves as described in {Point Game~\ref{PGBM}}. 
\end{definition}

As mentioned above, one only has the freedom to choose how the points are split at the beginning, the rest of the points are determined. Thus, every choice of point splitting yields a potentially different point game (keeping $\alpha_0, \alpha_1 \in \prob^A$ and $\beta_0, \beta_1 \in \prob^B$ fixed). A $\BCCF$-point game is defined by $\alpha_0, \alpha_1 \in \prob^A$ and $\beta_0, \beta_1 \in \prob^B$ which are the same parameters that uniquely define a $\BCCF$-protocol. However, there could be many point games corresponding to these same parameters. The analogous concept for $\BCCF$-protocols is that there could be many cheating strategies for the same protocol. Of course, there is an optimal cheating strategy just as there is an optimal $\BCCF$-point game.

The above definition is protocol independent since we have defined starting points, an ending point, and a description of how to move the points around. Indeed, the ``rules'' for the point moves correspond exactly to dual feasible solutions with objective function values being the two coordinates of the final point. This yields the following lemma which is the application of weak and strong duality in the language of protocols and point games.

\begin{lemma} \label{WCFLEMMA}
Suppose $\pg{\zeta_{\B,1}}{\zeta_{\A,0}}$ is the final point of a $\BCCF$-point game defined on the parameters ${\alpha_0, \alpha_1 \in \prob^A}$ and $\beta_0, \beta_1 \in \prob^B$. Then
\[ P_{\B,1}^* \leq \zeta_{\B,1} \aand P_{\A,0}^* \leq \zeta_{\A,0}, \]
where $P_{\B,1}^*$ and $P_{\A,0}^*$ are the optimal cheating probabilities for Bob forcing $1$ and Alice forcing $0$, respectively, in the corresponding $\BCCF$-protocol.
Moreover, there exists a $\BCCF$-point game with final point $\pg{P_{\B,1}^*}{P_{\A,0}^*}$.
\end{lemma}

In this paper, we are concerned with bounding the bias of strong coin-flipping protocols, and therefore would like to bound all four cheating probabilities. Recall that Alice and Bob's two cheating probabilities are swapped when $\beta_0$ and $\beta_1$ are swapped. This motivates the following definition.

\begin{definition}[$\BCCF$-point game pair]
Suppose we have a $\BCCF$-point game defined on the parameters $\alpha_0, \alpha_1 \in \prob^A$ and $\beta_0, \beta_1 \in \prob^B$ with final point $\pg{\zeta_{\B, 1}}{\zeta_{\A, 0}}$. Also, suppose we have another $\BCCF$-point game defined by the parameters $\alpha_0, \alpha_1 \in \prob^A$ and $\beta'_0 {:=} \beta_1$, $\beta'_1 {:=} \beta_0 \in \prob^B$ with final point $\pg{\zeta_{\B, 0}}{\zeta_{\A, 1}}$. We call the two point games a $\BCCF$-point game pair, defined by the parameters $\alpha_0, \alpha_1 \in \prob^A$ and $\beta_0, \beta_1 \in \prob^B$, with final point $\pgf{\zeta_{\B,0}}{\zeta_{\B,1}}{\zeta_{\A,0}}{\zeta_{\A,1}}$.
\end{definition}
It is worth commenting that $\BCCF$-point game pairs are defined over certain parameters even though one of the point games in the pair is  defined over swapped parameters.

Using Lemma~\ref{WCFLEMMA}, we have the following theorem. 
\begin{theorem} \label{SCFTHEOREM}
Suppose $\pgf{\zeta_{\B,0}}{\zeta_{\B,1}}{\zeta_{\A,0}}{\zeta_{\A,1}}$ is the final point of a $\BCCF$-point game pair defined on the parameters $\alpha_0, \alpha_1 \in \prob^A$ and $\beta_0, \beta_1 \in \prob^B$. Then
\[ P_{\B,0}^* \leq \zeta_{\B,0}, \quad P_{\B,1}^* \leq \zeta_{\B,1}, \quad P_{\A,0}^* \leq \zeta_{\A,0}, \; \textup{ and } \; P_{\A,1}^* \leq \zeta_{\A,1}, \]
where $P_{\B,0}^*$, $P_{\B,1}^*$, $P_{\A,0}^*$, $P_{\A,1}^*$ are the optimal cheating probabilities for  the corresponding $\BCCF$-protocol.
Moreover, there exists a $\BCCF$-point game pair with final point $\pgf{P_{\B,0}^*}{P_{\B,1}^*}{P_{\A,0}^*}{P_{\A,1}^*}$.
\end{theorem}
 
\section{A related family of classical coin-flipping protocols} \label{class}

In this section, we describe a family of classical protocols which is the classical counterpart to quantum $\BCCF$-protocols. That is, we choose messages according to the underlying probability distributions (instead of in a superposition) and we have a modified cheat detection step at the end of the protocol.
 
We now describe the protocol.

\begin{protocol}[Classical $\BCCF$-protocol]
\quad
\begin{itemize}
\item Alice chooses $a \in A_0$ uniformly at random and samples $x \in A$ with probability $\alpha_{a,x}$.
\item Bob chooses $b \in B_0$ uniformly at random and samples $y \in B$ with probability $\beta_{a,y}$.
\item For $i$ from $1$ to $n$: Alice sends $x_i \in A_i$ to Bob who replies with $y_i \in B_i$.
\item Alice fully reveals her bit by sending $a \in A_0$ to Bob. If $x \not\in \supp(\alpha_a)$, Bob aborts.
\item Bob fully reveals his bit by sending $b \in B_0$ to Alice. If $y \not\in \supp(\beta_b)$, Alice aborts.
\item The outcome of the protocol is $a \oplus b$, if no one aborts.
\end{itemize}
\end{protocol}
 
The rest of this section illustrates the connections between this classical protocol and the quantum version.

\subsection{Formulating optimal classical cheating strategies as linear programs}

{We can similarly formulate optimal cheating strategies in the classical protocols as optimization problems. In this case, we use linear programming as shown in the lemma below.} 
 
\begin{lemma}
For the classical $\BCCF$-protocol defined by the parameters $\alpha_0, \alpha_1 \in \prob^A$, ${\beta_0, \beta_1 \in \prob^B}$, we can write the  cheating probabilities {for Alice and Bob, each forcing outcome $0$, as} 
\[ P_{\A,0}^* = \max \left\{ 
\half \sum_{a \in A'_0} \sum_{y \in B} \sum_{x \in \supp(\alpha_a)} \beta_{a,y} s_{a,x,y}
: (s_{1}, \ldots, s_{n}, s) \in \calP_{\A} \right\}, \]
and  
\[ P_{\B,0}^* = \max \left\{ 
\half \sum_{a \in A'_0} \sum_{y \in \supp(\beta_a)} \sum_{x \in A} \alpha_{a,x} \, p_{n,x,y}
: (p_{1}, \ldots, p_{n}) \in \calP_{\B} \right\}, \] 
respectively. 
{We obtain $P_{\A,1}^*$ and $P_{\B,1}^*$ by switching the roles of $\beta_0$ and $\beta_1$.} 
\end{lemma}

\begin{proof}
We shall prove this for the case of cheating Bob as the case for cheating Alice is almost identical. By examining Alice's cheat detection, we see that if we switch the roles of $\beta_0$ and $\beta_1$ then we also switch $P_{\B,0}^*$ and $P_{\B,1}^*$, so we only need to prove the $P_{\B,0}^*$ case.

After receiving the first message from Alice, Bob must choose a message to send. He can do this probabilistically by choosing $y_1 \in B_1$ with probability $p_{1, x_1, y_1}$, yielding the first constraint in Bob's cheating polytope. Notice that his message can depend on Alice's first message. We can similarly argue that the probabilities with which he chooses the rest of his messages are captured by the rest of the constraints in the cheating polytope with the exception of the last message. For the last message, we assume that Bob replies with $b=a$, where $a \in A_0$ was Alice's last message, if he desires outcome $0$ and $b = \bar{a}$ otherwise. Therefore, this decision is deterministic and is not represented by the cheating polytope.

All that remains is to explain the objective function. Since Bob chooses his last message deterministically, the quantity $\half \alpha_{a,x} \, p_{n,x,y}$ is the probability that Alice reveals $(x,a)$ and Bob reveals $(y,a)$. If he reveals $y$ when $\beta_{a,y} = 0$, he gets caught cheating, otherwise, his choice of $b$ is accepted. Therefore the objective function captures the total probability Alice accepts an outcome of $0$. \qed
\end{proof}

These are very similar to the quantum cheating probabilities except for the nonlinearity in the objective functions. For example, in the quantum setting, cheating Alice's objective function is $\dfrac{1}{2} \dsum_{a \bit} \dsum_{y \in B} \beta_{a,y} \;  \rF(s^{(a,y)}, \alpha_{a})$ and for the classical setting, it is
$\dfrac{1}{2} \dsum_{a \bit} \dsum_{y \in B} \beta_{a,y} \; \inner{s^{(a,y)}}{e_{\supp(\alpha_a)}}$, where $e_{\supp(\alpha_a)}$ is the $0,1$-vector taking value $1$ only on the support of $\alpha_a$.
We have a similar observation for Bob. What is surprising is that we can capture the communication for both settings with the same {polytope}.

To better understand this connection, we can write the objective function of Alice's reduced cheating SDP (for the quantum case) as
\[ \half \sum_{a \in \zo} \sum_{y \in B} \beta_{a,y} \, \inner{\sqrtt{s^{(a,y)}}}{\sqrtt{\alpha_a}}. \]
Then the objective function for Alice's cheating LP can be recovered if we replace $\sqrtt{\alpha_{a}}$ with $\Diag(e_{\supp(\alpha_a)})$.
Suppose we define a new projection
\[ \Pi_{\B,0} := \sum_{a \in \zo} \kb{a} \otimes \kb{a} \otimes \Diag(e_{\supp(\alpha_a)}) \otimes \id_{B'}. \]
A quick check shows that we can repeat the entire proof of the reduced cheating problems (in the quantum case) with this new projection if we also replace each occurrence of $\sqrtt{\alpha_a}$ with $\Diag(e_{\supp(\alpha_a)})$. Similar statements can be made if we redefine the other projections as
\[ \Pi_{\B,1} := \sum_{a \in \zo} \kb{\bar{a}} \otimes \kb{a} \otimes \Diag(e_{\supp(\alpha_a)}) \otimes \id_{B'}, \]
\[ \Pi_{\A,0} := \sum_{a \in \zo} \kb{{a}} \otimes \kb{a} \otimes \Diag(e_{\supp(\beta_a)}) \otimes \id_{A'}, \]
\[ \Pi_{\A,1} := \sum_{a \in \zo} \kb{\bar{a}} \otimes \kb{a} \otimes \Diag(e_{\supp(\beta_a)}) \otimes \id_{A'}. \]

This provides two insights. First, it proves that if we weaken the quantum cheat detection, we recover the optimal cheating probabilities for the corresponding classical protocol. Second, it gives us a recipe for developing the point games for the classical version. Notice that the eigenvalues of the dual variables are the same as in the quantum case, it is just that we have the stronger constraints:
\[ \begin{array}{ccc}
\Diag(v_{a}) \succeq \Diag(e_{\supp(\beta_{\bar{a}})}) & \textup{compared to} & \Diag(v_{a}) \succeq \sqrtt{\beta_{\bar{a}}}, \\
\Diag(z_{n+1}^{(y)}) \succeq \half \beta_{a,y} \Diag(e_{\supp(\alpha_a)}) & \textup{compared to} & \Diag(z_{n+1}^{(y)}) \succeq \half \beta_{a,y} \sqrtt{\alpha_a}.
\end{array} \]
Any solution of the {constraint} on the left satisfies the respective constraint on the right since
\[ \Diag(e_{\supp(\beta_{\bar{a}})}) \succeq \sqrtt{\beta_{\bar{a}}} \aand \half \beta_{a,y} \Diag(e_{\supp(\alpha_a)}) \succeq \half \beta_{a,y} \sqrtt{\alpha_a}. \]
Since the dual feasible regions are smaller in the classical case, we get that the optimal objective value cannot be less than the quantum version {since they share the same objective function}. This makes sense since the classical protocol has a weaker cheat detection step and we could have larger cheating probabilities. We can think of the classical case having more general strategies since the cheat detection step in the quantum version rules out certain strategies from being optimal. In this sense, the classical primal feasible regions are larger {than those in the quantum version} and the classical dual feasible regions are smaller.
This is similar to the relationship between the duality of convex sets. We have that $C_1 \subseteq C_2$ implies $C_1^* \supseteq C_2^*$ and the converse holds if $C_1$ and $C_2$ are closed convex cones. 
This relationship is depicted in Figure~\ref{PDFR}. 
  
\begin{figure}[ht] 
  \centering
   \includegraphics[width=5in]{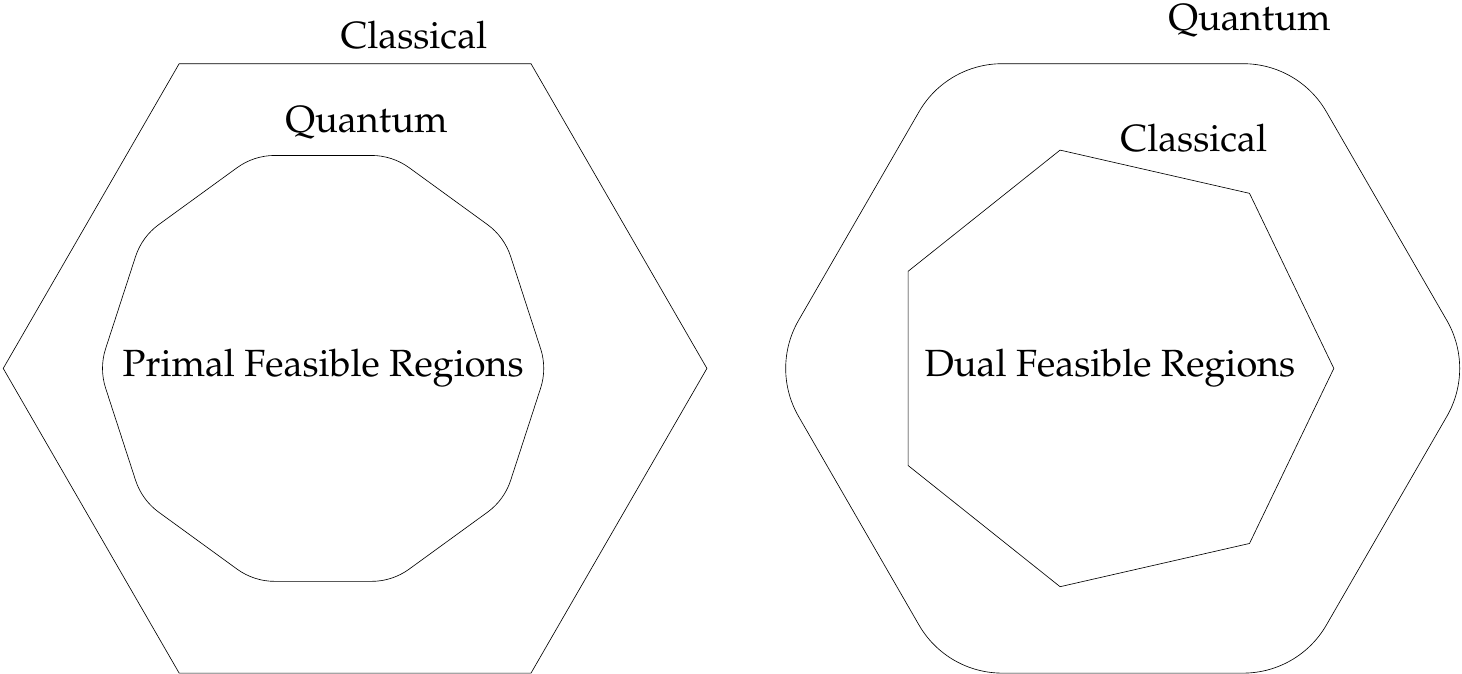} 
  \caption{{Relationship between the primal and dual feasible regions of the quantum and classical cheating strategy formulations.}}
\label{PDFR}
\end{figure}

\subsection{Point games for classical $\BCCF$-protocols}
\label{ssect:CPG}

In this subsection, we develop the classical analog to the quantum $\BCCF$-point games. Using these ``classical point games'', we prove that at least one party can cheat with probability $1$ in any classical $\BCCF$-protocol. A closer analysis shows that both cannot cheat with probability $1$, which holds true for quantum $\BCCF$-protocols as well.

Since point games are defined in terms of dual SDPs, we use the above embedding of the classical cheating LPs into SDPs to construct classical $\BCCF$-point games. Due to the similarities, very little about the quantum $\BCCF$-point games needs to be changed to attain classical  $\BCCF$-point games; we only need to change the definitions of Alice and Bob's projections. Of course, the dual solutions may be different due to the stronger constraints for the classical version.
The only differences are in the first few points (corresponding to the last few steps in Kitaev's proof that involve the projections). A quick calculation shows that these points are the same as well. The reason for this is because, in Bob's projections, we replace $\kb{\psi_a}$ with $\Diag(e_{\supp(\alpha_a)}) \otimes \id_{A'}$, but they have the same inner product with the honest state of the protocol
\[ \inner{\kb{\psi_a}}{\kb{\psi_a}} = \inner{\kb{\psi_a}}{\Diag(e_{\supp(\alpha_a)}) \otimes \id_{A'}} = 1. \]
A similar argument holds for Alice's projections as well.

Thus, the only difference between the classical point games are the values of the points, which are derived from slightly different dual constraints. Let us examine the point splits. In the quantum case, these are derived from the constraints
\[ \Diag(v_{a}) \succeq \sqrtt{\beta_{\bar{a}}} \aand  \Diag(z_{n+1}^{(y)}) \succeq \half \beta_{a,y} \sqrtt{\alpha_a}, \; \forall a \in \zo, y \in B. \]
In the classical case, the corresponding constraints are
\[ \Diag(\tilde{v}_{a}) \succeq \Diag(e_{\supp(\beta_{\bar{a}})}) \aand  \Diag(\tilde{z}_{n+1}^{(y)}) \succeq \half \beta_{a,y} \Diag(e_{\supp(\alpha_a)}), \; \forall a \in \zo, y \in B. \]
It is easy to see that $\tilde{v}_{a} = e_{\supp(\beta_{\bar{a}})}$ and
\[ \tilde{z}_{n+1, x,y} = \left\{ \begin{array}{rcl}
\half \beta_{0,y} & \textup{ if } & x \in \supp(\alpha_0) \setminus \supp(\alpha_1), \\
\half \beta_{1,y} & \textup{ if } & x \in \supp(\alpha_1) \setminus \supp(\alpha_0), \\
\half \max_{a \in \zo} \{ \beta_{a,y} \} & \textup{ if } & x \in \supp(\alpha_0) \cap \supp(\alpha_1), \\
0 & & \textup{otherwise,}
\end{array} \right. \]
are optimal assignments of these variables. Recall the two point splittings:
\[ \half \pg{1}{0} \to \sum_{b \in \zo} \sum_{y \in B} \quarter \beta_{\bar{a},y} \pg{\tilde{v}_{a,y}}{0} \aand \pg{v_{a,y}}{1} \to \sum_{x \in A} \alpha_{a,x} \pg{v_{a,y}}{\frac{2 \tilde{z}_{n+1, x,y}}{\beta_{a,y}}}. \]
We see that these are just probability splittings using the optimal assignment above (with possibly a point raise in the case of $x \in \supp(\alpha_0) \cap \supp(\alpha_1)$). These probability splittings are in contrast to the point splittings in the quantum case.
The rest of the constraints are the same as in the quantum case and correspond to point merging, probability merging, and aligning. Therefore, the only difference between quantum $\BCCF$-point games and the classical version is that non-trivial point splittings are allowed in the quantum version. Therefore, we get the following definition.

\begin{definition}[Classical $\BCCF$-point game (protocol independent)] 
A classical $\BCCF$-point game defined on the parameters $\alpha_0, \alpha_1 \in \prob^A$ and $\beta_0, \beta_1 \in \prob^B$, with final point $\pg{\zeta_{\B}}{\zeta_{\A}}$, is a quantum $\BCCF$-point game defined by the same parameters and having the final point $\pg{\zeta_{\B}}{\zeta_{\A}}$ but the point splittings are trivial (i.e., they are probability splittings).
\end{definition}
Using this definition, we define \emph{classical $\BCCF$-point game pairs} analogously to the quantum version.

To complete the picture, we now present the classical version of Theorem~\ref{SCFTHEOREM}.
\begin{theorem} \label{SCFTHEOREMc}
Suppose $\pgf{\zeta_{\B,0}}{\zeta_{\B,1}}{\zeta_{\A,0}}{\zeta_{\A,1}}$ is the final point of a classical $\BCCF$-point game pair defined on the parameters $\alpha_0, \alpha_1 \in \prob^A$ and $\beta_0, \beta_1 \in \prob^B$. Then
\[ P_{\B,0}^* \leq \zeta_{\B,0}, \quad P_{\B,1}^* \leq \zeta_{\B,1}, \quad P_{\A,0}^* \leq \zeta_{\A,0}, \; \textup{ and } \; P_{\A,1}^* \leq \zeta_{\A,1}, \]
where $P_{\B,0}^*$, $P_{\B,1}^*$, $P_{\A,0}^*$, $P_{\A,1}^*$ are the optimal cheating probabilities for  the corresponding classical $\BCCF$-protocol.
Moreover, there exists a classical $\BCCF$-point game pair with final point $\pgf{P_{\B,0}^*}{P_{\B,1}^*}{P_{\A,0}^*}{P_{\A,1}^*}$.
\end{theorem}

Figure~\ref{crystal} depicts the intricate connections between quantum and classical $\BCCF$-protocols and their point games.
 
\begin{figure}[ht]
  \centering
   \includegraphics[width=6.25in] {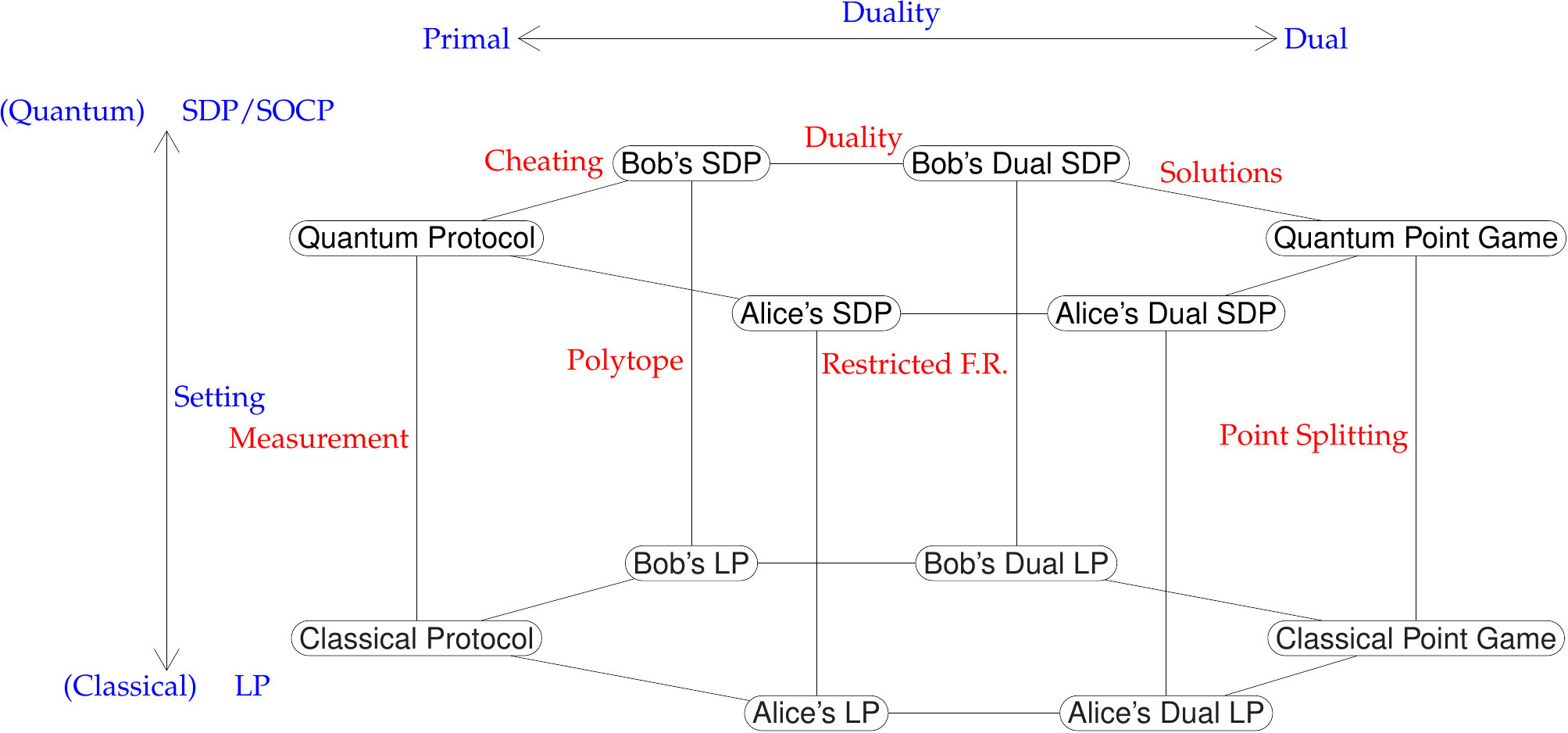} 
  \caption{Crystal structure of $\BCCF$-protocols. F.R. denotes ``feasible region''.}
\label{crystal}
\end{figure}

This crystal illustrates an analogy between physical theories and certain classes of convex optimization problems.
More specifically, we see that the generalization of classical mechanics to quantum mechanics is analogous
to the generalization of linear optimization to semidefinite optimization. In the rest of this subsection,
we elaborate on this idea and explain the benefits of viewing cryptographic protocols with this perspective.
 
It was shown in \cite{NST14} that the cheating SDPs can be written as SOCPs. The fact that the cheating in our protocols can be modelled as second-order cone programs hints that our protocol is using only a well-structured ``part'' of quantum mechanics. Indeed, apart from the initialization of the states at the beginning and the measurements at the end, our protocols only exchange quantum systems. Thus, our {protocols are} conceptually not using every aspect of quantum mechanics at every opportunity. Furthermore, the fact that the reduced problems are ``almost linear programs'' hints that our {protocols are} ``almost classical'', which is indeed the case. As we have discussed, only the measurement at the end makes the $\BCCF$-protocol ``quantum''. Thus, there is some philosophical connections between how ``quantum'' the protocol is and how ``SDP'' is the cheating formulation. This is a purely philosophical statement, of course, but it could provide insights towards protocol design. For example, in \cite{Moc07}, it was shown how to create quantum \emph{weak coin-flipping protocols} with arbitrarily small bias using SDPs. Since the structure of {his} protocols is implicit in the analysis and is very complicated, perhaps a better understanding of the SDPs used could shed light of the specific quantum mechanical behaviours to look for in designing such a protocol.
 
On a more quantitative side, once could use optimization theory to provide a measure of the complexity of a protocol. For example, in \cite{NST14}, we provide an SOCP representation for the hypograph of the fidelity function (characterizing the cheating probabilities in the reduced problems). Recall that the hypograph of a concave function is a convex set. Also, the dimension of the hypograph of $\rF(á, q) : \R^n_+ \to \R$ is equal to $n$ (assuming $q > 0$). Since the hypograph is $O(n)$-dimensional and convex, there exists a self-concordant barrier function for the set with complexity parameter $O(n)$, shown by Nesterov and Nemirovski \cite{NN94}. {The details of such functions are not necessary for this paper, but we mention that such a function allows} the derivation of interior-point methods for the underlying convex optimization problem which use $O(\sqrt{n} \log(1/\epsilon))$ iterations, where $\epsilon$ is an accuracy parameter. This suggests that we can use the complexity parameter of the self-concordant barrier function of the objective function characterizing cheating in the protocol as a complexity measure for the protocol. Such a measure could also {lend} itself to more general theories. For example, if one were to consider coin-flipping in a theory generalizing quantum mechanics, then one could still measure the complexity of the protocol by considering the complexity parameter of the objective function in a class of optimization problems possibly generalizing semidefinite optimization. Considering this paper uses restrictions of semidefinite optimization to characterize sub-quantum behaviour, it would not be surprising if generalizations of semidefinite optimization would characterize super-quantum behaviour. 
  
\subsection{Security analysis of classical $\BCCF$-protocols}
\label{ssect:ClassicalSecurity}

We start by giving an alternative proof that these classical protocols have bias $\eps = 1/2$ using the language of point games.

\begin{lemma} \label{lemCPG}
Suppose we have the following point game
\[ p_0 := \half \pg{0}{1} + \half \pg{1}{0} \to p_1 \to \cdots \to p_{m-1} \to p_m := \pg{\zeta_{\B}}{\zeta_{\A}}, \]
where each move is either point raising, point merging, probability merging, or probability splitting. Then $\zeta_{\B} \geq 1$ or $\zeta_{\A} \geq 1$.
\end{lemma}
 
\begin{proof}
Suppose for {the purpose of} contradiction that $\zeta_{\B}, \zeta_{\A} < 1$ and let $i \in \set{1, \ldots, m}$ be the smallest index such that $p_i$ has a point of the form $\pg{\zeta_{\B,i}}{\zeta_{\A,i}}$ with $\zeta_{\B,i}, \zeta_{\A,i} < 1$. Since $p_{i-1}$ has no such points, $\pg{\zeta_{\B,i}}{\zeta_{\A,i}}$ could not have been generated from a point raise, a probability merge, nor a probability split. Thus, $p_{i-1} \to p_i$ must be a point merge and suppose without loss of generality, it acted on the first coordinate. Then $p_{i-1}$ has two points $q_1 \pg{\zeta_1}{\zeta_{\A,i}}$ and $q_2 \pg{\zeta_2}{\zeta_{\A,i}}$ with $\dfrac{q_1 \, \zeta_1 + q_2 \, \zeta_2}{q_1 + q_2} = \zeta_{\B,i} < 1$ implying $\zeta_1 < 1$ or $\zeta_2 < 1$, 
a contradiction to the minimality of the choice of $i$. \qed
\end{proof}
   
Using the above lemma and Theorem~\ref{SCFTHEOREMc}, we have that classical $\BCCF$-protocols are insecure; at least one party can cheat with probability $1$. 
 
There are two special cases of classical protocols we consider in greater detail. Recall the points in the point game (before merging on $a$ in the first coordinate)
\begin{equation}
\dsum_{a \in \zo} \dsum_{y \in B} \dsum_{x \in A} p(x,a) \, p(y) \pg{v_{a,y}}{\frac{z_{n+1, x,y}}{p(y)}}. \label{step}
\end{equation}

The first case we consider is when $\alpha_0, \alpha_1, \beta_0, \beta_1 > 0$. Then we can set $v_{a,y} = 1$ for all ${a \in \zo}$, ${y \in B}$ and $z_{n+1, x,y} = \half \max_{a \in \zo} \beta_{a,y}$ for all $x \in A, y \in B$. After the merges and aligns, we have the final point being
\[ \pg{1}{\sum_{y \in B} \max_{a \in \zo} \half \beta_{a,y}} = \pg{1}{\half + \half \Delta(\beta_0, \beta_1)}, \]
using Lemma~\ref{mn}.
We can see that this is a $\BCCF$-point game with an optimal assignment of dual variables. Thus, Bob can cheat towards $1$ perfectly and Alice can force $0$ with probability $\half + \half \Delta(\beta_0, \beta_1)$ as seen on the left in Figure~\ref{AliceCPG}. These two quantities are invariant under switching $\beta_0$ and $\beta_1$, thus $P_{\B,0}^* = P_{\B,1}^* = 1$ and ${P_{\A,0}^* = P_{\A,1}^* = \half + \half \Delta(\beta_0, \beta_1)}$. The corresponding optimal cheating strategies in the classical $\BCCF$-protocol are obvious by noticing the cheat detection step does nothing when the vectors have full support. Bob can send anything during the first $n$ messages and then return $b = a$. Alice can send $a$ corresponding to her best guess of $b$ from her information about $y \in B$, i.e., she can cheat with the probability she can infer $b$ from $y \in B$. An interesting observation is that this is a valid point game pair for both the classical and quantum versions for \emph{any} $\alpha_0, \alpha_1 \in \prob^A$ and $\beta_0, \beta_1 \in \prob^B$ {since the dual feasible regions for the classical formulations are contained in the dual feasible regions of the quantum formulations.}
Therefore, we have that $P_{\A,0}^*, P_{\A,1}^* \leq \half + \half \Delta(\beta_0, \beta_1)$ for every quantum $\BCCF$-protocol as well. This bound can be interpreted as follows. Suppose we change the order of the messages in the $\BCCF$-protocol in Alice's favour, so that Bob's first $n$ messages are sent first, followed by \emph{all} of Alice's messages, then finally Bob's last message. Then Alice's new cheating probability would be $\half + \half \Delta(\beta_0, \beta_1)$ and is an obvious upper bound on the amount she can cheat in the original protocol (since she gets information about $b$ sooner than intended). This argument works for the classical and quantum versions.

It may seem that classical protocols favour a cheating Bob, but this is not always the case. Consider the case when $\beta_0 \perp \beta_1$ and $\alpha_0, \alpha_1 > 0$. Then $\frac{z_{n+1, x,y}}{p(y)} = 1$ for all ${y \in \supp(\beta_0) \cup \supp(\beta_1)}$, thus the second coordinate equals $1$ for all points in~(\ref{step}), and remains that way until the end of the point game. This proves Alice can cheat with probability $1$, which is obvious since Bob's first message fully reveals $b$ and she can always pass the cheat detection step. The extent to which Bob can cheat depends on the choice of $\alpha_0$ and $\alpha_1$ and can be calculated as 
\[ \dsum_{x_1 \in A_1} \max_{a \in \zo} \sum_{x_2 \in A_2} \cdots \sum_{x_n \in A_n} \half \alpha_{a,x} = \half + \half \Delta \left( \tr_{A_2 \times \cdots \times A_n}(\alpha_0), \tr_{A_2 \times \cdots \times A_n}(\alpha_1) \right), \]
using Lemma~\ref{mn}. This is a distance measure between the two marginal distributions over Alice's first message $x_1$. This point game is depicted on the right in Figure~\ref{AliceCPG}. Bob can cheat with this probability since he can choose $b$ equal to his best guess for ${\bar{a}}$ from his information about $x_1$. Once his first message is sent, he must keep his choice of $b$ or he will be caught cheating with certainty. These cheating probabilities do not depend on $\beta_0$ or $\beta_1$, so we have $P_{\A,0}^* = P_{\A,1}^* = 1$ and $P_{\B,0}^* = P_{\B,1}^* = \dfrac{1}{2} + \dfrac{1}{2} \Delta \left(\tr_{A_2 \times \cdots \times A_n}(\alpha_0),  \tr_{A_2 \times \cdots \times A_n}(\alpha_1) \right)$. 
Therefore, a classical $\BCCF$-protocol could favour either party. 

\begin{figure}[ht]
  \centering
   \includegraphics[width=6.5in] {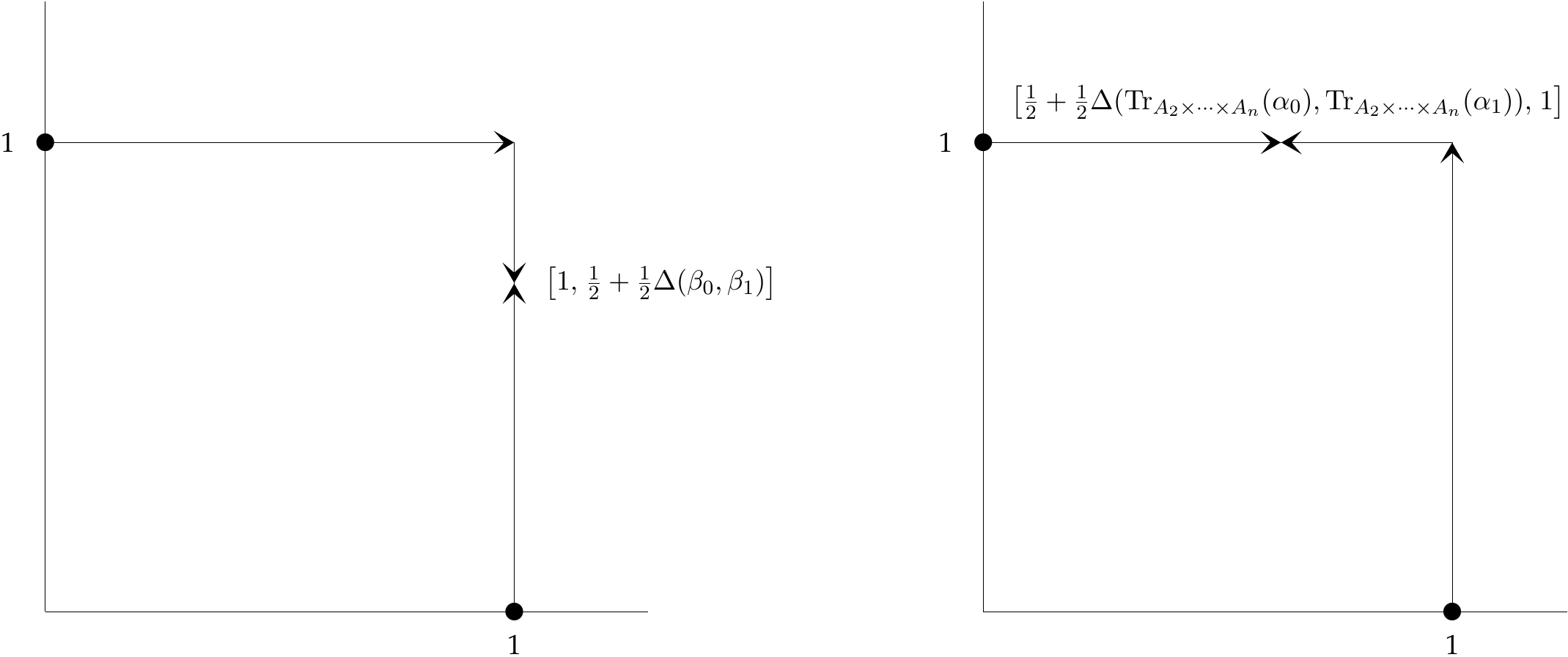} 
  \caption{{Classical $\BCCF$-point game examples. Left: A classical $\BCCF$-point game favouring cheating Bob. Right: A classical $\BCCF$-point game favouring cheating Alice.}}
  \label{AliceCPG}
\end{figure}

This raises the question: Can we find a $\BCCF$-protocol such that both parties can perfectly control the outcome? We now argue that no such classical, and hence no such quantum, $\BCCF$-protocol exists.
{Suppose for the purpose of} contradiction that this is the case. Then we must have
\[ 1 = P_{\A,0}^* \leq \half + \half \Delta(\beta_0, \beta_1) \leq 1 \]
which implies $\beta_0 \perp \beta_1$. {Then} the only way for Bob to cheat with probability $1$ is to have complete information about $a$ after Alice's first message, implying the orthogonality condition 
\[ {\tr_{A_2 \times \cdots \times A_n}(\alpha_0) \perp  \tr_{A_2 \times \cdots \times A_n}(\alpha_1)}. \] 
This can only be the case when $\alpha_0 \perp \alpha_1$ and in this case, {as} we have argued before, that Alice must stick to her choice of $a$ after her first message. Since she has no information about $b$ before the start of the protocol, she can only cheat with probability $1/2$, a contradiction.  

Therefore, there is no classical {$\BCCF$-protocol} where both Alice and Bob can cheat perfectly, and hence no quantum protocol. Along with the fact that classical protocols are insecure, this proves the following theorem. 

\begin{theorem} \label{Thm:security}
In {every} quantum $\BCCF$-protocol, at most one party can cheat with probability $1$. In {every} classical $\BCCF$-protocol, exactly one party can cheat with probability $1$.
\end{theorem}

\newpage 
\section{Using classical protocols to lower bound the quantum bias} \label{LB}

In this section, we prove that no quantum $\BCCF$-protocol can have bias $\eps = 1/\sqrt{2} - 1/2$. More specifically, we prove that only protocols that share optimal cheating probabilities with their classical counterpart can saturate Kitaev's lower bound on the product of the cheating probabilities. This shows yet another connection between quantum and classical $\BCCF$-protocols. 

We start with rederiving Kitaev's lower bound using the reduced SDPs. The duals of Bob's and Alice's reduced SDPs, each for {forcing outcome $0$, are} 
\[ \begin{array}{rrcclrrrcll}
\textrm{} & \inf                         & \tr_{A_1}(w_1) & & 
& \quad  
& \inf                         & z_1 \\

                     & \textup{s.t.} & w_1 \otimes e_{B_1} & \geq & \tr_{A_2}(w_2), 
& \quad 
& \textup{s.t.} & z_1 \cdot e_{A_1} & \geq & \tr_{B_1}(z_2), \\

                     & & w_2 \otimes e_{B_2} & \geq & \tr_{A_3}(w_3), 
& \quad 
&                               & z_2 \otimes e_{A_2} & \geq & \tr_{B_2}(z_3), \\

                     & & & \vdots 
& \quad 
& & & & \vdots \\

                     &                               & w_{n} \otimes e_{B_{n}} & \geq & \half \sum_{a \in \zo} \alpha_a \otimes v_a,
& \quad 
&                               & z_n \otimes e_{A_n} & \geq & \tr_{B_n}(z_{n+1}), \\

                     & & \Diag(v_a) & \succeq & \sqrtt{\beta_a}, \;\; \forall a, 
& \quad 
& & \Diag(z_{n+1}^{(y)}) & \succeq & \half \beta_{a,y} \sqrtt{\alpha_a}, \;\; \forall a, y, 
\end{array} \] 
respectively. 
Let $(w_1, \ldots, w_n, v_0, v_1)$ be optimal for Bob's dual above and let $(z_1, \ldots, z_{n+1})$ be optimal for Alice's dual above. We have 
\[ P_{\B,0}^* P_{\A,0}^*
= \tr_{A_1}(w_1) \, z_1 
= \inner{\tr_{A_1}(w_1)}{z_1} 
= \inner{w_1}{z_1 \cdot e_{A_1}} 
\geq \inner{w_1}{\tr_{B_1}(z_2)} 
= \inner{w_1 \otimes e_{B_1}}{z_2}. 
\] 
In a similar manner as was done in Section~\ref{BCCFPG}, we can alternate through most of the vector inequality dual constraints to show that 
\[ P_{\B,0}^* P_{\A,0}^* \geq \inner{w_n \otimes e_{B_n}}{z_{n+1}}. 
\] 
We bound the quantity $\inner{w_n \otimes e_{B_n}}{z_{n+1}}$ using the rest of the dual constraints, albeit in a slightly different manner. We decompose $z_{n+1} = \sum_{y \in B} z_{n+1}^{(y)} \otimes e_y$ and use the rest of the dual constraints to get 
\begin{eqnarray*}
\inner{w_n \otimes e_{B_n}}{z_{n+1}}
& \geq & \half \sum_{a \in \zo} \sum_{y \in B} \inner{\alpha_a \otimes v_a}{z_{n+1}^{(y)} \otimes e_y} \\
& = & \half \sum_{a \in \zo} \sum_{y \in B} \inner{\sqrtt{\alpha_a}}{\Diag(z_{n+1}^{(y)})} \inner{\Diag(v_a)}{e_y e_y^{\T}} \\ 
& \geq & 
\half \sum_{a \in \zo} \sum_{y \in B} \inner{\sqrtt{\alpha_a}}{\half \beta_{a,y} \sqrtt{\alpha_a}} \inner{\Diag(v_a)}{e_y e_y^{\T}} \\ 
& = & \quarter \sum_{a \in \zo} \inner{\Diag(v_a)}{\Diag(\beta_a)} \\
& = & \quarter \sum_{a \in \zo} \inner{\Diag(v_a)}
{\sqrtt{\beta_a}} \\
& \geq & \quarter \sum_{a \in \zo} \inner{\sqrtt{\beta_a}}
{\sqrtt{\beta_a}} \\
& = & \half.
\end{eqnarray*}
Therefore, we get Kitaev's lower bound $P_{\A,0}^* P_{\B,0}^* \geq 1/2$ implying that $P_{\A,0}^* \geq 1/\sqrt{2}$ or $P_{\B,0}^* \geq 1/\sqrt{2}$. We get the inequality $P_{\A,1}^* P_{\B,1}^* \geq 1/2$ by switching $\beta_0$ with $\beta_1$ in the proof above (and the dual variables accordingly). 

Using these two lower bounds, we show that it is impossible to have a quantum $\BCCF$-protocol with bias $\eps = 1/\sqrt{2} - 1/2$ by proving Kitaev's bounds can only be saturated with protocols where one party can cheat perfectly. More specifically, we show that if there exist four dual solutions that saturate both of Kitaev's bounds
\[ P_{\A,0}^* P_{\B,0}^* \geq 1/2 \aand P_{\A,1}^* P_{\B,1}^* \geq 1/2, \]
{then all four of the dual solutions must also be feasible in the duals of the classical versions.} 
  
\begin{theorem} \label{Kitproof}
Suppose a quantum $\BCCF$-protocol satisfies ${P_{\A,0}^* \, P_{\B,0}^* = 1/2}$ and ${P_{\A,1}^* \, P_{\B,1}^* = 1/2}$. Then the cheating probabilities are the same as in the corresponding classical protocol {defined on the same parameters}.
\end{theorem}

\begin{proof} 
{By looking at the proof of Kitaev's bound above, we see} that if it were saturated, then every inequality must hold with equality. {Therefore, we know $\Diag(v_a) \succeq \sqrtt{\beta_a}$ 
has no slack on the subspace spanned by $\sqrtt{\beta_a}$, i.e., 
$
\inner{\Diag(v_a) - \sqrtt{\beta_a}}{\sqrtt{\beta_a}} = 0$, {or equivalently,} $\inner{\Diag(v_{a})}{\sqrtt{\beta_a}} = 1$,  
for both $a \in \zo$}. Consider $v_{a} = e_{\supp(\beta_a)}$ which satisfies {$\Diag(v_a) \succeq \sqrtt{\beta_a}$} and the condition {${\inner{\Diag(v_{a})}{\sqrtt{\beta_a}} = 1}$}. We show this choice is unique (on $\supp(\beta_a)$). Consider the optimization problem 
\begin{eqnarray*}
& & \inf \left\{ \inner{\Diag(v_a)}{\sqrtt{\beta_a}} : \Diag(v_a) \succeq \sqrtt{\beta_a} \right\} \\
& = & \inf \left\{ \sum_{y \in \supp(\beta_a)} v_{a,y} \beta_{a,y} : \sum_{y \in \supp(\beta_a)} \frac{\beta_{a,y}}{v_{a,y}} \leq 1, \; v_{a,y} > 0 \right\}.
\end{eqnarray*}
Obviously $v_{a} = e_{\supp(\beta_a)}$ is an optimal solution since $1$ is a lower bound on the optimal objective value. Suppose there are two different optimal solutions $v'$ and $v''$. {Notice that} $\half v' + \half v''$ has the same objective value, but satisfies the constraint
$\sum_{y \in \supp(\beta_a)} \frac{\beta_{a,y}}{v_{a,y}} \leq 1$ 
with strict inequality since the function $\sum_{y \in \supp(\beta_a)} \frac{\beta_{a,y}}{v_{a,y}}$ is strictly convex. Thus, we can scale $\half v' + \half v''$ to get a better objective function value, a contradiction.
Therefore, if Kitaev's bound is saturated, we must have $v_{a,y} = 1$ for all $a \in \zo$, $y \in \supp(\beta_a)$.

We argue the same about Alice's dual variables $z_{n+1}^{(y)}$.
If Kitaev's inequalities are saturated, we have 
{$\inner{\sqrtt{\alpha_a}}{\Diag(z_{n+1}^{(y)}) - \half \beta_{a,y} \sqrtt{\alpha_a}} = 0$, or just, $\inner{\sqrtt{\alpha_a}}{\Diag(z_{n+1}^{(y)})} = \half \beta_{a,y}$, 
for all $a,y$ such that $v_{a,y} > 0$, i.e., for all $y \in \supp(\beta_{a})$}. 
Similar to the arguments above, we need $[z_{n+1}^{(y)}]_x = \half \beta_{a,y}$ for $a \in \zo$, $x \in \supp(\alpha_a)$, and $y \in \supp(\beta_a)$. 

To summarize, if we have Kitaev's bounds saturated, then the optimal dual solutions satisfy $\Diag(v_a) \succeq \Diag(e_{\supp(\beta_a)})$ and $\Diag(z_{n+1}^{(y)}) \succeq \half \beta_{a,y} \, \Diag(e_{\supp(\alpha_a)})$, for all $a \in \zo, y \in B$, which are exactly the constraints in the dual {LPs} for the classical version. {Therefore, the protocol must have the property that relaxing the cheat detections in $\Pi_{\A,0}$ and $\Pi_{\B,0}$ (obtaining the classical cheat detections) preserves the two cheating probabilities}. We can repeat the same argument with Alice and Bob cheating towards $1$ and get the two corresponding classical cheating probabilities. Therefore, we have all four cheating probabilities are equal to those of the corresponding classical protocol, as desired. \qed
\end{proof} 
 
Since every classical protocol allows one party to cheat perfectly, we obtain Corollary~\ref{CorSec}, that $\eps = 1/\sqrt{2} - 1/2$ is impossible for any $\BCCF$-protocol. 

The proof of Theorem~\ref{Kitproof} gives necessary conditions on classical protocols that saturate Kitaev's bound. Note from the condition on $z_{n+1}^{(y)}$, we have $[z_{n+1}]_{x,y} = \half \beta_{a,y}$ when $\beta_{a,y} , \alpha_{a,x} > 0$. 
In the case when $\alpha_0, \alpha_1, \beta_0, \beta_1 > 0$, then $\beta_0$ must equal $\beta_1$. This makes sense since Bob can easily cheat with probability $1$, but if $\beta_0 \neq \beta_1$, then Alice could cheat with probability greater than $1/2$. In the case when $\alpha_0 \perp \alpha_1$, the condition above tells us nothing, but it is easy to see that Alice fully reveals $a$ in the first message, thus she can cheat with probability $1/2$ and Bob can cheat with probability $1$.

\section{Conclusions}
\label{conc}

We studied the security of quantum coin-flipping protocols based on bit-commitment utilizing SDP formulations of cheating strategies. These SDPs allowed us to use concepts from convex optimization to further our understanding of the security of such protocols. In particular, using a reduction of the SDPs and duality theory, we were able to find the classical protocol counterpart and develop a family of point games corresponding to each of the classical and quantum protocols.

Using the connections between classical and quantum $\BCCF$-protocols, we were able to show that a bias of $\eps = 1/\sqrt{2} - 1/2$ is impossible for $\BCCF$-protocols using a modified {proof} of Kitaev's lower bound.

An open problem is to find the optimal cheating strategies for a general $n$-round $\BCCF$-protocol. This can be accomplished by finding closed-form optimal solutions to the cheating SDPs or the reduced cheating SDPs. Very few highly interactive protocols, such as $\BCCF$-protocols, have descriptions of optimal cheating strategies and therefore having such for this family of protocols would be very interesting.

A benefit of knowing the optimal strategies would be to help resolve the problem of finding the smallest bias for $\BCCF$-protocols. In~\cite{NST14}, we analyzed $\BCCF$-protocols from a computational perspective. We computationally checked the bias of over $10^{16}$ four and six-round $\BCCF$-protocols and based on the findings we conjecture that having all four cheating probabilities strictly less than $3/4$ is impossible.

A related open problem is to find an explicit construction of optimal protocols for coin-flipping and bit-commitment. We can accomplish both of these tasks by finding {an explicit} construction of optimal weak coin-flipping protocols (see~\cite{CK09, CK11}), so this would be very rewarding. Technically, such a construction is implicit in~\cite{Moc07}, however it involves many reductions and is quite complicated.
   
\section*{Acknowledgements}
We thank Andrew Childs, Michele Mosca, Peter H\o yer, and John Watrous for their comments and suggestions. A.N.'s research is supported in part by NSERC Canada, CIFAR, ERA (Ontario), QuantumWorks, and MITACS. A part of this work was completed
at Perimeter Institute for Theoretical Physics. Perimeter Institute
is supported in part by the Government of Canada
through Industry Canada and by the Province of Ontario through MRI.
J.S.'s research is supported by NSERC Canada, MITACS, and ERA (Ontario).
L.T.'s research is supported in part by Discovery Grants from NSERC. 

Research at the Centre for Quantum Technologies at the National University of Singapore is partially funded by the Singapore Ministry of Education and the National Research Foundation, also through the Tier 3 Grant ``Random numbers from quantum processes'', (MOE2012-T3-1-009). 
   
\nocite{AharonovCGKM14}

\bibliographystyle{alpha}
\bibliography{CCF_bib}

\newcommand{\etalchar}[1]{$^{#1}$}
\begin{thebibliography}{ACG{\etalchar{+}}14}

\bibitem[ACG{\etalchar{+}}14]{AharonovCGKM14}
Dorit Aharonov, Andr{\'e} Chailloux, Maor Ganz, Iordanis Kerenidis, and
  Lo{\"\i}ck Magnin.
\newblock A simpler proof of existence of quantum weak coin flipping with
  arbitrarily small bias.
\newblock Available as arXiv.org e-Print quant-ph/1402.7166, 2014.

\bibitem[AG03]{AG03}
Farid Alizadeh and Donald Goldfarb.
\newblock Second-order cone programming.
\newblock {\em Mathematical Programming}, 95:3--51, 2003.

\bibitem[Alb83]{A83}
Peter~M. Alberti.
\newblock A note on the transition probability over {$C^*$}-algebras.
\newblock {\em Letters in Mathematical Physics}, 7(1):25--32, 1983.

\bibitem[Amb01]{Amb01}
Andris Ambainis.
\newblock A new protocol and lower bounds for quantum coin flipping.
\newblock In {\em Proceedings of 33rd Annual ACM Symposium on the Theory of
  Computing}, pages 134 -- 142. ACM, 2001.

\bibitem[ATVY00]{ATVY00}
Dorit Aharonov, Amnon {Ta-Shma}, Umesh Vazirani, and Andrew Chi-Chih Yao.
\newblock Quantum bit escrow.
\newblock In {\em Proceedings of 32nd Annual ACM Symposium on the Theory of
  Computing}, pages 705--714. ACM, 2000.

\bibitem[BB84]{BB84}
Charles Bennett and Gilles Brassard.
\newblock Quantum cryptography: Public key distribution and coin tossing.
\newblock In {\em Proceedings of the IEEE International Conference on
  Computers, Systems, and Signal Processing}, pages 175--179. IEEE Computer
  Society, 1984.

\bibitem[Blu81]{Blu81}
Manuel Blum.
\newblock Coin flipping by telephone.
\newblock In Allen Gersho, editor, {\em Advances in Cryptology: A Report on
  CRYPTO 81, CRYPTO 81, IEEE Workshop on Communications Security, Santa
  Barbara, California, USA, August 24-26, 1981}, pages 11--15. U. C. Santa
  Barbara, Dept. of Elec. and Computer Eng., ECE Report No. 82-04, 1982, 1981.

\bibitem[CK09]{CK09}
Andr{\'e} Chailloux and Iordanis Kerenidis.
\newblock Optimal quantum strong coin flipping.
\newblock In {\em Proceedings of 50th IEEE Symposium on Foundations of Computer
  Science}, pages 527--533. IEEE Computer Society, 2009.

\bibitem[CK11]{CK11}
Andr{\'e} Chailloux and Iordanis Kerenidis.
\newblock Optimal bounds for quantum bit commitment.
\newblock In {\em Proceedings of the 52nd Annual IEEE Symposium on Foundations
  of Computer Science}, pages 354--362. IEEE Computer Society Press, October
  2011.

\bibitem[FMP{\etalchar{+}}12]{FMPTW2012}
Samuel Fiorini, Serge Massar, Sebastian Pokutta, Hans~Raj Tiwary, and Ronald
  de~Wolf.
\newblock Linear vs. semidefinite extended formulations: exponential separation
  and strong lower bounds.
\newblock In {\em {P}roceedings of the 2012 {ACM} {S}ymposium on {T}heory of
  {C}omputing}, pages 95--106. ACM, New York, 2012.

\bibitem[GW07]{GW07}
Gus Gutoski and John Watrous.
\newblock Toward a general theory of quantum games.
\newblock In {\em Proceedings of the Thirty-Ninth Annual ACM Symposium on
  Theory of Computing}, pages 565--574, New York, NY, USA, 2007. ACM.

\bibitem[Kit02]{Kit03}
Alexei Kitaev.
\newblock Quantum coin-flipping.
\newblock Unpublished result. Talk in the 6th Annual workshop on Quantum
  Information Processing, QIP 2003, Berkeley, CA, USA, December 2002, 2002.

\bibitem[KN04]{KN04}
Iordanis Kerenidis and Ashwin Nayak.
\newblock Weak coin flipping with small bias.
\newblock {\em Information Processing Letters}, 89(3):131--135, 2004.

\bibitem[LC97]{LC97}
Hoi-Kwong Lo and Hoi~Fung Chau.
\newblock Is quantum bit commitment really possible?
\newblock {\em Physical Review Letters}, 78(17):3410--3413, 1997.

\bibitem[LC99]{LC99}
Hoi-Kwong Lo and Hoi~Fung Chau.
\newblock Unconditional security of quantum key distribution over arbitrarily
  long distances.
\newblock {\em Science}, 283:2050--2056, 1999.

\bibitem[May97]{May97}
Dominic Mayers.
\newblock Unconditionally secure quantum bit commitment is impossible.
\newblock {\em Physical Review Letters}, 78(17):3414--3417, 1997.

\bibitem[May01]{M01}
Dominic Mayers.
\newblock Unconditional security in quantum cryptography.
\newblock {\em Journal of the ACM}, 48(3):351--406, 2001.

\bibitem[Mit03]{Mit03}
Hans~D. Mittelmann.
\newblock An independent benchmarking of {SDP} and {SOCP} solvers.
  {C}omputational semidefinite and second order cone programming: the state of
  the art.
\newblock {\em Mathematical Programming}, 95(2):407--430, 2003.

\bibitem[Moc05]{Moc05}
Carlos Mochon.
\newblock A large family of quantum weak coin-flipping protocols.
\newblock {\em Physical Review A}, 72(2):022341, 2005.

\bibitem[Moc07]{Moc07}
Carlos Mochon.
\newblock Quantum weak coin flipping with arbitrarily small bias.
\newblock Available as arXiv.org e-Print quant-ph/0711.4114, 2007.

\bibitem[MVW12]{MVW12}
Abel Molina, Thomas Vidick, and John Watrous.
\newblock Optimal counterfeiting attacks and generalizations for {W}iesner's
  quantum money.
\newblock In {\em Proceedings of the 7th Conference on Theory of Quantum
  Computation, Communication, and Cryptography}, pages 45--64, 2012.

\bibitem[NN94]{NN94}
Yurii Nesterov and Arkadi Nemirovski.
\newblock {\em Interior-Point Polynomial Algorithms in Convex Programming}.
\newblock Society for Industrial and Applied Mathematics, 1994.

\bibitem[NST14]{NST14}
Ashwin Nayak, Jamie Sikora, and Levent T{un\c cel}.
\newblock A search for quantum coin-flipping protocols using optimization
  techniques.
\newblock Available as arXiv.org e-Print math.OC/1403.0505, 2014.

\bibitem[PS00]{PS00}
John Preskill and Peter~W. Shor.
\newblock Simple proof of security of the {BB84} quantum key distribution
  protocol.
\newblock {\em Physical Review Letters}, 85(2):441--444, 2000.

\bibitem[Sik12]{SikoraPHD12}
Jamie Sikora.
\newblock {\em Analyzing Quantum Cryptographic Protocols Using Optimization
  Techniques}.
\newblock PhD thesis, University of Waterloo, 2012.

\bibitem[SR01]{SR01}
Robert~W. Spekkens and Terence Rudolph.
\newblock Degrees of concealment and bindingness in quantum bit commitment
  protocols.
\newblock {\em Physical Review A}, 65:012310, 2001.

\bibitem[Stu99]{Stu99}
Jos~F. Sturm.
\newblock Using {SeDuMi} 1.02, a {MATLAB} toolbox for optimization over
  symmetric cones.
\newblock {\em Optimization Methods and Software}, 11:625--653, 1999.

\bibitem[Stu02]{Stu02}
Jos~F. Sturm.
\newblock Implementation of interior point methods for mixed semidefinite and
  second order cone optimization problems.
\newblock {\em Optimization Methods and Software}, 17(6):1105--1154, 2002.

\bibitem[TW12]{TW08}
Levent Tun{\c c}el and Henry Wolkowicz.
\newblock Strong duality and minimal representations for cone optimization.
\newblock {\em Computational Optimization and Applications}, pages 1--30, 2012.

\bibitem[Wie83]{Wiesner83}
Stephen Wiesner.
\newblock Conjugate coding.
\newblock {\em SIGACT News}, 15(1):78--88, January 1983.

\bibitem[WSV00]{SDP}
Henry Wolkowicz, Romesh Saigal, and Lieven Vandenberghe, editors.
\newblock {\em Handbook of Semidefinite Programming}.
\newblock Kluwer Academic Publishers, 2000.

\bibitem[Yao79]{Yao79}
Andrew Chi-Chih Yao.
\newblock Some complexity questions related to distributive computing.
\newblock In {\em Proceedings of the Eleventh Annual ACM Symposium on Theory of
  Computing}, STOC '79, pages 209--213, New York, NY, USA, 1979. ACM.

\bibitem[Yao93]{Yao93}
Andrew Chi-Chih Yao.
\newblock Quantum circuit complexity.
\newblock In {\em Proceedings of the 34th Annual IEEE Symposium on Foundations
  of Computer Science}, pages 352--361, Los Alamitos, CA, USA, 1993. IEEE
  Computer Society Press.

\end{thebibliography}

\appendix 
 
\newpage  
\section{Coin-flipping and Kitaev's protocol and point game formalisms} \label{CFSDP}
 
Kitaev developed point games from his SDP formulation of cheating strategies for coin-flipping protocols. Here, we review the construction in \cite{Moc07}, see also \cite{AharonovCGKM14}.

We start with a general setting for a coin-flipping protocol. This setting has a space devoted for messages and each message has the same dimension. This is done for convenience as it makes the analysis in this section  simpler.

A coin-flipping protocol can be described by the following parameters:
\begin{itemize}
\item The number of messages, denoted here as $n$. We can assume $n$ is even,
\item three Hilbert spaces: Alice's private space $\C^A$, a message space $\C^M$, and Bob's private space $\C^B$,
\item a set of unitaries $\set{U_{\A,1}, U_{\A,3}, \ldots, U_{\A,n-1}}$ acting on $\C^{A \times M}$. These correspond to Alice's messages to Bob,
\item a set of unitaries $\set{U_{\B,2}, U_{\B,4}, \ldots, U_{\B,n}}$ acting on $\C^{M \times B}$. These correspond to Bob's messages to Alice,
\item a projective measurement on $\C^A$ for Alice $(\Pi_{\A,0}, \Pi_{\A,1}, \Pi_{\A, \textup{abort}})$ determining Alice's protocol outcome,
\item a projective measurement on $\C^B$ for Bob $(\Pi_{\B,0}, \Pi_{\B,1}, \Pi_{\B,\textup{abort}})$ determining Bob's protocol outcome.
\end{itemize}

The protocol proceeds as follows. Alice initializes the space $\C^A$ to $\ket{\psi_{A,0}}$ and Bob initializes $\C^{M \times B}$ to $\ket{\psi_{M,0}}_{M} \ket{\psi_{B,0}}_{B}$ and sends $\C^M$ to Alice. Then Alice applies her first unitary $U_{\A,1}$ and sends $\C^M$ to Bob. Then he applies his first unitary $U_{\B,2}$ and returns $\C^M$ to Alice. They repeat this until Bob applies his last unitary $U_{\B,n}$. Then they both measure their private spaces to get the outcome of the protocol. This process is depicted in Figure~\ref{protocol} for the case of $n=4$.
 
\begin{figure}[ht]
  \centering
   \includegraphics[width=4in]{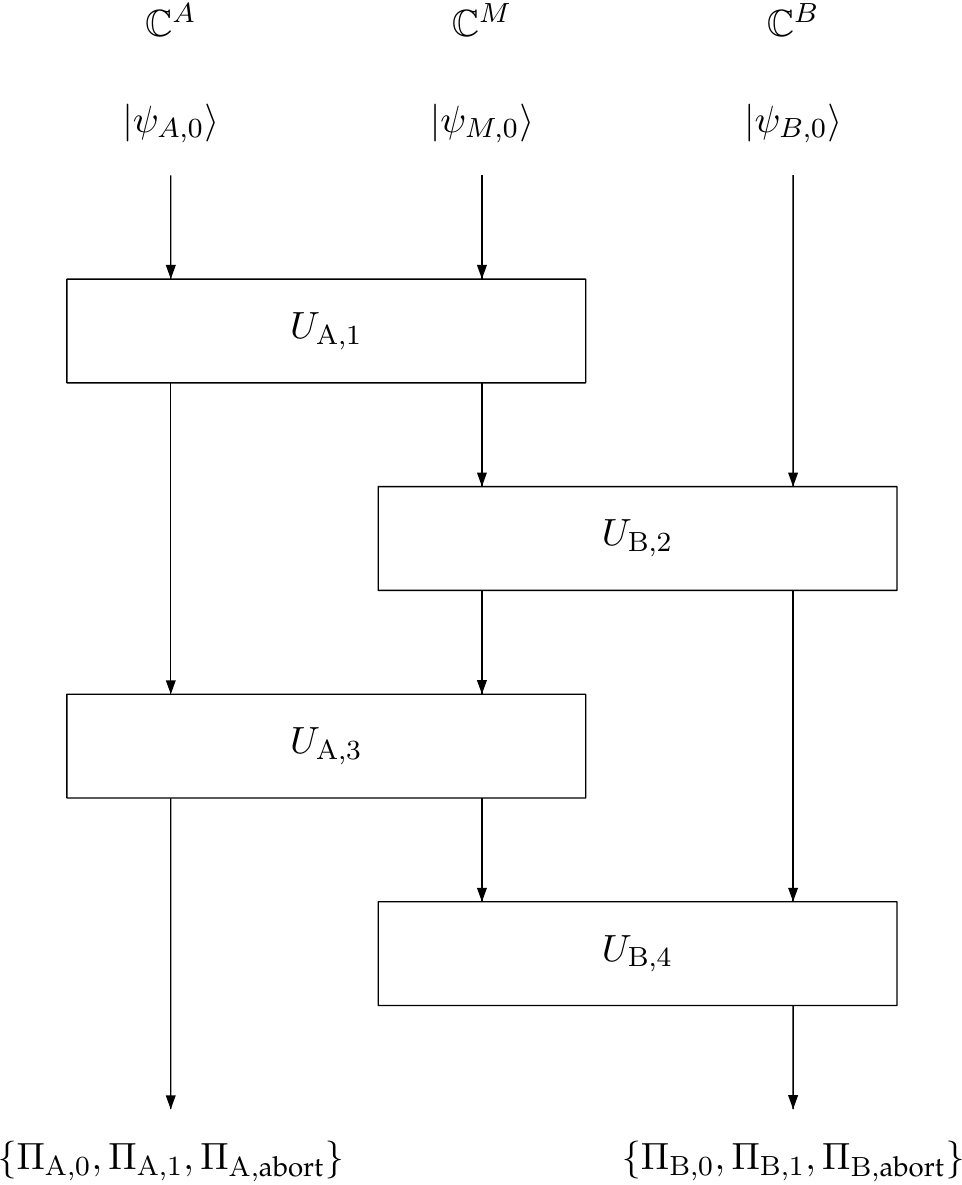} 
\caption{Four-round coin-flipping protocol.}
\label{protocol}
\end{figure}

The protocol parameters must satisfy the requirements: 
\begin{enumerate}
\item Alice and Bob do not abort when both are honest.
\item They output the same bit when they are honest, and that bit is randomly generated.
\end{enumerate}
If we let $\ket{\psi_n} \in \C^{A \times M \times B}$ be the state at the end of the protocol when Alice and Bob are honest, both requirements are satisfied when
\begin{equation}
\inner{\Pi_{\A,0} \otimes \id_M \otimes \Pi_{\B,0}}{\kb{\psi_n}} = \inner{\Pi_{\A,1} \otimes \id_M \otimes \Pi_{\B,1}}{\kb{\psi_n}} = \dfrac{1}{2}. \label{CFcondition}
\end{equation}

\subsection{Cheating SDPs}

We can calculate the extent cheating Bob can force honest Alice to output a fixed desired outcome, say $c \in \zo$, by solving the following SDP:
\[ \begin{array}{rrrcllllllllllllll}
& P_{\B,c}^* = \max                         & \inner{\Pi_{\A,c}}{\rho_{A,n}} \\
                     & \textrm{subject to} & \rho_{A,0} & = & \kb{\psi_{A,0}}, \\
                     & & \rho_{A,i} & = & \rho_{A, i-1}, & \textup{ for all } i \textup{ even}, \\
&                               & \tr_{M} (\trho_{A,i}) & = & \rho_{A,i}, & \textup{ for all } i \textup{ even}, \\
& & \rho_{A,i} & = & \tr_{M} \left( U_{\A,i} \trho_{A,i-1} U_{\A,i}^* \right), & \textup{ for all } i \textup{ odd}, \\
& & \rho_{A,i} & \in & \pos^A, & \textup{ for all } i, \\
& & \trho_{A,i} & \in & \pos^{A \times M}, & \textup{ for all } i \textup{ even}. \\
\end{array} \]
The variables describe the parts of the quantum state under Alice's control during different times in the protocol as depicted in Figure~\ref{fig:protocolwithvariables}. The constraints model how much cheating Bob can change the current state of the protocol in each message and the objective function is the probability Alice accepts outcome $c \in \{ 0, 1 \}$ by measuring the state she has at the end of the protocol.

We get a very similar SDP for cheating Alice by switching the projections and interchanging the ``odd'' constraints with the ``even'' ones:
\[ \begin{array}{rrrcllllllllllllll}
& P_{\A,c}^* = \max                         & \inner{\Pi_{\B,c}}{\rho_{B,n}} \\
                     & \textrm{subject to} & \rho_{B,0} & = & \kb{\psi_{B,0}}, \\
                     & & \rho_{B,i} & = & \rho_{B, i-1}, & \textup{ for all } i \textup{ odd}, \\
&                               & \tr_{M} (\trho_{B,i}) & = & \rho_{B,i}, & \textup{ for all } i \textup{ odd}, \\
& & \rho_{B,i} & = & \tr_{M} \left( U_{\B,i} \trho_{B,i-1} U_{\B,i}^* \right), & \textup{ for all } i \textup{ even}, \\
& & \rho_{B,i} & \in & \pos^B, & \textup{ for all } i, \\
& & \trho_{B,i} & \in & \pos^{M \times B}, & \textup{ for all } i \textup{ odd}. \\
\end{array} \]
The variables for a cheating Alice are also depicted in Figure~\ref{fig:protocolwithvariables}. These SDPs are referred to as Alice and Bob's cheating SDPs.
 
\begin{figure}[ht]
  \centering
   \includegraphics[width=5in]{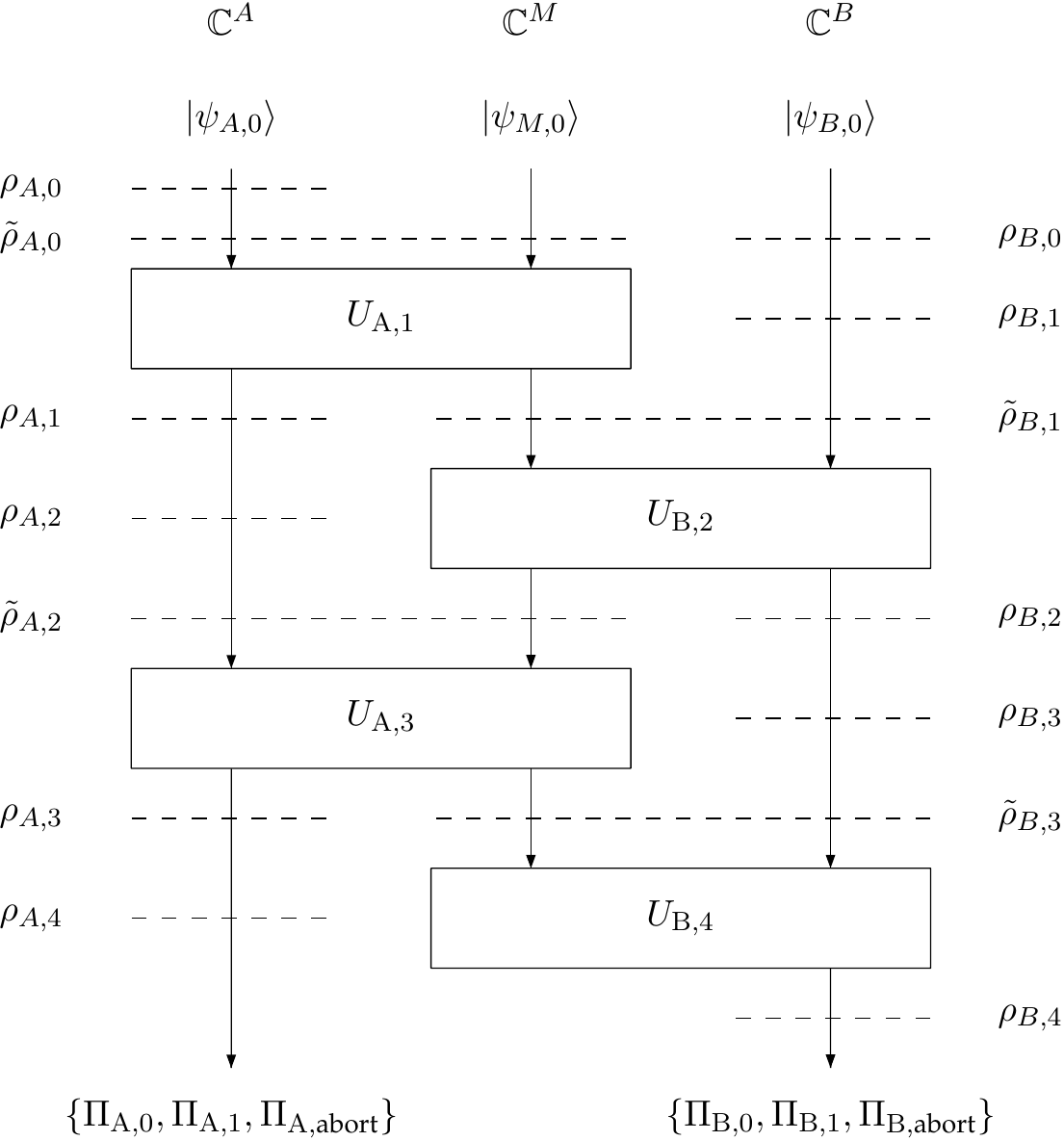} \caption{Four-round coin-flipping protocol with SDP variables depicted.}
\label{fig:protocolwithvariables}
\end{figure}

The duals of the above SDPs are as follows:
\[ \begin{array}{rrrcllllllllllllll}
\textrm{} & \inf    &      \inner{Z_{A,0}}{\kb{\psi_{A,0}}} &  \\
                     & \textrm{subject to} & Z_{A,i-1} \otimes \id_{M} & \succeq & U_{\A,i}^* (Z_{A,i} \otimes \id_{M}) U_{\A,i}, & \textup{ for all } i \textup{ odd}, \\
                     &                               & Z_{A,i-1} & = & Z_{A,i}, & \textup{ for all } i \textup{ even}, \\
                     &                               & Z_{A,n} & = & \Pi_{\A,c},
\end{array} \]
and
\[ \begin{array}{rrrcllllllllllllll}
\textrm{} & \inf    &      \inner{Z_{B,0}}{\kb{\psi_{B,0}}} &  \\
                     & \textrm{subject to} & Z_{B,i-1} \otimes \id_{M} & \succeq & U_{\B,i}^* (Z_{B,i} \otimes \id_{M}) U_{\B,i}, & \textup{ for all } i \textup{ even}, \\
                     &                               & Z_{B,i-1} & = & Z_{B,i}, & \textup{ for all } i \textup{ odd}, \\
                     &                               & Z_{B,n} & = & \Pi_{\B,c}.
\end{array} \]
 
We can derive a lower bound on the bias of any strong coin-flipping protocol by examining feasible dual solutions. Since the dual SDPs have strictly feasible solutions and the objective function is bounded on the feasible region, there is zero duality gap. Therefore, for Alice and Bob forcing outcome $0$, and for any $\delta > 0$, we can find feasible dual solutions $(Z_{B,0}, \ldots, Z_{B,n})$ and $(Z_{A,0}, \ldots, Z_{A,n})$, such that
\[ P_{\A,0}^* + \delta > \inner{Z_{B,0}}{\kb{\psi_{B,0}}} \aand P_{\B,0}^* + \delta > \inner{Z_{A,0}}{\kb{\psi_{A,0}}}. \]
Therefore, we have
\begin{eqnarray*}
\left( P_{\B,0}^* + \delta \right) \left( P_{\A,0}^* + \delta \right)
& > & \inner{Z_{B,0}}{\kb{\psi_{B,0}}} \inner{Z_{A,0}}{\kb{\psi_{A,0}}} \\
& = & \inner{Z_{A,0} \otimes \id_M \otimes Z_{B,0}}{\kb{\psi_{A,0}} \otimes \kb{\psi_{M,0}} \otimes \kb{\psi_{B,0}}} \\
& = & \inner{Z_{A,0} \otimes \id_M \otimes Z_{B,0}}{\kb{\psi_0}},
\end{eqnarray*} 
where we define $\ket{\psi_0}$ {to be the state at the beginning of the protocol when Alice and Bob are honest. Let $\ket{\psi_i}$} be the state after Bob applies $U_{\B,i}$ in an honest run of the protocol for $i$ even. We have from the dual constraints, for $i$ even,  
\begin{eqnarray*}
& & 
\inner{Z_{A,i} \otimes \id_M \otimes Z_{B,i}}{\kb{\psi_{i}}} \\ 
& \geq & \inner{U_{\A,i+1}^*(Z_{A,i+1} \otimes \id_M) U_{\A,i+1} \otimes Z_{B,i}}{\kb{\psi_{i}}} \\
& = & \inner{Z_{A,i+1} \otimes \id_M \otimes Z_{B,i}}{(U_{\A,i+1} \otimes \id_B) \kb{\psi_{i}}(U_{\A,i+1}^* \otimes \id_B)} \\
& = & \inner{Z_{A,i+2} \otimes \id_M \otimes Z_{B,i+1}}{(U_{\A,i+1} \otimes \id_B) \kb{\psi_{i}}(U_{\A,i+1}^* \otimes \id_B)} \\
& \geq & \inner{Z_{A,i+2} \otimes U_{\B,i+2}^* (\id_M \otimes Z_{B,i+2}) U_{\B,i+2}}{(U_{\A,i+1} \otimes \id_B)\kb{\psi_{i}}(U_{\A,i+1}^* \otimes \id_B)} \\
& = & \inner{Z_{A,i+2} \otimes \id_M \otimes Z_{B,i+2}}{\kb{\psi_{i+2}}}. 
\end{eqnarray*}
We can compute 
\[ \inner{Z_{A,n} \otimes \id_M \otimes Z_{B,n}}{\kb{\psi_{n}}} =
\inner{\Pi_{\A,0} \otimes \id_M \otimes \Pi_{\B,0}}{\kb{\psi_{n}}} = 1/2, \]
from condition~(\ref{CFcondition}). Taking the limit as $\delta \to 0$, we get
\[ P_{\B,0}^*P_{\A,0}^* \geq \half \implies \max \set{P_{\B,0}^*, P_{\A,0}^*} \geq \frac{1}{\sqrt 2} \implies \eps \geq \frac{1}{\sqrt 2} - \half. \]
This lower bound was later reproven by Gutoski and Watrous~\cite{GW07} using a different representation of quantum strategies.

Notice that we can reproduce the proof above using dual feasible solutions for Bob cheating towards $1$ and Alice cheating towards $0$. In this case, we get the final condition
\[ \inner{Z_{A,n} \otimes \id_M \otimes Z_{B,n}}{\kb{\psi_{n}}} =
\inner{\Pi_{\A,0} \otimes \id_M \otimes \Pi_{\B,1}}{\kb{\psi_{n}}} = 0. \]
This gives a trivial bound on the product of the cheating probabilities. However, Kitaev used this to create point games which we discuss below. We refer the reader to~\cite{Moc07} for the full details of the construction of general point games as all the details are not needed for this paper.

\subsection{Point games} \label{sssect:PG}

Let $\eig(Z)$ denote the set of eigenvalues for an operator $Z$ and let $\Pi_Z^{[\lambda]}$ denote the projection onto the eigenspace of $Z$ corresponding to eigenvalue $\lambda \in \eig(Z)$.
For a quantum state $\sigma \in \pos^n$, and $X, Y \in \pos^n$, denote by $\prob(X,Y, \sigma): \R^2 \to \R_+$ the function
\[ \Prob(X,Y, \sigma) := \sum_{\lambda \in \eig(X)} \sum_{\mu \in \eig(Y)} \inner{\Pi_{X}^{[\lambda]} \otimes \Pi_{Y}^{[\mu]}}{\sigma} \, \pg{\lambda}{\mu}, \]
where we use the notation $\pg{\lambda}{\mu} : \R^2 \to \R$ to denote the function that takes value $1$ on input $(\lambda, \mu)$ and $0$ otherwise. Note this function has finite support. Using this definition, we can create a point game from feasible dual variables as follows
\[ p_{n-i} :=
\Prob(Z_{B,i}, Z_{A,i}, \tr_{M} \kb{\psi_{i}}), \]
recalling that $\ket{\psi_{i}} \in \C^{A \times M \times B}$ is the state after Bob applies $U_{\B,i}$ in an honest run of the protocol.
Consider the dual SDPs for weak coin-flipping, i.e., Bob trying to force outcome $1$ and Alice trying to force outcome $0$. We can calculate $p_0 = \half \pg{0}{1} + \half \pg{1}{0}$, which acts as the starting point of the point game. Notice for any $\delta > 0$, there exists a large constant $\Lambda$ such that
\begin{equation} 
Z_{A,0} (\delta) := \left( \bra{\psi_{A,0}} Z_{A,0} \ket{\psi_{A,0}} + \delta \right) \kb{\psi_{A,0}} + \Lambda \left( \id - \kb{\psi_{A,0}} \right) \succeq Z_{A,0}, 
\label{eqn:delta} 
\end{equation} 
which can be proved using the Schur complement after writing $Z_{A,0}$ in a basis containing $\ket{\psi_{A,0}}$. Notice $(Z_{A,0}(\delta), Z_{A,1}, \ldots, Z_{A,n})$ is feasible if $(Z_{A,0}, Z_{A,1}, \ldots, Z_{A,n})$ is feasible and has the same objective function value as $\delta \to 0$. If we replace $Z_{A,0}$ with $Z_{A,0}(\delta)$, and replace $Z_{B,0}$ with the properly modified definition of $Z_{B,0}(\delta)$, we get that the final point is
\[ p_n = 1 \pg{\bra{\psi_{A,0}} Z_{A,0} \ket{\psi_{A,0}} + \delta}{\bra{\psi_{B,0}} Z_{B,0} \ket{\psi_{B,0}} + \delta}. \]
By strong duality, we see that we can choose the dual feasible solutions and $\delta$ such that this final point gets arbitrarily close to $\pg{P_{\B,1}^*}{P_{\A,0}^*}$.
 
A point game $p_0 \to p_1 \to \cdots \to p_n$ with final point $\pg{\zeta_{\B}}{\zeta_{\A}}$ can be defined independent of protocols. Define $[x] : \R \to \R$ to be the function that takes value $1$ on input $x$ and equals $0$, otherwise. Then $p_0 \to p_1 \to \cdots \to p_n$ is a point game with final point $\pg{\zeta_{\B}}{\zeta_{\A}}$, if each $p_i$ is a function with finite support, $p_0 = \half \pg{0}{1} + \half \pg{1}{0}$, $p_n = 1 \pg{\zeta_{\B}}{\zeta_{\A}}$, and the moves (or transitions) $p_i \to p_{i+1}$ have one of the following forms (possibly acting on only a subset of the points) 
\begin{itemize}
\item $\dsum_{a \in A} p_{i,a} \pg{x_{a}}{y} \to \sum_{b \in B} p_{i+1, b} \pg{z_{b}}{y}$ \quad (called a horizontal move),
\item $\dsum_{a \in A} p_{i,a} \pg{y}{x_{a}} \to \sum_{b \in B} p_{i+1, b} \pg{y}{z_{b}}$ \quad
(called a vertical move),
\end{itemize}
where $\dsum_{a \in A} p_{i,a} = \dsum_{b \in B} p_{i+1,b}$ (called conservation of probability) and
\[ \sum_{b \in B} p_{i+1, b} [z_b] - \sum_{a \in A} p_{i,a} [x_a] \in \mathrm{OMF}^*, \]
where $\mathrm{OMF}$ is the cone of operator monotone functions. The purpose of the last condition above is beyond the scope of this paper, but it is used to prove that if there is a point game with final point $\pg{\zeta_{\B}}{\zeta_{\A}}$, then for any
$\delta > 0$, there exists a coin-flipping protocol with $P_{\B,1}^* \leq \zeta_{\B} + \delta$ and $P_{\A,0}^* \leq \zeta_{\A} + \delta$ (see~\cite{Moc07} for details). Mochon proved that there exists a point game with final point $\pg{1/2 + \delta}{1/2 + \delta}$, for any $\delta > 0$, proving the existence of weak coin-flipping protocols with arbitrarily small bias. 
 
\section{A $\BCCF$-point game example with final point $[3/4, 3/4]$} 
\label{PGexample} 

In this section, we give an example $\BCCF$-protocol and give an (optimal) $\BCCF$-point game with final point $\pg{3/4}{3/4}$. It can be shown that all four cheating probabilities are equal to $3/4$, which is the best $\BCCF$-protocol we know how to construct to date and we conjecture is optimal based on numerical evidence~\cite{NST14}.

The $\BCCF$-protocol we consider is a four-round protocol defined by the parameters
\[ \alpha_0 := \alpha_1 := [1,0]^\top \aand \beta_0 := [1/2, 1/2, 0]^\top, \quad \beta_1 := [1/2, 0, 1/2]^\top. \]
Solving for the optimal dual solution, we get
\[ w_1 = [3/4, 0]^\top, \quad v_0 = [3/4, 0, 3/2]^\top, \quad v_1 =  [3/4, 3/2, 0]^\top \]
for cheating Bob and, for cheating Alice,
\[ z_1 = 3/4, \quad
z_2^{(0)} = [1/4, 0]^\top, \quad
z_2^{(1)} = [1/4, 0]^\top, \quad
z_2^{(2)} = [1/4, 0]^\top.
\]
   
The point game is as follows which can be visualized using Figures~\ref{QPG1}, \ref{QPG2}, and \ref{QPG3}.

\begin{pointgame}[$\BCCF$-point game example with final point $\pg{3/4}{3/4}$] \label{PGBMexample} 
\begin{align*}
\half \pg{0}{1} + \half \pg{1}{0}
& \to
\half \pg{0}{1} + \quarter \pg{\dfrac{3}{4}}{0} + \quarter \pg{\dfrac{3}{2}}{0} & \textup{Horizontal Split} \\
& \to
\quarter \pg{0}{1} + \quarter \pg{\dfrac{3}{4}}{1}
+ \quarter \pg{\dfrac{3}{4}}{0} + \quarter \pg{\dfrac{3}{2}}{0}
& \textup{Horizontal Raise} \\
& \to
\quarter \pg{0}{1} + \half \pg{\dfrac{3}{4}}{\dfrac{1}{2}} + \quarter \pg{\dfrac{3}{2}}{0}
& \textup{Vertical Merge} \\
& \to
\quarter \pg{0}{1} + \half \pg{\dfrac{3}{4}}{\dfrac{1}{2}} + \quarter \pg{\dfrac{3}{2}}{1}
& \textup{Vertical Raise} \\
& \to
\half \pg{\dfrac{3}{4}}{1} + \half \pg{\dfrac{3}{4}}{\dfrac{1}{2}}
& \textup{Horizontal Merge} \\
& \to
\pg{\dfrac{3}{4}}{\dfrac{3}{4}} & \textup{Vertical Merge} 
\end{align*}
\end{pointgame} 
A few things to note is that the four probability vectors defining the protocol do not have full support. Therefore, there are some points with ``$0$ probability''. For example, from the figures we would be tempted to think there should be a point $\pg{3/2}{1}$, but this point has $0$ probability and is thus not effectively there. For the same reasons the dual vectors do not have full support and thus we are able to have a point remain at $\pg{0}{1}$ after the first horizontal point raises. 
  
\section{Extra properties of $\BCCF$-protocols} 

In this section, we give some extra properties of $\BCCF$-protocols and of their cheating polytopes. 
 
\subsection{Extreme points of the cheating polytopes}

This subsection examines the extreme points of Alice and Bob's cheating polytopes which appear in both the quantum and classical cheating strategy formulations. We show that deterministic strategies correspond to the extreme points of the cheating polytopes. One can argue this directly from the properties of the protocol. However, we give a strictly algebraic proof based on the properties of the cheating polytopes.

\begin{definition}
An extreme point of a convex set $C$ is a point $x \in C$ such that if $x = \lambda y + (1-\lambda) z$, for $\lambda \in (0,1)$, $y \neq z$, then $y \not\in C$ or $z \not\in C$.
\end{definition}

We start with a well-known fact.
\begin{fact}
Suppose $\tilde{x} \in \{ x \geq 0: \Gamma x=b \}$. Then $\tilde{x}$ is an extreme point of $\{ x \geq 0: \Gamma x=b \}$ if and only if there does not exist nonzero $u \in \nulll(\Gamma)$ with $\supp(u) \subseteq \supp(\tilde{x})$.
\end{fact}

We can use this fact to prove the following lemma.

\begin{lemma}
Suppose $(p_1, \ldots, p_n) \in \calP_{\B}$ and $(s_1, \ldots, s_n, s) \in \calP_{\A}$. Then the vectors are extreme points of their respective polytopes if and only if they are Boolean, i.e., all of their entries are $0$ or $1$.
\end{lemma}

\begin{proof}
We prove it for Bob's cheating polytope as the proof for Alice's is nearly identical. Suppose $(p_1, \ldots, p_n) \in \calP_{\B}$ is Boolean, we show it is an extreme point. Let Bob's polytope $\calP_{\B}$ be represented by the linear system $\Gamma (p_1, \ldots, p_n) = b$, $(p_1, \ldots, p_n) \geq 0$. Let $(u_{1}, \ldots, u_{n}) \in \Null(\Gamma)$ satisfy $\supp(u_{1}, \ldots, u_{n}) \subseteq \supp(p_1, \ldots, p_n)$. We argue that $(u_{1}, \ldots, u_{n})$ must be the zero vector. The constraint on $p_1$ is $\sum_{y_1} p_{1, x_1,y_1} = 1$ for all $x_1 \in A_1$. Therefore, since $p_1$ is Boolean, there is exactly one value of $y_1$ for every $x_1$ such that $p_{1, x_1, y_1} = 1$. These are the only entries of $u_1$ that can be nonzero, but since $(u_1, \ldots, u_n) \in \Null(\Gamma)$ we must have that entry equal to $0$. We can repeat this argument to get $u_i = 0$ for all $i \in \set{1, \ldots, n}$. Therefore, $(p_1, \ldots, p_n)$ is an extreme point.

Conversely, suppose $(p_1, \ldots, p_n) \in \calP_{\B}$ is not Boolean. Let $i$ be the smallest index where $p_i$ is not Boolean.
If $i > 1$, define $u_j := 0$ for $j \in \set{1, \ldots, i-1}$.
Let $(\hat{x}_1, \hat{y}_1, \ldots, \hat{x}_i, \hat{y}_i)$ be an index such that $p_{i, \hat{x}_1, \hat{y}_1, \ldots, \hat{x}_i, \hat{y}_i} \in (0,1)$. From the constraints, we must have another $\hat{y}'_{i}$ such that ${p_{i, \hat{x}_1, \hat{y}_1, \ldots, \hat{x}_i, \hat{y}'_i} \in (0,1)}$ as well (since they must add to $1$).
Now define $u_{i, \hat{x}_1, \hat{y}_1, \ldots, \hat{x}_i, \hat{y}_i} := t$, for some $t \neq 0$, and $u_{i, \hat{x}_1, \hat{y}_1, \ldots, \hat{x}_i, \hat{y}'_i} := -t$, and the rest of the entries of $u_i$ to be $0$.
We define $u_{i+1}$ to be equal to $p_{i+1}$, but we scale each entry such that
\[ \tr_{B_{i+1}}(u_{i+1}) = u_i \otimes e_{A_{i+1}}. \]
We inductively define $u_j$ in this way for all $j \in \set{i+2, \ldots, n}$.
{Thus}, since we scaled $(p_1, \ldots, p_n)$ to get $(u_1, \ldots, u_n)$, we have that $\supp(u_1, \ldots, u_n) \subseteq \supp(p_1, \ldots, p_n)$ and also $(u_1, \ldots, u_n) \in \Null(\Gamma)$ implying $(p_1, \ldots, p_n)$ cannot be an extreme point. \qed
\end{proof}
 
We see that extreme points of the cheating polytopes correspond to the strategies where Alice and Bob choose their next bit deterministically depending on the bits revealed.

\begin{corollary}
In a classical $\BCCF$-protocol, Alice and Bob each have an optimal cheating strategy which is deterministic.
\end{corollary}

\begin{proof}
In a linear program whose feasible region does not contain lines, if there exists an optimal solution then there exists an optimal solution which is an extreme point of the feasible region. The result follows since the feasible region is nonempty and compact implying the existence of an optimal solution. \qed
\end{proof} 
 
\subsection{A succinct way to write the duals of the reduced formulations} \label{succinct} 

In this subsection, we present a simple form for the duals of the reduced cheating SDPs. We show that we only need to consider the variables in the positive semidefiniteness  constraints, since the linear inequalities reveal how to optimally assign the rest of the variables. Sometimes it is easier to work with the succinct form developed in this section because handling many dual variables can overcomplicate simple ideas. For example, in Appendix~\ref{qubits}, we show  that the smallest bias attainable by $\BCCF$-protocols is not affected if we restrict $\BCCF$-protocols to have $2$-dimensional (qubit) messages.

{Recall the dual of Bob's reduced cheating SDP for forcing outcome $0$, below}
\[ \begin{array}{rrrcllllllllllllll}
\textrm{} & \inf                         & \tr_{A_1}(w_1) \\
                     & \textup{subject to} & w_1 \otimes e_{B_1} & \geq & \tr_{A_2}(w_2), \\
                     & & w_2 \otimes e_{B_2} & \geq & \tr_{A_3}(w_3), \\
                     & & & \vdots \\
                     &                               & w_{n} \otimes e_{B_{n}} & \geq & \half \sum_{a \in \zo} \alpha_a \otimes v_a, \\ 
                     & & \Diag(v_a) & \succeq & \sqrtt{\beta_a}, & \forAll a \in \zo. 
\end{array} \]
Let us examine the first constraint $w_1 \otimes e_{B_1} \geq \tr_{A_2}(w_2)$. This is equivalent to 
\[ {w_{1,x_1} \geq \sum_{x_2 \in A_2} w_{2, x_1, y_1, x_2}} \] 
for all $x_1 \in A_1$, $y_1 \in B_1$. Once we fix a value for $w_2$, an optimal choice of $w_1$ is 
\[ w_{1, x_1} = \max_{y_1 \in B_1} \sum_{x_2 \in A_2} w_{2, x_1, y_1, x_2}. \] 
Using this idea, we can rewrite Bob's dual as
\[ \displaystyle\inf_{\Diag(v_a) \succeq \sqrtt{\beta_a}} \left\{ \sum_{x_1 \in A_1} \max_{y_1 \in B_1} \sum_{x_2 \in A_2} \max_{y_2 \in B_2} \cdots \sum_{x_n \in A_n} \max_{y_n \in B_n} \sum_{a \in \zo} \half \alpha_{a,x} v_{a,y} \right\} \]
and Alice's as
\[ \displaystyle\inf_{\Diag(z_{n+1}^{(y)}) \succeq \half \beta_{a,y} \sqrtt{\alpha_a}} \left\{ \max_{x_1 \in A_1} \sum_{y_1 \in B_1} \cdots \max_{x_n \in A_n} \sum_{y_n \in B_n} z_{n+1, x,y} \right\}, \]
each for forcing outcome $0$. We can switch $\beta_0$ with $\beta_1$ to get the succinct forms for Alice and Bob forcing outcome $1$.

This shows that the objective values are determined once some of the dual variables are fixed. We see this idea when designing the point games corresponding to $\BCCF$-protocols.
  
\subsection{An SDP proof for why qubit messages are sufficient} \label{qubits}

In this subsection, we show how the succinct representation of the duals helps us prove a novel result, that we can bound the dimension of the messages in $\BCCF$-protocols without increasing the smallest attainable bias.

We use the reduced cheating SDPs to prove that we can assume $A_i = B_i = \zo$, that is, each message is a single qubit. More specifically, we show that for any $\BCCF$-protocol, there exists another $\BCCF$-protocol with qubit messages where the bias is no greater. We prove it for Alice's messages as the proof for Bob's messages is nearly identical.

Suppose we have a protocol defined by
\[ A = A_1 \times \cdots \times A_n, \; B = B_1 \times \cdots \times B_n, \; \alpha_0, \alpha_1 \in \prob^A, \;\beta_0, \beta_1 \in \prob^B. \]
Suppose Alice's $i$'th message has large dimension, that is, $|A_i| > 2$. We define a new protocol by replacing $A_i$ with $A'_i \times A''_i$, where $|A_i| \leq |A'_i \times A''_i|$. Notice that $\alpha_{0}$ and $\alpha_1$ can be viewed as probability distributions over $A_1 \times \cdots \times A_{i-1} \times A'_i \times A''_i \times A_{i+1} \times \cdots \times A_n$ in the obvious way. We also add a ``dummy'' message from Bob by adding $B_{d}$ in between $B_i$ and $B_{i+1}$. This dummy message needs to be independent of the protocol, so we can suppose Bob sends $\ket{0}$. This effectively replaces $\beta_b$ with $\beta'_b := \beta_b \otimes [1, 0]^{\T}_{d}$, for each $b \in \zo$. If Alice and Bob cannot cheat more in this new protocol, then we can repeat these arguments to show that all of Alice's messages are qubit messages by inductively breaking up the $\C^{A_i}$ spaces.

\paragraph{Bob's cheating probabilities do not increase} \quad \\

We now show that Bob cannot use the extra message to cheat more in the new protocol. We show this by constructing a dual feasible solution.

In the original protocol, cheating Bob can force outcome $0$ with maximum probability given by the optimal objective value of the following problem
\[ \displaystyle\inf_{\Diag(v_a) \succeq \sqrtt{\beta_a}} \left\{ \sum_{x_1 \in A_1} \max_{y_1 \in B_1} \sum_{x_2 \in A_2} \max_{y_2 \in B_2} \cdots \sum_{x_n \in A_n} \max_{y_n \in B_n} \sum_{a \in \zo} \half \alpha_{a,x} v_{a,y} \right\}. \]

In the new protocol, cheating Bob can force outcome $0$ with maximum probability given by the optimal objective value of the following problem
\[ \displaystyle\inf_{\Diag(\tilde{v}_a) \succeq \sqrtt{\beta'_a}} \left\{ \sum_{x_1} \max_{y_1} \sum_{x_2} \max_{y_2} \cdots \sum_{x'_i \in A'_i} \max_{y_{d} \in B_{d}} \sum_{x''_i \in A''_i} \cdots \sum_{x_n} \max_{y_n} \sum_{a \in \zo} \half \alpha_{a,x} \tilde{v}_{a,y} \right\}. \]
For any $(v_0, v_1)$ feasible in the first problem, we can define a solution feasible in the second problem $(\tilde{v}_0, \tilde{v}_1) := \left( v_0 \otimes [1,0]^{\T}_{d}, v_1 \otimes [1,0]^{\T}_{d} \right)$ with the same objective function value. Notice the same argument holds if we switch $\beta_0$ with $\beta_1$ and $\beta'_0$ with $\beta'_1$, i.e., if Bob wants outcome $1$. Since these are minimization problems, Bob can cheat no more in the new protocol.

\paragraph{Alice's cheating probabilities do not increase} \quad \\

We now show that Alice cannot use her extra message to cheat more in the new protocol. To show this, we repeat the same argument as in the case for cheating Bob.

In the original protocol, cheating Alice can force outcome $0$ with maximum probability given by the optimal objective value of the following problem
\[ \displaystyle\inf_{\Diag(z_{n+1}^{(y)}) \succeq \half \beta_{a,y} \sqrtt{\alpha_a}} \left\{ \max_{x_1 \in A_1} \sum_{y_1 \in B_1} \cdots \max_{x_n \in A_n} \sum_{y_n \in B_n} z_{n+1,x,y} \right\}. \]

In the new protocol, cheating Alice can force outcome $0$ with maximum probability given by the optimal objective value of the following problem
\[ \displaystyle\inf_{\Diag(\tilde{z}_{n+1}^{(y)}) \succeq \half \beta'_{a,y} \sqrtt{\alpha_a}} \left\{ \max_{x_1 \in A_1} \sum_{y_1 \in B_1} \cdots \max_{x'_i \in A'_i} \sum_{y_{d} \in B_{d}} \max_{x''_i \in A''_i} \cdots \max_{x_n \in A_n} \sum_{y_n \in B_n} \tilde{z}_{n+1,x,y} \right\}. \]
For any $z_{n+1}$ feasible in the first problem, we can define a solution feasible in the second problem $\tilde{z}_{n+1} := z_{n+1} \otimes [1,0]^{\T}_{d}$ with the same objective function value. Notice the same argument holds if we switch $\beta_0$ with $\beta_1$ and $\beta'_0$ with $\beta'_1$, i.e., if Alice wants outcome $1$. Since these are minimization problems, Alice can cheat no more in the new protocol.

\section{Proof of correctness for the reduced problems}
\label{app:reduced_app} 
 
{We start with a technical lemma whose proof is obvious using Equation~(\ref{eqn:delta}).}  

\begin{lemma}[Subspace lemma] \label{SubspaceLemma}
For a vector $\ket{\psi} \in \C^n$, a set $S \subseteq \Herm^n$, and a continuous, monotonically nondecreasing function $F$, we have
\[ \inf_{X, Y \in \Herm^n} \! \{ F(\bra{\psi} X \ket{\psi}) : X \succeq Y, \, Y \in S \} = \inf_{X, Y \in \Herm^n} \! \{ F(\bra{\psi} X \ket{\psi}) : \bra{\psi} X \ket{\psi} \geq \bra{\psi} Y \ket{\psi}, \, Y \in S \}. \]
\end{lemma}
 
This lemma can be generalized. We can use this lemma whenever the constraint on $X$ is satisfied by replacing it with $X(\delta)$ (from Equation~(\ref{eqn:delta})) for $\delta > 0$. The most complicated constraints that arise later in this paper are of the form
\[ \sum_{x \in A} W_{x,y} \otimes \kb{x} \otimes \id_B \succeq C, \]
where $W_{x,y}$ are the variables and the objective function is continuous and nondecreasing on the value of $\bra{\phi} W_{x,y} \ket{\phi}$. We see that a necessary condition is
\[ \sum_{x \in A} \bra{\phi} W_{x,y} \ket{\phi} \cdot \kb{x} \otimes \id_B \succeq (\bra{\phi} \otimes \id_A \otimes \id_B) \, C \,  (\ket{\phi} \otimes \id_A \otimes \id_B). \]
By using a properly modified definition for $W_{x,y}(\delta)$, we have that this condition is also sufficient. The idea is to increase the eigenvalues on subspaces that do not affect the objective function. 

\subsection{On the structure of the proofs}

Here we prove that the cheating SDPs can have a certain, restricted form while retaining the same optimal objective function value.
That is, we cut down the feasible region to something that is much cleaner and illustrates the simple communication of the protocol.
The main technique used in proving that we do not cut off all of the optimal solutions comes from duality theory of semidefinite programming.
We generalize the following idea.
If we wish to prove that a certain feasible solution is optimal for the primal problem, it suffices to exhibit a feasible dual solution with the same objective function value. Here, we claim that a restricted feasible region contains an optimal solution.
Let $p_1^*$ be the optimal value of the original SDP, $p_2^*$ be the optimal value of the restricted SDP, and let $d_1^*$ and $d_2^*$ be the optimal values of the respective dual problems and assume all of them are finite. We want to show that $p_1^* = p_2^*$. Suppose the restricted problem and its dual have zero duality gap. Then if we can prove that $d_1^* \leq d_2^*$, we have
\[ p_1^* \leq d_1^* \leq d_2^* = p_2^* \leq p_1^*, \]
proving $p_1^* = p_2^*$ as desired. To show $d_1^* \leq d_2^*$, it suffices to find a restriction of the dual of the original SDP to get to a problem equivalent to the dual of the restricted SDP. This is depicted in Figure~\ref{PDFR3}.
 
\begin{figure}[ht]
  \centering
   \includegraphics[width=5in] {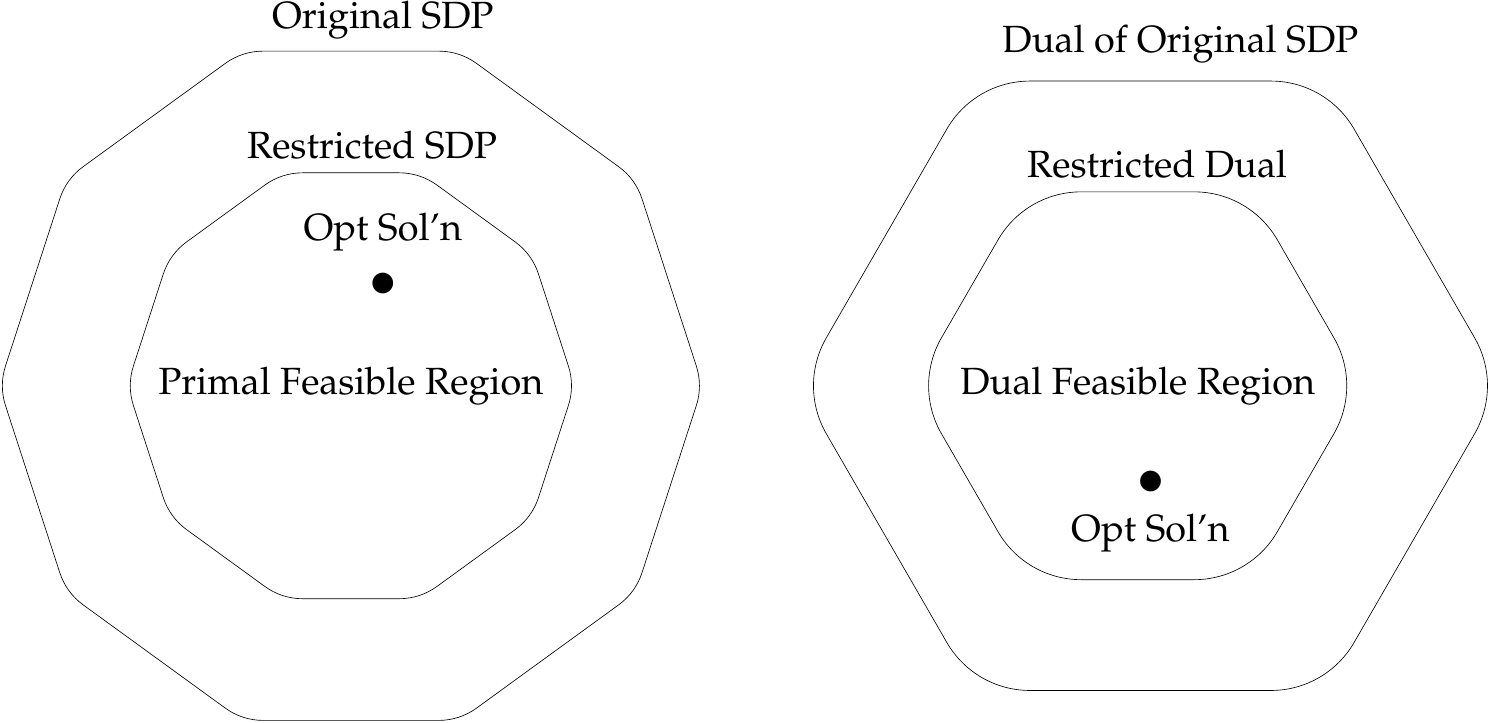} 
  \caption{There exist optimal solutions in the restricted feasible regions.}
\label{PDFR3}
\end{figure}
  
\subsection{Proof of Theorem~\ref{thm:reducedBob}} 

{The contexts of the ``reduced strategies'' are very simple, Alice or Bob simply change the probability of which the next message is chosen, controlled on the messages sent and received so far (doing so in superposition). This is a very simple form, their cheating is certainly not limited to such strategies. However, we show here that such strategies are optimal, starting with a cheating Bob.} 

{We now restrict the feasible region of Bob's cheating SDPs by defining the following parameterized primal feasible solutions: } 
\[
\bar{\rho}_j := \sum_{x_1 \in A_1} \cdots \sum_{x_j \in A_j} \kb{x_1, \ldots, x_j} \otimes \kb{\psi_{x_1, \ldots, x_j}} \otimes \Diag(p_j),
\]
for $j \in \{ 1, \ldots, n \}$, and 
\[ \bar{\rho}_F := \sum_{a \in \zo} \kb{aa} \otimes \kb{\psi'_a}, \]  
where $p_j \in \R_+^{A_1 \times B_1 \times \cdots \times A_j \times B_j}$ is a variable,
\[ \ket{\psi_{x_1, \ldots, x_j}} := \frac{1}{\sqrt 2} \sum_{x_{j+1} \in A_{j+1}} \cdots \sum_{x_n \in A_n} \sum_{a \in \zo} \sqrt{\alpha_{a,x}} \, \ket{aa} \ket{x_{j+1}, \ldots, x_n}\ket{x_{j+1}, \ldots, x_n}, \]
and 
\[ \ket{\psi'_a} := \sum_{y \in B} \sqrt{\half \sum_{x \in A} \alpha_{a,x} [p_n]_{x,y}} \, \ket{yy}, \] 
for all $a \in \zo$. The new objective function for forcing outcome $0$ becomes
\[ \inner{\bar{\rho}_F}{\Pi_{\A, 0}} = \half \sum_{a \bit} \, \rF \left( (\alpha_a \otimes \id_A)^{\T} p_n, \, \beta_a \right) \] 
and the variables $(p_1, \ldots, p_n)$ belong to Bob's cheating polytope as defined in Definition~\ref{BobsPolytope}. 
 
Since we have restricted the feasible region of a maximization SDP, we have proved that
\[ P_{\B,0}^* \geq \max \left\{ \half \sum_{a \bit} \, \rF \left( (\alpha_a \otimes \id_B)^{\T} p_n, \, \beta_a \right) : (p_{1}, \ldots, p_{n}) \in \calP_{\B} \right\}. \]
By changing the value of $\bar{\rho}_F \in \pos^{A_0 \times B'_0 \times B \times B'}$ above to $\bar{\rho}_F = \sum_{a \in \zo} \kb{a \bar{a}} \otimes \kb{\psi'_a}$, 
we get
\[ P_{\B,1}^* \geq \max \left\{ \half \sum_{a \bit} \, \rF \left( (\alpha_a \otimes \id_B)^{\T} p_n, \, \beta_{\bar{a}} \right) : (p_{1}, \ldots, p_{n}) \in \calP_{\B} \right\}. \]
This swaps Bob's choice of commitment reveal in the last message.
 
We now show that these inequalities hold with equality by exhibiting a family of feasible dual solutions with matching optimal objective function value. 
 
We begin by proving this for the case of $P_{\B,0}^*$. Consider the dual of Bob's cheating SDP below:
\[ \begin{array}{rrrcllllllllllllll}
\textrm{} &  P_{\B, 0}^* \; = \;  \inf                             & \inner{W_1}{\tr_{A_1} \kb{\psi}} \\
                     & \textup{subject to} & W_j \otimes \id_{B_j} & \succeq & W_{j+1} \otimes \id_{A_{j+1}}, \\
                     & & & & \forAll \, j \in \set{1, \ldots, n-1}, \\
                     &                               & W_n \otimes \id_{B_n} & \succeq & W_{n+1} \otimes \id_{A'} \otimes \id_{A'_0}, \\
                     &                               & W_{n+1} \otimes \id_{B'} \otimes \id_{B'_0} & \succeq & \Pi_{\A,0}, \\
                     &                               & W_j & \in & \mathbb{S}^{A_0 \times A'_0 \times B_1 \times \cdots \times B_{j-1} \times A_{j+1} \times \cdots \times A_n \times A'}, \\
                     & & & & \forAll j \in \{ 1, \ldots, n \}, \\
                     &                               & W_{n+1} & \in & \mathbb{S}^{A_0 \times B}. \\
\end{array} \]

We now define a restriction of the following form:
\[ W_j := \! \sum_{x_1 \in A_1} \sum_{y_1 \in B_1} \cdots \!\!\!\! \sum_{y_{j-1} \in B_{j-1}} \sum_{x_j \in A_j} \kb{x_1, y_1, \ldots, y_{j-1}, x_j} \otimes W_{j, x_1, y_1, \ldots, y_{j-1}, x_j}, \]
for $j \in \{ 1, \ldots, n \}$, and
\[ W_{n+1} := \sum_{a \in \zo} \kb{a} \otimes \Diag(v_a). \]
Under this restriction, we have the following problem:
\[ \begin{array}{rrrcllllllllllllll}
\textrm{} & d_2^* \; = \; \inf                             & \dsum_{x_1 \in A_1} \inner{W_{1, x_1}}{\kb{\psi_{x_1}}} \\
                     & \textup{subject to} & W_{j, x_1, y_1, \ldots, y_{j-1}, x_j} & \succeq & \displaystyle\sum_{x_{j+1}} \kb{x_{j+1}} \otimes \id_{A_{j+1}} \otimes W_{j+1, x_1, y_1, \ldots, y_{j}, x_{j+1}}, \\ 
                     & & & & \forAll \, j \in \set{1, \ldots, n-1}, \\
                     & & & & (x_1, \ldots, x_{j}) \in A_1 \times \cdots \times A_{j}, \\
                     & & & & (y_1, \ldots, y_j) \in B_1 \times \cdots \times B_j, \\
                     &                               & W_{n, x_1, y_1, \ldots, y_{n-1}, x_n} & \succeq & \displaystyle\sum_{a \in \zo} v_{a,y} \, \kb{a} \otimes \id_{A'_0}, \\
                     &                               & \Diag(v_a) & \succeq & \sqrtt{\beta_a}, \forAll a \in \zo, \\
\end{array} \]
where the last constraint was obtained using Lemma~\ref{FidLemma2}. Note that this shows $d_2^* \geq P_{\B, 0}^*$. 

The last constraint changes to $\Diag(v_a) \succeq \sqrtt{\beta_{\bar{a}}}, \forAll a \in \zo$, if Bob is cheating towards $1$ and the rest of the proof follows similarly in this case.

Since the objective function only depends on $W_{1, x_1}$ in the subspace $\kb{\psi_{x_1}}$, we apply the Subspace lemma (Lemma~\ref{SubspaceLemma}) to the first constraint and replace it with
\begin{eqnarray*}
\bra{\psi_{x_1}} W_{1, x_1} \ket{\psi_{x_1}}
& \geq & \bra{\psi_{x_1}} \sum_{x_2 \in A_2} \kb{x_{2}} \otimes \id_{A_{2}} \otimes W_{2, x_1, y_1, x_2} \ket{\psi_{x_1}} \\
& = & \sum_{x_{2} \in A_2} \bra{\psi_{x_1, x_2}} W_{2, x_1, y_1, x_2} \ket{\psi_{x_1, x_2}}.
\end{eqnarray*}
Examining the next constraint, we need to choose $W_{2, x_1, y_1, x_2}$ to satisfy
\[ W_{2, x_1, y_1, x_2} \succeq \sum_{x_3 \in A_3} \kb{x_3} \otimes \id_{A_3} \otimes W_{3, x_1, y_2, x_2, y_2, x_3}. \]
Since the objective function value only depends on $\bra{\psi_{x_1, x_2}} W_{2, x_1, y_1, x_2} \ket{\psi_{x_1, x_2}}$, we can repeat the same argument and replace the constraint by
\[ \bra{\psi_{x_1, x_2}} W_{2, x_1, y_1, x_2} \ket{\psi_{x_1, x_2}} \geq \sum_{x_3 \in A_3} \bra{\psi_{x_1, x_2, x_3}} W_{3, x_1, y_2, x_2, y_2, x_3} \ket{\psi_{x_1, x_2, x_3}}. \]
Continuing in this fashion, we can replace each constraint to get the following problem with the same optimal  objective value: 
\[ \begin{array}{rrrcllllllllllllll}
\textrm{} &    \inf                             & \dsum_{x_1 \in A_1} \inner{W_{1, x_1}}{\kb{\psi_{x_1}}} \\
                     & \textrm{s.t.} & \bra{\psi_{x_1, \ldots, x_j}} W_{j, x_1, y_1, \ldots, y_{j-1}, x_j} \ket{\psi_{x_1, \ldots, x_j}}
                      & \geq & \displaystyle\sum_{x_{j+1}} \bra{\psi_{x_1, \ldots, x_{j+1}}} W_{j+1, x_1, y_1, \ldots, y_{j}, x_{j+1}} \ket{\psi_{x_1, \ldots, x_{j+1}}} \\
                     & & & & \forAll \, j \in \set{1, \ldots, n-1}, \\
                     & & & & (x_1, \ldots, x_{j+1}) \in A_1 \times \cdots \times A_{j+1}, \\
                     & & & & (y_1, \ldots, y_j) \in B_1 \times \cdots \times B_j, \\
                     &                               & \bra{\psi_x} W_{n, x_1, y_1, \ldots, y_{n-1}, x_n} \ket{\psi_x} & \geq & \displaystyle\sum_{a \in \zo} \alpha_{a,x} \, v_{a,y}, \forAll x \in A, y \in B, \\
                     &                               & \Diag(v_a) & \succeq & \sqrtt{\beta_a}, \forAll a \in \zo. \\
\end{array} \]
Define
\[ w_{j, x_1, y_1, \ldots, y_{j-1}, x_j} := \bra{\psi_{x_1, \ldots, x_j}} W_{j, x_1, y_1, \ldots, y_{j-1}, x_j} \ket{\psi_{x_1, \ldots, x_j}}, \]
for all
$j \in \set{1, \ldots, n-1}, (x_1, \ldots, x_{j+1}) \in A_1 \times \cdots \times A_{j+1}, (y_1, \ldots, y_j) \in B_1 \times \cdots \times B_j$,
to get the equivalent problem
\[ \begin{array}{rrrcllllllllllllll}
\textrm{} & d_2^* \; = \; \inf                             & \displaystyle\sum_{x_1 \in A_1} w_{1, x_1} \\
                     & \textrm{subject to} & w_{j, x_1, y_1, \ldots, y_{j-1}, x_j}
                      & \geq & \!\! \displaystyle\sum_{x_{j+1} \in A_{j+1}} \!\! w_{j+1, x_1, y_1, \ldots, y_{j}, x_{j+1}}, \\
                     & & & & \forAll \, j \in \set{1, \ldots, n-1}, \\
                     & & & & (x_1, \ldots, x_{j+1}) \in A_1 \times \cdots \times A_{j+1}, \\
                     & & & & (y_1, \ldots, y_j) \in B_1 \times \cdots \times B_j, \\
                     &                               & w_{n, x_1, y_1, \ldots, y_{n-1}, x_n} & \geq & \displaystyle\sum_{a \in \zo} \half \alpha_{a,x} \, v_{a,y},  \forAll x \in A, y \in B, a \in \zo, \\
                     &                               & \Diag(v_a) & \succeq & \sqrtt{\beta_a}, \forall a \in \zo, \\
\end{array} \] 
{noting $w_{j,x_1, y_1, \ldots, y_{j-1}, x_j} = 0$ when $\ket{\psi_{x_1, \ldots, x_j}} = 0$ can be assumed in an optimal solution.}
This problem has a strictly feasible solution and the objective function is bounded from below on the feasible region, thus strong duality holds and there is zero duality gap. The dual of this problem is
\[ \max_{\substack{(p_1, \ldots, p_n) \in \calP_{\B} \\ \rho_0, \rho_1 \in \mathbb{S}_+^B}} \left\{ \sum_{a \in \zo} \half \inner{\rho_a}{\sqrtt{\beta_a}} :  \diag(\rho_a) = (\alpha_a \otimes \id_B)^{\T} p_n, \, \forall a \in \zo \right\}, \]
which has optimal value $d_2^*$ due to zero duality gap. This problem is equivalent to the reduced problem by Lemma~\ref{FidelityLemma}. Therefore, we have $P_{\B, 0}^* \leq d_2^* \leq P_{\B,0}^*$ implying $P_{\B, 0}^* = d_2^*$  which is the optimal value of the reduced problem, as desired. \qed
 
\subsection{Proof of Theorem~\ref{thm:reducedAlice}}  

We now restrict the feasible region of Alice's cheating SDPs by defining the following parameterized primal feasible solutions. Intuitively, this strategy is similar to that of cheating Bob. The solution is given below
\[ {\bar{\sigma}_j} := \sum_{y_1 \in B_1} \cdots \sum_{y_{j-1} \in B_{j-1}} \kb{y_1, \ldots, y_{j-1}} \otimes \kb{\phi_{y_1, \ldots, y_{j-1}}} \otimes \Diag(s_j), \]
for $j \in \set{2, \ldots, n}$, and 
\[ \bar{\sigma}_F := \sum_{a \in A'_0} \sum_{y \in B} \kb{a} \otimes \kb{y} \otimes \kb{\phi_y} \otimes \kb{\phi'_{a,y}}, \]  
where $s_{j} \in \R_+^{A_1 \times B_1 \times \cdots \times B_{j-1} \times A_j}$ and $s \in \R_+^{A'_0 \times A \times B}$ are variables,
\[ \ket{\phi_{y_1, \ldots, y_{j-1}}} := \frac{1}{\sqrt 2} \sum_{y_{j} \in B_j} \cdots \sum_{y_n \in B_n} \sum_{b \in \zo} \sqrt{\beta_{b,y}} \, \ket{bb} \ket{y_{j}, \ldots, y_n}\ket{y_{j}, \ldots, y_n}, \]
and 
\[ \ket{\phi'_{a,y}} := \sum_{x \in A} \sqrt{s_{a,y,x}} \, \ket{xx}, \] 
for all $y \in B, a \in \zo$. With this restriction, we have 
\[ \inner{\bar{\sigma}_F}{\Pi_{\B, 0} \otimes \id_{B'_0 \times B'}} = \half \sum_{a \in A'_0} \sum_{y \in B'} \beta_{a,y} \,  \rF(s^{(a,y)}, \alpha_a) \] 
as the new objective function for forcing outcome $0$ where $s^{(a,y)} \in \C^A$ is defined as the restriction of $s$ with $a$ and $y$ fixed. We can define it element-wise as $[s^{(a,y)}]_x := s_{a,y,x}$.
The new objective function for forcing outcome $1$ is 
\[ \inner{\bar{\sigma}_F}{\Pi_{\B, 1} \otimes \id_{B'_0 \times B'}} = \half \sum_{a \in A'_0} \sum_{y \in B'} \beta_{\bar{a},y} \,  \rF(s^{(a,y)}, \alpha_a). \]  
The variables $(s_1, \ldots, s_n, s)$ belong to Alice's cheating polytope as defined in Definition~\ref{AlicePolytope}. 
 
We have proved 
\[ P_{\A,0}^* \geq \max_{(s_{1}, \ldots, s_{n}, s) \in \calP_{\A}} \left\{ \half \sum_{a \bit} \sum_{y \in B} \beta_{a,y} \; \rF(s^{(a,y)}, \alpha_{a}) \right\} \] 
and 
\[ P_{\A,1}^* \geq \max_{(s_{1}, \ldots, s_{n}, s) \in \calP_{\A}} \left\{ \half \sum_{a \bit} \sum_{y \in B} \beta_{\bar{a},y} \;  \rF(s^{(a,y)}, \alpha_{a}) \right\}. \] 

{We now show that the above inequalities hold with equality by exhibiting a family of feasible dual solutions with matching optimal objective function value.} 
 
Consider the dual to Alice's cheating SDP for forcing outcome $0$, below: 
\[ \begin{array}{rrrcllllllllllllll}
\textrm{} & P_{\A, 0}^* \; = \; \inf                                & \inner{Z_1}{\kb{\phi}}  \\
                     & \textup{subject to} & Z_j \otimes \id_{A_j} & \succeq & Z_{j+1} \otimes \id_{B_j}, \\
                     & & & & \forAll j \in \set{1, \ldots, n}, \\
                     &                               & Z_{n+1} \otimes \id_{A'} \otimes \id_{A'_0} & \succeq & \Pi_{\B,0} \otimes \id_{B'_0} \otimes \id_{B'}, \\
                     & & Z_j & \in & \mathbb{S}^{B_0 \times B'_0 \times A_1 \times \cdots \times A_{j-1} \times B_j \times \cdots \times B_n \times B'}, \\
                     & & & & \forAll j \in \{ 1, \ldots, n, n+1 \}. 
\end{array} \]
Consider the following restriction:
\[ Z_{j+1} := \sum_{x_1 \in A_1} \sum_{y_1 \in B_1} \cdots \sum_{x_j \in A_j} \sum_{y_j \in B_j} \kb{x_1, y_1, \ldots, x_j, y_j} \otimes Z_{j+1, x_1, y_1, \ldots, x_j, y_j}, \]
for $j \in \set{1, \ldots, n}$.
Substituting this into the constraints, we get the following new problem
\[ \begin{array}{rrrcllllllllllllll}
\textrm{} & d_2^* \; = \; \inf                                & \inner{Z_1}{\kb{\phi}}  \\
                     & \textup{subject to} & Z_1 & \succeq & \dsum_{y_1 \in B_1} \kb{y_1} \otimes \id_{B_1} \otimes Z_{2, x_1, y_1}, \\
                     & & Z_{j, x_1, y_1, \ldots, x_{j-1}, y_{j-1}} & \succeq & \displaystyle\sum_{y_{j} \in B_j} \kb{y_{j}} \otimes \id_{B_j} \otimes Z_{j+1, x_1, y_1, \ldots, x_{j}, y_{j}}, \\
                     & & & & \forAll j \in \{ 2, \ldots, n \}, \\
                     & & & & (x_1, \ldots, x_{j}) \in A_1 \times \cdots \times A_{j}, \\
                     & & & & (y_1, \ldots, y_{j}) \in B_1 \times \cdots \times B_{j}, \\
                     &                               & \displaystyle\sum_{x \in A} Z_{n+1,x,y} \otimes \kb{x} \otimes \id_{A'} & \succeq & \kb{a} \otimes \id_{B'_0} \otimes \kb{\psi_a}, \, \forall a \in \zo, y \in B.
\end{array} \]
This shows that $d_2^* \geq P_{\A, 0}^*$. Applying the Subspace lemma (Lemma~\ref{SubspaceLemma}) recursively, as in the case for cheating Bob, we get the following problem with the same optimal objective value
\[ \begin{array}{rrrcllllllllllllll}
\textrm{} & \inf                                & \inner{Z_1}{\kb{\phi}}  \\
                     & \textrm{s.t.} & \bra{\phi_{y_1, \ldots, y_{j-1}}} Z_{j, x_1, y_1, \ldots, x_{j-1}, y_{j-1}} \ket{\phi_{y_1, \ldots, y_{j-1}}} & \geq & \displaystyle\sum_{y_{j} \in B_j} \bra{\phi_{y_1, \ldots, y_{j}}}Z_{j+1, x_1, y_1, \ldots, x_{j}, y_{j}} \ket{\phi_{y_1, \ldots, y_j}}, \\
                                          & & & & \forAll j \in \{ 1, \ldots, n \}, \\
                     & & & & (x_1, \ldots, x_{j}) \in A_1 \times \cdots \times A_{j}, \\
                     & & & & (y_1, \ldots, y_{j}) \in B_1 \times \cdots \times B_{j}, \\
                     &                               & \displaystyle\displaystyle\sum_{x \in A} \bra{\phi_{y}} Z_{n+1,x,y} \ket{\phi_{y}} \, \kb{x} \otimes \id_{A'} & \succeq & \frac{1}{2} \beta_{a,y} \, \kb{\psi_a}, \, \forAll a \in \zo, y \in B.
\end{array} \]
Defining
\[ z_{j+1, x_1, y_1, \ldots, x_j, y_j} := \bra{\phi_{y_1, \ldots, y_j}} Z_{j+1, x_1, y_1, \ldots, x_j, y_j} \ket{\phi_{y_1, \ldots, y_j}}, \]
for $j \in \{ 0, 1, \ldots, n-1 \}$, $(x_1, \ldots, x_j) \in A_1 \times \cdots \times A_j$, and $(y_1, \ldots, y_{j}) \in B_1 \times \cdots \times B_j$,
and
\[ \Diag(z_{n+1}^{(y)}) := \sum_{x \in A} \bra{\phi_y} Z_{n+1, x,y} \ket{\phi_y} \, \kb{x}, \]
for $y \in B$,
we get the following equivalent problem 
\[ \begin{array}{rrrcllllllllllllll}
\textrm{} & d_2^* \; = \; \inf                                & z_1  \\
                     & \textrm{subject to} & z_{j, x_1, y_1, \ldots, x_{j-1}, y_{j-1}} & \geq & \dsum_{y_{j} \in B_{j}} z_{j+1, x_1, y_1, \ldots, x_{j}, y_{j}}, \\
                                          & & & & \forAll j \in \{ 1, \ldots, n \}, \\
                     & & & & (x_1, \ldots, x_{j}) \in A_1 \times \cdots \times A_{j}, \\
                     & & & & (y_1, \ldots, y_{j}) \in B_1 \times \cdots \times B_{j}, \\
                     &                               & \Diag(z_{n+1}^{(y)}) & \succeq & \half \beta_{a,y} \, \sqrtt{\alpha_a}, \, \forAll y \in B, a \in \zo, 
\end{array} \] 
{noting $z_{j+1, x_1, y_1, \ldots, x_j, y_j} = 0$ when $\ket{\phi_{y_1, \ldots, y_j}} = 0$ can be assumed in an optimal solution.}
This problem has a strictly feasible solution and the objective function is bounded from below on the feasible region, thus it and its dual have zero duality gap.
The dual of this problem is
\[ \max_{\substack{(s_1, \ldots, s_n, s) \in \calP_{\A} \\ \sigma_{a,y} \in \pos^A}} \set{\half \sum_{a \in \zo} \sum_{y \in B} \beta_{a,y} \inner{\sigma_{a,y}}{\sqrtt{\alpha_a}} : \diag(\sigma_{a,y}) = s^{(a,y)}, \, \forall a \in \zo, y \in B} \]
which is equivalent to Alice's reduced problem for forcing outcome $0$ by Lemma~\ref{FidelityLemma} and has optimal objective value $d_2^*$. Therefore, we have $P_{\A, 0}^* \leq d_2^* \leq P_{\A, 0}^*$, as desired.

The case for forcing outcome $1$ is almost the same, except every occurrence of $\alpha_a$ is replaced with $\alpha_{\bar{a}}$. 
The above SDP thus becomes
\[ \max_{\substack{(s_1, \ldots, s_n, s) \in \calP_{\A} \\ \sigma_{a,y} \in \pos^A}} \set{\half \sum_{a \in \zo} \sum_{y \in B} \beta_{a,y} \inner{\sigma_{a,y}}{\sqrtt{\alpha_{\bar{a}}}} : \diag(\sigma_{a,y}) = s^{(a,y)}, \, \forall a \in \zo, y \in B}. \]
Since the last constraint is symmetric in $a$, we can replace $s^{(a,y)}$ with $s^{(\bar{a},y)}$ and $\sigma_{a,y}$ with $\sigma_{\bar{a},y}$ and the optimal objective value does not change. We can write it as
\[ \max_{\substack{(s_1, \ldots, s_n, s) \in \calP_{\A} \\ \sigma_{a,y} \in \pos^A}} \set{\half \sum_{a \in \zo} \sum_{y \in B} \beta_{\bar{a},y} \inner{\sigma_{a,y}}{\sqrtt{\alpha_{a}}} : \diag(\sigma_{a,y}) = s^{(a,y)}, \, \forall a \in \zo, y \in B}, \]
which is equivalent to Alice's reduced problem {for forcing outcome $1$} by Lemma~\ref{FidelityLemma} and the rest of the argument follows similarly as {in} the case of Alice forcing outcome $0$. \qed 
     
\end{document}